\documentclass[acmsmall]{acmart}
\AtBeginDocument{%
  \providecommand\BibTeX{{%
    \normalfont B\kern-0.5em{\scshape i\kern-0.25em b}\kern-0.8em\TeX}}}

\setcopyright{acmlicensed}
\acmJournal{PACMMOD}
\acmYear{2023} \acmVolume{1} \acmNumber{2} \acmArticle{187} \acmMonth{6} \acmPrice{15.00}\acmDOI{10.1145/3589332}

\usepackage{subcaption}
\usepackage{hyperref}
\usepackage{newtxmath}
\usepackage{enumitem}
\usepackage{amsmath}
\usepackage[linesnumbered,ruled,vlined,noend]{algorithm2e}
\usepackage{tabularx}
\usepackage{color}
\usepackage{graphicx}
\usepackage{textcomp}
\usepackage{color}
\usepackage{caption}
\definecolor{edit}{rgb}{0,0,0}

\newcommand{\idxname}{WISK\xspace}
\newtheorem{prob}{Problem}
\newtheorem{definition}{Definition}
\newcommand{\eat}[1]{}

\begin{document}

\title{WISK: A Workload-aware Learned Index for Spatial Keyword Queries}

\author{Yufan Sheng}
\affiliation{
    \institution{University of New South Wales}
    \city{Sydney}
    \country{Australia}
}
\email{yufan.sheng@unsw.edu.au}

\author{Xin Cao}
\authornote{Xin Cao is the corresponding author.}
\affiliation{
    \institution{University of New South Wales}
    \city{Sydney}
    \country{Australia}
}
\email{xin.cao@unsw.edu.au}

\author{Yixiang Fang}
\affiliation{
    \institution{The Chinese University of Hong Kong, Shenzhen}
    \city{Shenzhen}
    \country{China}
}
\email{fangyixiang@cuhk.edu.cn}

\author{Kaiqi Zhao}
\affiliation{
    \institution{The University of Auckland}
    \city{Auckland}
    \country{New Zealand}
}
\email{kaiqi.zhao@auckland.ac.nz}

\author{Jianzhong Qi}
\affiliation{
    \institution{The University of Melbourne}
    \city{Melbourne}
    \country{Australia}
}
\email{jianzhong.qi@unimelb.edu.au}

\author{Gao Cong}
\affiliation{
    \institution{Nanyang Technological University}
    \city{Singapore}
    \country{Singapore}
}
\email{gaocong@ntu.edu.sg}

\author{Wenjie Zhang}
\affiliation{
    \institution{University of New South Wales}
    \city{Sydney}
    \country{Australia}
}
\email{wenjie.zhang@unsw.edu.au}

\renewcommand{\shortauthors}{Yufan Sheng et al.}
\received{October 2022}
\received[revised]{January 2023}
\received[accepted]{February 2023}

\begin{abstract}
    Spatial objects often come with textual information, such as Points of Interest (POIs) with their descriptions, which are referred to as geo-textual data. To retrieve such data, spatial keyword queries that take into account both spatial proximity and textual relevance have been extensively studied. 
Existing indexes designed for spatial keyword queries are mostly built based on the geo-textual data 
without considering the distribution of queries already received. 
However, previous studies have shown that utilizing the known query distribution can 
improve the index structure for future query processing.
%
%
In this paper, we propose \idxname, a learned index for spatial keyword queries, which self-adapts for optimizing querying costs given a query workload. One key challenge is how to utilize both structured spatial attributes and unstructured textual information during learning the index. 
We first divide the data objects into partitions, aiming to minimize the processing costs of the given query workload. We prove the NP-hardness of the partitioning problem and propose a machine learning model to find the optimal partitions.
%
Then, to achieve more pruning power, we build a hierarchical structure based on the generated partitions in a bottom-up manner with a reinforcement learning-based approach. We conduct extensive experiments on real-world datasets and query workloads with various distributions, and the results show that \idxname outperforms all competitors, achieving up to 8$\times$ 
speedup in querying time with comparable storage overhead.
\end{abstract}

\maketitle

\section{Introduction}
\label{sec1}

The worldwide mobile internet use has already surpassed desktop use since late 2016\footnote{Mobile vs. Desktop Internet Usage: \href{https://www.broadbandsearch.net/blog/mobile-desktop-internet-usage-statistics/}{https://www.broadbandsearch.net/blog/mobile-desktop-internet-usage-statistics/}}. In such a mobile internet, massive \textbf{geo-textual data} with spatial and textual attributes are generated, e.g., Points of Interest (POIs) in Google Maps are associated with geo-locations and descriptive texts. 
Managing and retrieving geo-textual data at scale has attracted much attention. 
In recent years, the \textbf{spatial keyword queries}~\cite{DBLP:conf/sigmod/CaoCJO11, DBLP:conf/icde/ZhangZZL13, DBLP:conf/cikm/ChristoforakiHDMS11, DBLP:conf/ssd/TampakisSDPKV21, DBLP:journals/pvldb/CongJW09, DBLP:conf/ssdbm/HariharanHLM07, DBLP:conf/icde/FelipeHR08, DBLP:journals/vldb/ChenCCJ21, fang2018spatial} have been extensively studied, which take a location and a set of keywords as arguments and return objects based on different definitions of spatial proximity and textual relevance. \textbf{Spatial keyword indexes} are developed to process such queries efficiently by incorporating the techniques of indexing spatial objects and documents. 

\begin{figure}[tb]
   \begin{minipage}{\linewidth}
       \centering
       \includegraphics[width=.55\linewidth]{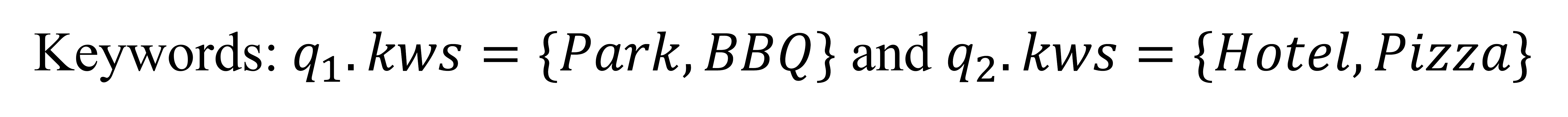}
   \end{minipage}
   \begin{minipage}{\linewidth}
       \centering
       \subcaptionbox{Data-driven index\label{data-driven}}{
            \vspace{-0.2cm}
            \includegraphics[width=.25\linewidth]{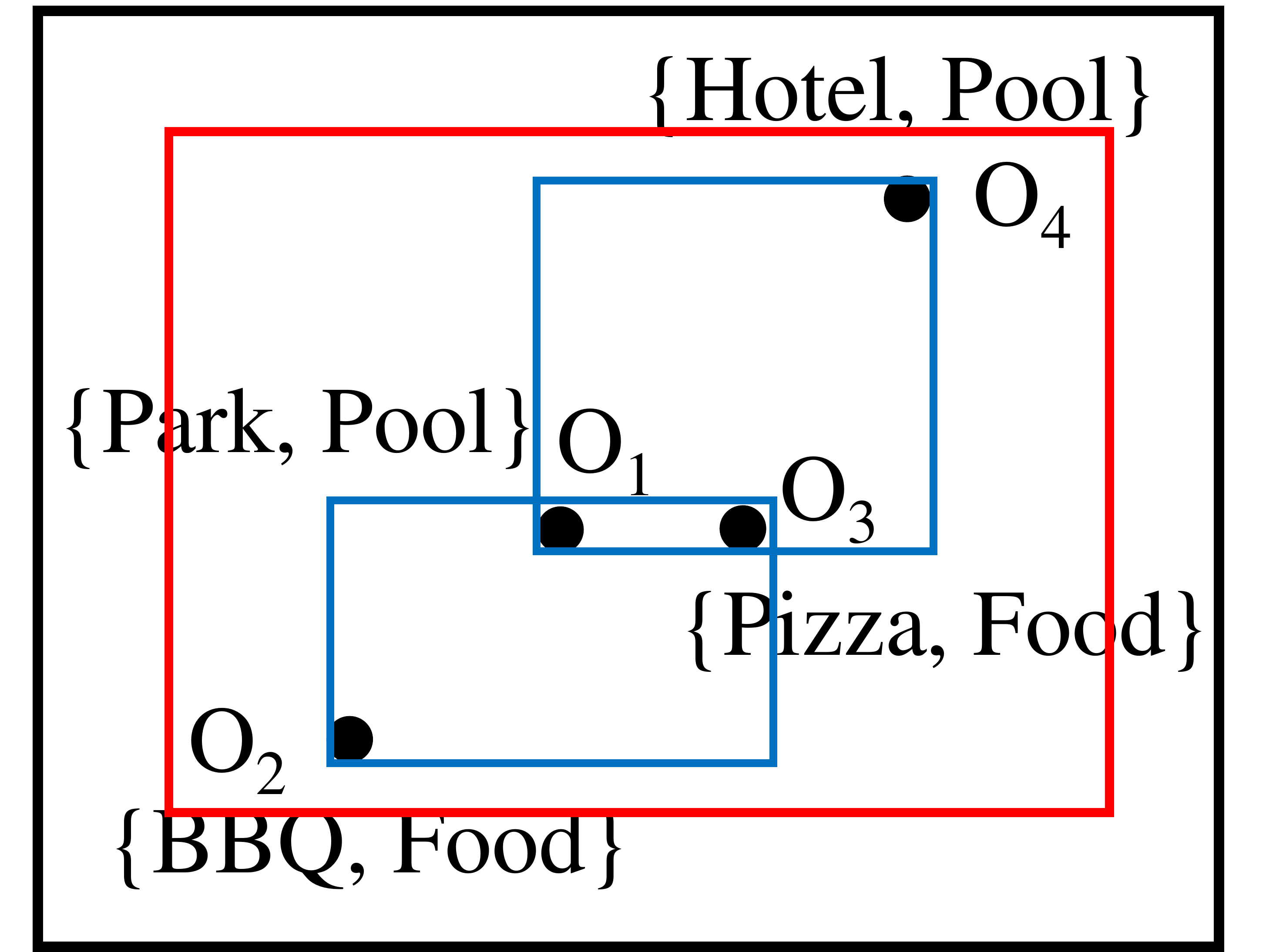}
       }
       \subcaptionbox{Query-aware index\label{query-aware}}{
            \vspace{-0.2cm}
            \includegraphics[width=.25\linewidth]{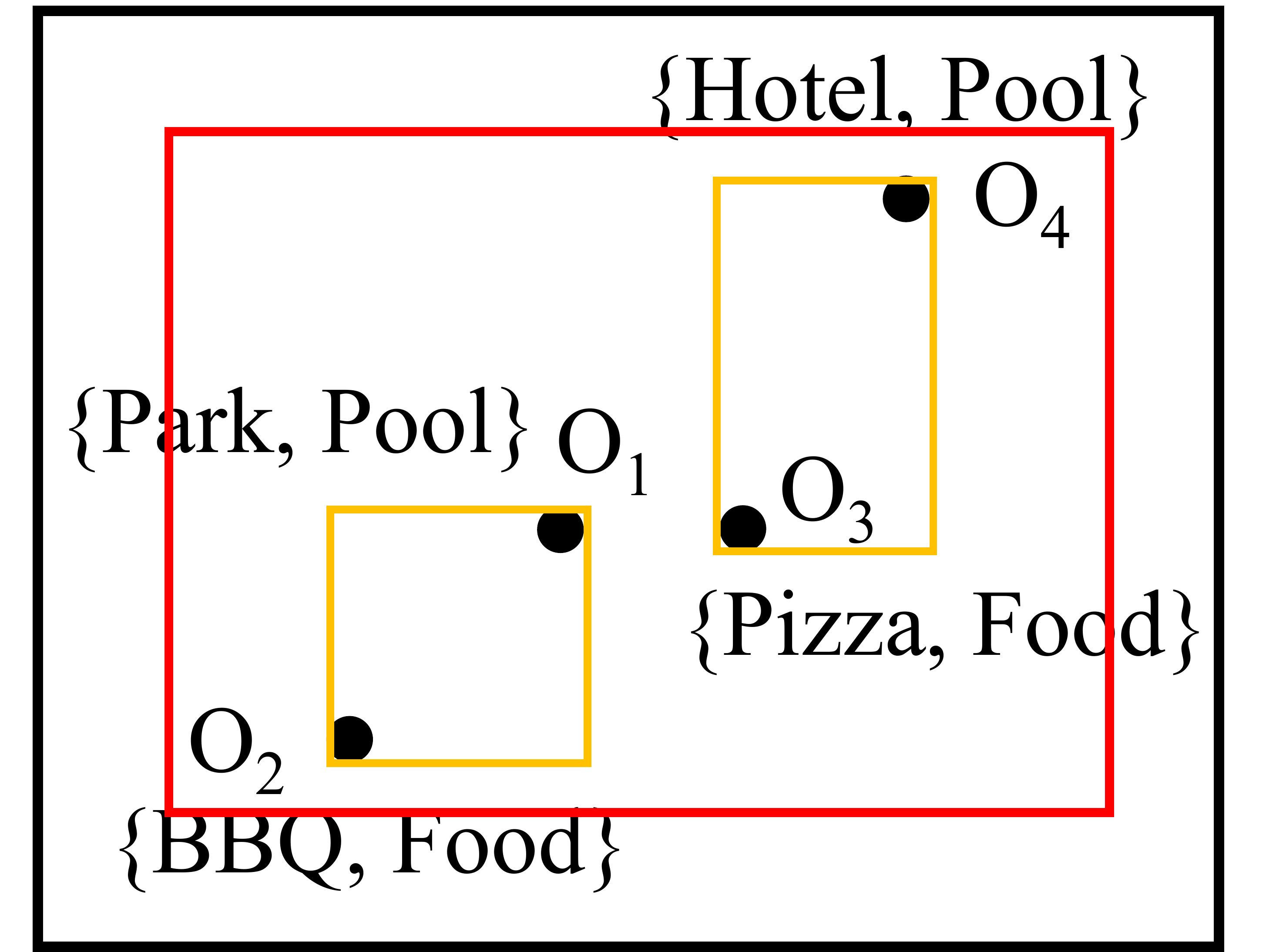}
       }
   \end{minipage}
  \vspace{-0.3cm}
   \caption{Data-driven indexes vs query-aware indexes}
   \label{fig:ex1}
\end{figure}


However, state-of-the-art spatial keyword indexes still have several drawbacks. First, as shown in previous studies~\cite{DBLP:journals/pvldb/ChenCJW13, DBLP:series/synthesis/2019Mahmood, DBLP:journals/vldb/ChenCCJ21}, no existing index can work efficiently for all data distributions, and there is no single approach that dominates all others. Second, \textcolor{edit}{traditional indexes may have some parameters to be set with fixed values across the entire input data space. For example, CDIR-tree~\cite{DBLP:journals/pvldb/CongJW09} needs a parameter to balance the importance of spatial proximity and textual relevancy. Since the query and data distributions vary in different areas, it is hard to select a single parameter value that fits all distributions.}~
Third, no indexes have considered utilizing the query workload.  Previous works \cite{DBLP:conf/sigmod/BrunoCG01, DBLP:conf/sigir/SkobeltsynLZRA07} have shown that certain data regions could be much more heavily queried. For spatial keyword queries, both the query keyword distribution and query location distribution can be various over different spatial regions. Thus, utilizing the known query distribution could further optimize the index structure for future query processing~\cite{DBLP:conf/sigmod/NathanDAK20,DBLP:journals/pvldb/DingNAK20}. For example, as shown in Figure~\ref{fig:ex1}, a spatial keyword query workload includes $q_1$ and $q_2$ with the same query region (red rectangle) and different keywords. Existing data-driven indexes (Figure~\ref{data-driven}) store objects with close spatial distances and large textual similarities into a partition (enclosed by blue rectangles), and both queries have to check two partitions containing four objects. If the four objects are grouped in a different way to adapt for the query keywords, both queries only need to check two objects in a single partition. Query-aware index can be expected to achieve better performance on future queries following similar known query distributions.

Motivated by these observations, in this work, we propose a novel \underline{W}orkload-aware learned \underline{I}ndex for \underline{S}patial \underline{K}eyword queries (\idxname). The objective is to learn an index structure using both the spatial and textual information such that the processing cost for the known query workload using this index is minimized. We 
focus on the spatial keyword range query workload.

There exist some learned indexes for spatial query processing, which can also be classified into \emph{data-driven} indexes (such as ZM~\cite{DBLP:conf/mdm/WangFX019}, LISA~\cite{DBLP:conf/sigmod/Li0ZY020}, and RSMI~\cite{ DBLP:journals/pvldb/QiLJK20}) and \emph{query-aware} indexes (such as Flood~\cite{DBLP:conf/sigmod/NathanDAK20} and Tsunami~\cite{ DBLP:journals/pvldb/DingNAK20}). 
%
These learned spatial index structures are not suitable for processing spatial keyword queries directly, because they only use spatial attributes for index learning. 
\textcolor{edit}{Concurrent with our work, a learned index has been proposed for spatial keyword querying~\cite{ding2022learned}. It uses spatial attributes for index learning first and then creates the textual index, and thus it cannot well learn the spatial and textual correlations for building the index.}

The key challenge of learning a spatial keyword index is how to capture the data and query distribution considering both the structured spatial information and unstructured textual attributes during index learning, and make use of the distribution captured to partition the data objects such that more irrelevant partitions can be filtered out during query processing, thus improving the querying efficiency. 
One simple method is to learn a spatial index first, which does not consider keywords, and then build an inverted file to manage the textual information within each partition of the index. Our experiments show that this has poor performance. The reason could be briefly explained in Figure~\ref{fig:ex2}, with a dataset containing 8 objects and a workload with 3 queries. The learned spatial indexes put fewer objects into the left partition which contains more queries to reduce costs. However, such partitions are worse than that in Figure~\ref{use-keyword}, when objects without query keywords can be ignored during query processing using inverted files. Figure~\ref{spatial-only} has 3 query-relevant objects (\texttt{Store}) in the left partition and 4 query-relevant objects (\texttt{Food}) in the right partition, and thus the number of checks required is $2 \times 3 + 1 \times 4 = 10$. But Figure~\ref{use-keyword} only needs $2 \times 3 + 1 \times 2 = 8$ checks. 

\begin{figure}[tb]
    \centering
    \subcaptionbox{Learned spatial index\label{spatial-only}}{
        \vspace{-0.2cm}
        \includegraphics[width=.25\linewidth]{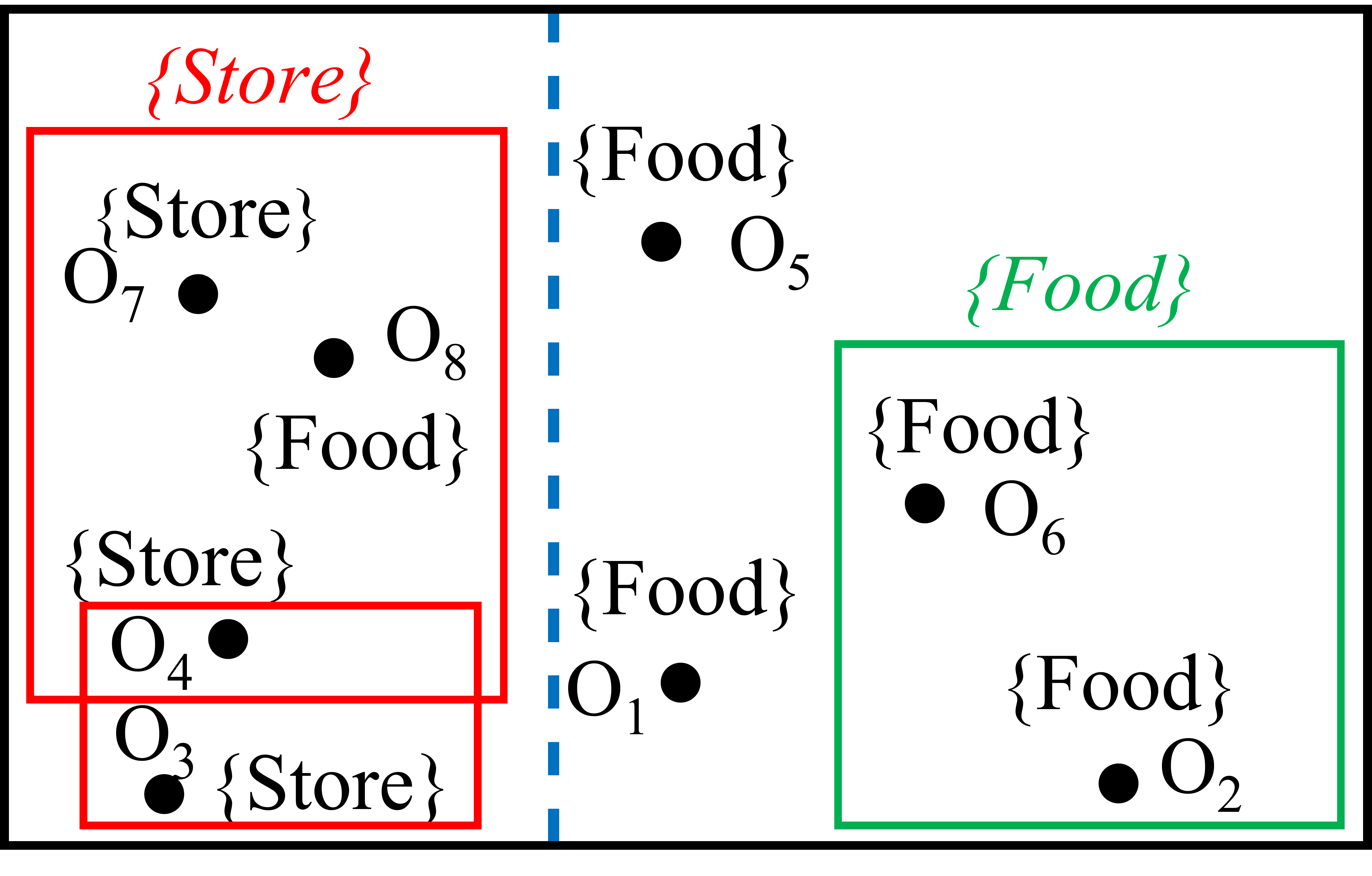}
    }
    \subcaptionbox{WISK\label{use-keyword}}{
        \vspace{-0.2cm}
        \includegraphics[width=.25\linewidth]{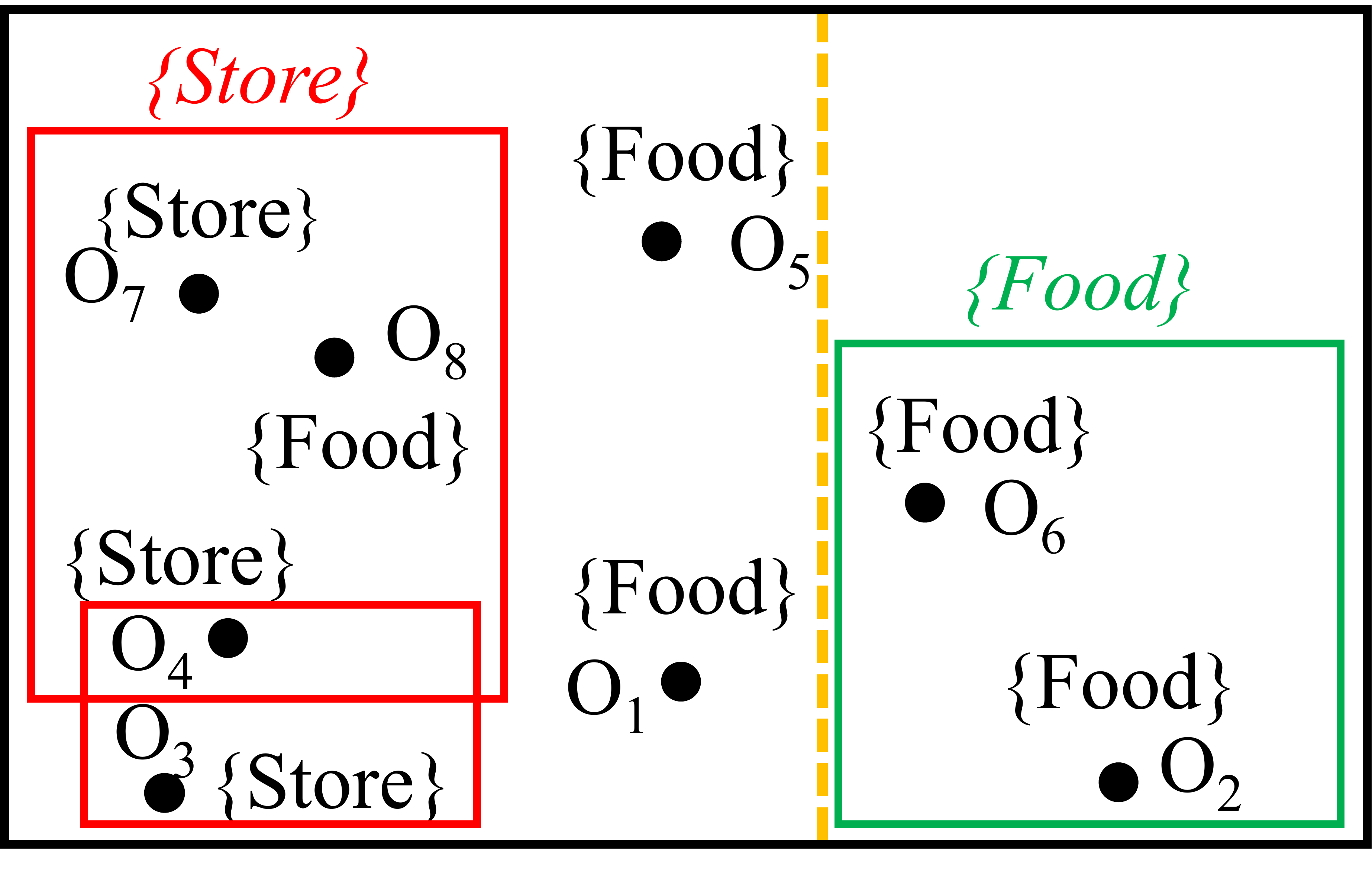}

    }
    \vspace{-0.3cm}
    \caption{Learning without and with textual information}
    \label{fig:ex2}
\end{figure}


%

It can be observed that the query cost largely depends on how the objects are partitioned.
We first formulate a problem of finding $k$ partitions with the optimal cost for a given query workload. We show the NP-hardness of this problem by a reduction from the MaxSkip partitioning problem~\cite{DBLP:conf/sigmod/SunFKX14, DBLP:conf/sigmod/YangCWGLMLKA20}. To learn good partitions, we design a cost estimation method considering both spatial and textual information of the query and data, based on 
trained models that can approximate the Cumulative Distribution Function (CDF) of geo-textual data.
%
%
Then, we propose a heuristic algorithm and use Stochastic Gradient Descent (SGD) \cite{bottou-98x} to generate the $k$ partitions.

%
If there is only one-level of partitions, we need to check many partitions irrelevant to the query.
%
Hence, we further group the partitions into a hierarchy as do most indexes.
A simple way is to adopt the method in CDIR-tree~\cite{DBLP:journals/pvldb/CongJW09} for building the tree in a bottom-up manner. However, it is hard to select 
the weights of the spatial proximity and textual relevance between partitions. We show that it might even have worse query time in experiments.
Instead, we propose to pack the nodes level by level, and model the one-level node packing problem as a sequential decision-making process.
%
In particular, we develop a reinforcement learning~\cite{DBLP:journals/jair/KaelblingLM96} algorithm to find the optimal packing for each level and build the index in a bottom-up manner by considering the query workload.


In summary, we make the following main contributions:
\begin{itemize}
    \item  We propose a query-aware learned index named \idxname considering spatial and textual attributes simultaneously.

    \item To generate the leaf nodes of \idxname, we define an optimal partitioning problem and show its NP-hardness. We  propose a heuristic algorithm to solve the problem using machine learning techniques. 
    
    \item To build the hierarchy of \idxname, we propose to pack the nodes level by level in a bottom-up manner, and we treat the node packing as a sequential decision-making process. We develop a solution based on reinforcement learning.
    
    \item We 
    perform a comprehensive empirical study using real-world datasets and synthetic query workloads with various distributions. The results show that \idxname\ outperforms the state-of-the-art spatial keyword indexes consistently in terms of efficiency, achieving up to 8$\times$ times speedup while having a comparable index size.
\end{itemize}


\section{Preliminaries}
\label{sec2}

\subsection{Problem Statement}
We consider a geo-textual dataset $D$ where each data object, i.e., a geo-textual object $o\in D$, has a point location denoted as $o.loc$ and a text description denoted as $o.kws$, which is available from a range of source \cite{cong2016querying}. For ease of discussion, we assume two-dimensional coordinates in Euclidean space to represent $o.loc$, although our proposed techniques can generalize to multi-dimensional spaces easily. The text description $o.kws$ is represented as a set of keywords, e.g., tags indicating the functionality of a POI.
We aim to process \emph{spatial keyword range queries} over $D$.



\begin{definition}[\textbf{Spatial Keyword Range (SKR) Query}]
    An SKR query \textit{q} is represented by a pair $(q.area, q.keys)$ where \textit{q.area} and \textit{q.keys} denote a  spatial region and a set of keywords, respectively. The result of \textit{q}, $q(D) = \{o \in D \mid o.loc \ in \ q.area, o.kws \cap q.keys \neq \varnothing\}$, is a subset of \textit{D} that includes all objects within the query region containing at least one query keyword.
\end{definition}

Here, we use a rectangular query region. 
Our techniques can be easily extended to handle other shapes (e.g., circles) by an extra filtering after querying with the bounding rectangle. 


\smallskip\noindent\textbf{Problem.} Our goal is to learn an index structure that can efficiently process SKR queries utilizing the distributions of the geo-textual data and the given query workload.

\subsection{Reinforcement Learning}
\label{sec2.3}
Reinforcement learning (RL) \cite{DBLP:journals/jair/KaelblingLM96} is a machine learning technique where an agent learns from feedback obtained from trial-and-error interactions with an  environment. It has been shown to be effective for sequential decision-making problems \cite{DBLP:journals/nature/SilverHMGSDSAPL16, DBLP:conf/aaai/ZhaoS0Y021}.

RL formulation is based on the Markov Decision Process (MDP) \cite{DBLP:journals/siamrev/Feinberg96}. An MDP has four components: a set of states \textit{S}, a set of actions \textit{A}, transition probabilities $P$, and  rewards $R$. 
At some state $s \in S$, an agent may 
take an action $a \in A$. As a result, there is a probability $P_{a}(s, s^{\prime})$ that the agent transits to state $s^{\prime}$, and a reward $R_{a}(s, s^{\prime})$ is received from such an action and state transition. The goal of the agent is to learn a policy function $\pi : A \times S \rightarrow [0, 1]$, i.e., the probability of taking action $a$ at state $s$, such that the cumulative reward of state transitions is maximized.
Figure \ref{fig:basic_rl} shows the basic workflow of RL. 
The environment connects to the agent via perception and action and it offers the agent the possible action choices based on the current state of the agent. 
The agent learns its policy based on rewards accumulated from interactions with the environment. Its learning process stops when a terminal state is reached.

\begin{figure}[htb]
    \centering
    \includegraphics[width=.55\columnwidth]{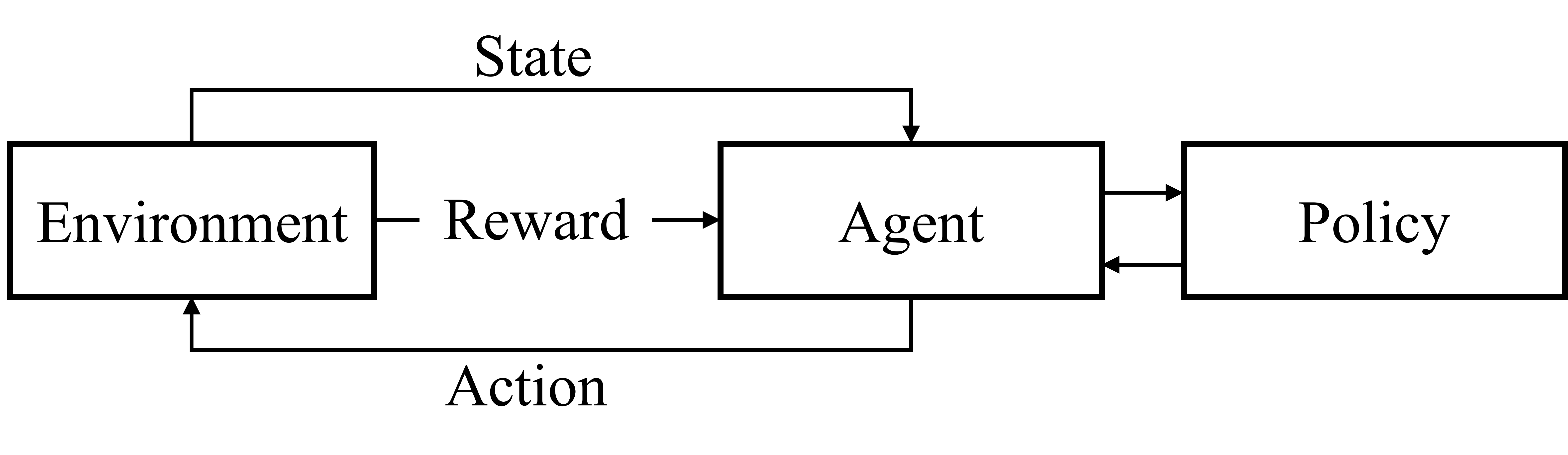}
    \vspace{-0.3cm}
    \caption{The typical RL learning framework}
    \label{fig:basic_rl}
\end{figure}


Q-learning \cite{DBLP:journals/ml/WatkinsD92} is a commonly used value-based policy learning algorithm, which learns the value of an action given a state. 
It learns a policy that maximizes the value of a so-called  Q-function, $Q(s, a)$, i.e., the overall expected reward when an agent plays following the policy \cite{melo2001convergence}. State-of-the-art RL models such as Deep-Q-Network (DQN)~\cite{DBLP:journals/nature/MnihKSRVBGRFOPB15} use a deep neural network $Q(s, a;\theta)$ with parameters $\theta$ to estimate the value of the Q-function $Q(s, a)$. Once $Q(s, a;\theta)$ is trained, it can be used for decision-making for future events.



\section{Index Overview}
\label{sec3}
\idxname consists of two parts: (1) learn an optimal data layout for the given query workload, and (2) create an index based on that layout. 

Query processing on an index that partitions objects into clusters typically involves two costly operations: filtering and verification. Irrelevant partitions are filtered out, and objects in the remaining partitions are verified. 
Thus, in \idxname we first aim to learn an optimal partition of the geo-textual objects, such that for the given SKR query workload we can achieve the minimum query processing costs computed using both the filtering and verification costs. However, given a large dataset, the number of possible partitions of the objects is extremely huge, and it is hard to learn an optimal partition. We propose to simplify the problem: we divide the 2D space into disjoint partitions to obtain an optimal spatial layout. 
We will give the detail in Section~\ref{sec4}.

If there is only one layer of the index, we need to check all partitions to see if they are relevant to the query, leading to a high filtering cost. We can organize the partitions obtained in the first step into a tree structure to build the final index such that the query processing cost can be further optimized. This is a type of combinatorial optimization task. 
We propose to pack the partitions level by level, and view the one-level packing problem as a sequential decision-making process, which can solved by reinforcement learning. We will present the detail of this step in Section~\ref{sec5}.


The framework of \idxname is shown in Figure~\ref{fig:framework}. A leaf node contains a number of objects, a minimum bounding rectangle (MBR) of the objects, and an inverted file to index the objects in this node. A non-leaf node contains pointers to its child nodes, an MBR for all child nodes, and a bitmap to index keywords appearing in its sub-tree. Here, the bitmap is used due to its much smaller size. 

\begin{figure}[htb]
  \centering
  \includegraphics[width=.5\linewidth]{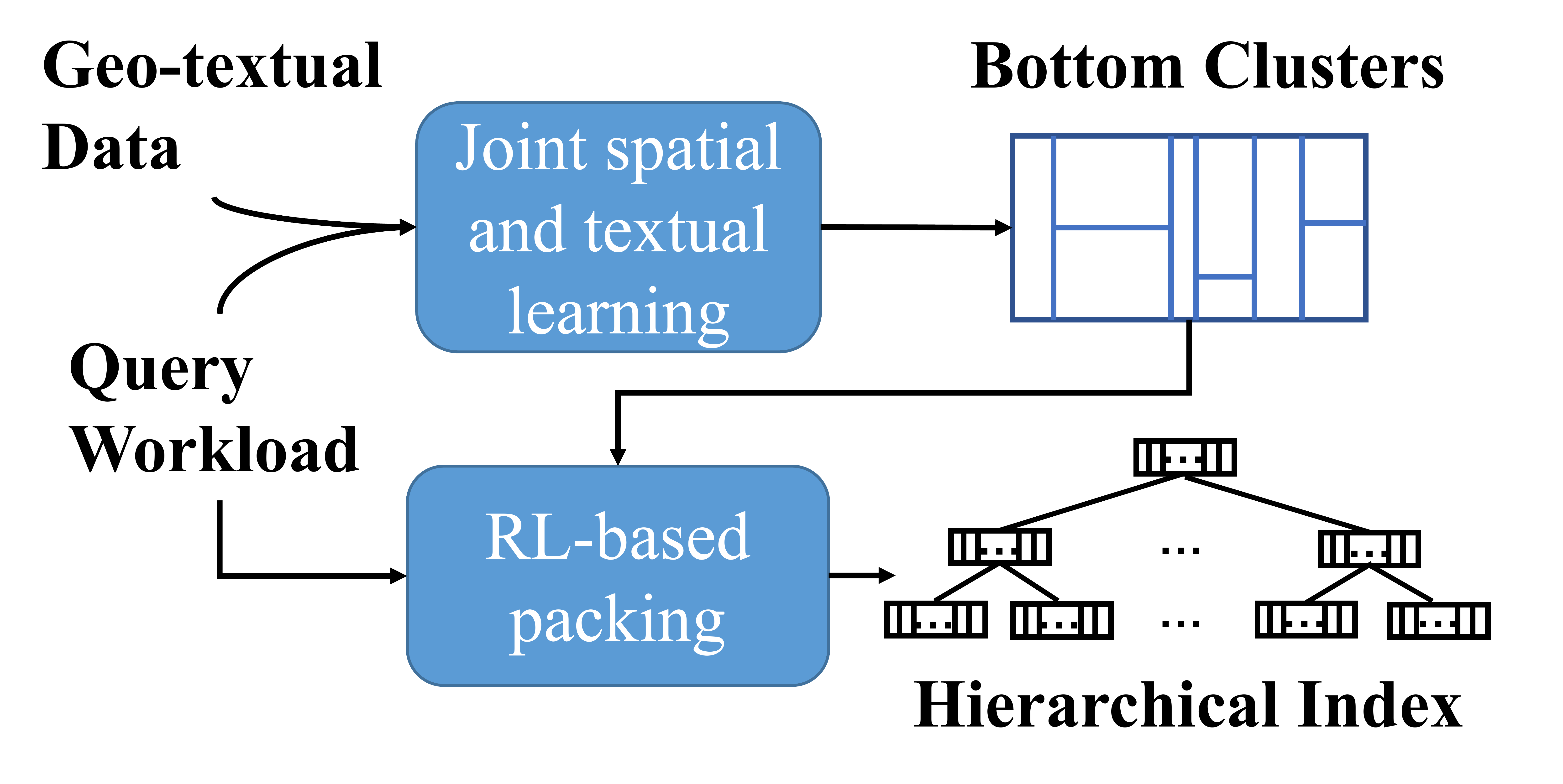}
  \vspace{-0.3cm}
  \caption{The \idxname\ framework}
  \label{fig:framework}
\end{figure}
\noindent\textbf{Index construction.}
Algorithm \ref{alg:gen} summarizes the construction process of \idxname, which consists of two main steps.
Step 1 (lines~1 and~2) is to construct the bottom clusters. A high-quality partition of bottom clusters should result in a low query cost given a query workload.
We train machine learning models to approximate the Cumulative Distribution Function (CDF) of geo-textual objects. 
Then, we define the cost estimation function based on the learned CDFs and make it differentiable, such that we can use stochastic gradient descent (SGD) to learn the optimal partitions. 
%
Step 2 (lines~3 and~4) is to construct the hierarchy of \idxname\ by a bottom-up packing of the bottom clusters. Our goal is to minimize the filtering cost when an SKQ query is processed by this hierarchy. 
We model the packing process in one level as an MDP and apply RL to solve the problem. Based on the given query workload, we design a reward to measure the reduction of the filtering cost given a packing decision, and we train a model to predict the reward, which can be used to guide packing the bottom clusters level by level. The RL approach can also be used to group objects into bottom clusters. However, due to a large number of objects, this will require huge amounts of states during the RL procedure, and both the training time and the performance would be unacceptable \cite{uther1998tree}.
\begin{algorithm}[tb]
\caption{\idxname\ Construction}
\label{alg:gen}
\footnotesize
\LinesNumbered
\KwIn{\textit{Q}, the workload; \textit{D}, the dataset; \textit{C}, the cost function}
\KwOut{\textit{I}, a learned \idxname index}

\textit{KM} $\leftarrow$ MLModel(\textit{D}) \tcc*{Learn CDF models of objects} 

\textit{G} $\leftarrow$ SGDPartition(\textit{Q}, \textit{C}, \textit{KM}) \tcc*{\textit{G} minimizes \textit{C}}

\textit{RM} $\leftarrow$ RLTrain(\textit{G}, \textit{Q}) \tcc*{Learn a model \textit{RM} based on \textit{G}, \textit{Q}}

\textit{I} $\leftarrow$ Packing(\textit{RM}, \textit{G}) \tcc*{Group \textit{G} by using \textit{RM}}
\Return{\textit{I}}\;
\end{algorithm}

\noindent\textbf{Query processing.}
Our query algorithm is similar to those using traditional spatial keyword indexes. 
Given a SKR query $q$, it traverses \idxname\ in a breath-first manner starting from the root. A queue $Q$ is used to manage the nodes visited. For every non-leaf node visited, only the child nodes whose MBRs overlap with $q.area$ and contain some query keywords are added to $Q$ for future verification. When a leaf node (i.e., the bottom cluster) is reached, we use its inverted file to fetch the query-relevant objects and return those in the query region.

\section{Partitioning Optimization}
\label{sec4}
A core problem in \idxname\ construction is to form partitions (i.e., bottom clusters) for cost minimization over the query workload.
In this section, we model the query cost, define an optimal partition problem, show the NP-hardness of the problem, and present a heuristic algorithm for the problem. 

\subsection{Cost Model}
We model the time cost $C(q_i)$ to process an SKR query $q_i$ over a set of bottom clusters $G$ as a linear combination of (1) the cost to scan all bottom clusters to find a subset $G_i \subset G$ that overlap with $q_i.area$ and containing at least one keyword in $q_i.kws$, and (2) the cost to examine the inverted file in each cluster $c \in G_i$ and find the objects in $q_i.area$ and contain at least one query keyword. Eq.~\ref{equ:cost} formalizes the cost, where $|G|$ denotes the total number of clusters, and $\sum_{c \in G_i}|O_c|$ denotes the number of objects in $G_i$ that contains at least one query keyword. 
\textcolor{edit}{In particular, $w_1$ measures the time cost for checking (1) if the MBR of a cluster intersects with the query region, and (2) if the cluster contains some query keywords by scanning the textual index of the cluster. Both checks are independent of the cluster size. Meanwhile, $w_2$ measures the time cost to perform the same checks but at the object level. Following recent studies~\cite{DBLP:conf/ssd/ZhangRLZ21, DBLP:journals/pvldb/DingNAK20}, we use fixed values for these parameters.}
\begin{equation}
    C(q_i) = w_1|G| + w_2\sum\nolimits_{c \in G_i}|O_c|
    \label{equ:cost}
\end{equation}


\noindent \textbf{Example 4.1:\label{exa:4.1}} Figure \ref{fig:trade-off} illustrates this cost function. Suppose that the red and green points represent objects that contain  keywords $k_1$ and $k_2$ respectively. There are two queries, and $q_1.kws=\{k_1\}$ and $q_2.kws=\{k_2\}$. 
If all objects are in a cluster (i.e., no partitioning,  Figure~\ref{trade-off1}), according to Eq.~\ref{equ:cost},  the two queries incur a cost of 
$2(w_1 + 4w_2) = 2w_1 + 8w_2$. This is because there is only one cluster, and each query needs to check four objects containing the query keywords (i.e., four red points for $k_1$ and four green points for $k_2$).   
\begin{figure}[b]
    \centering
    \subcaptionbox{All objects in a cluster\label{trade-off1}}{
        \vspace{-0.2cm}
        \includegraphics[width=.25\columnwidth]{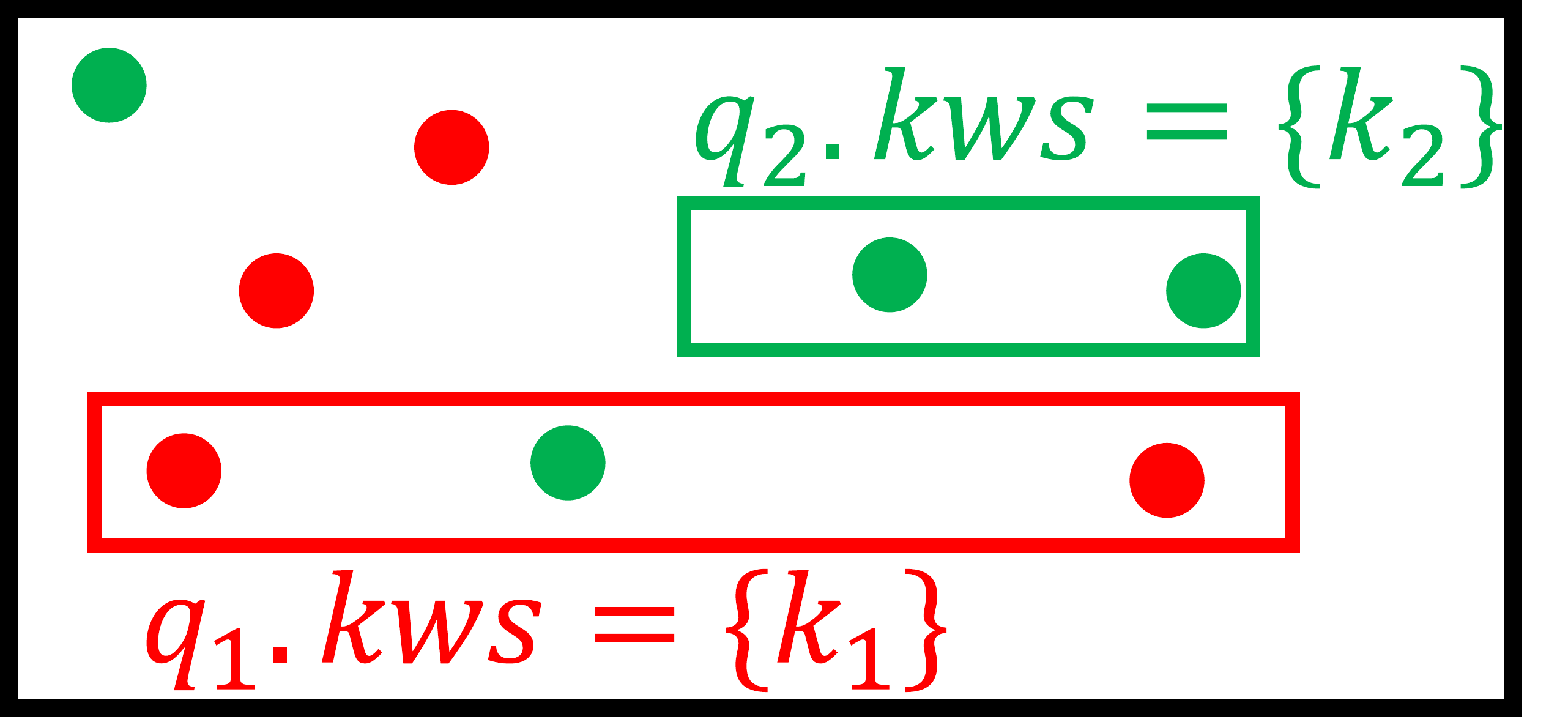}
    }
    \subcaptionbox{Objects in two clusters\label{trade-off2}}{
        \vspace{-0.2cm}
        \includegraphics[width=.25\columnwidth]{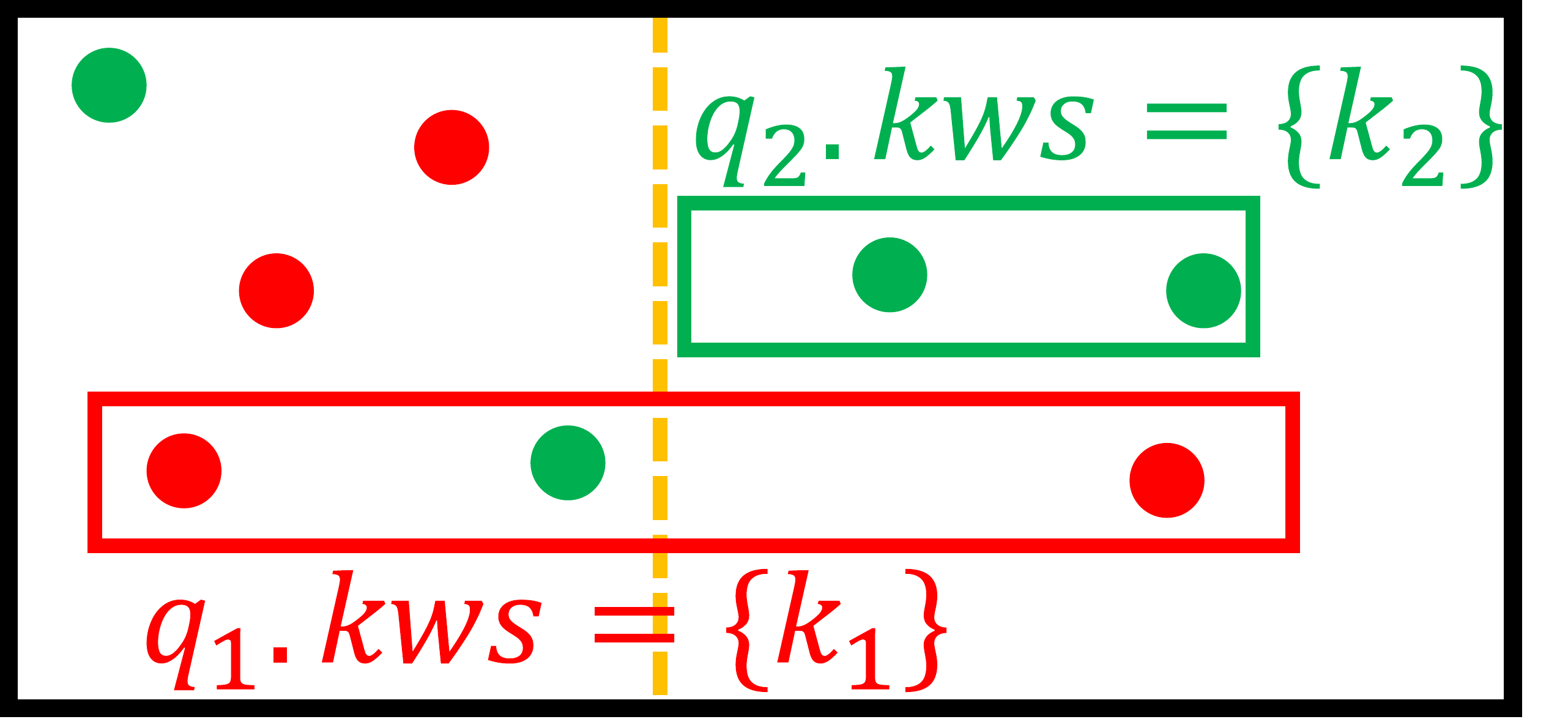}
    }
    \vspace{-0.3cm}
    \caption{Partitioning the space increases the number of clusters, which leads to a larger cluster scanning cost, but also potentially a lower object scanning cost.}
    \label{fig:trade-off}
\end{figure}
If the space is split forming two clusters of five and three points each (Figure~\ref{trade-off2}), the cost of $q_{2}$ and $q_{1}$ will become $2w_1 + 2w_2$  (checking two clusters and two green points) and $2w_1 + 4w_2$ (checking two clusters and four red points), which sum up to $4w_1 + 6w_2$. The partitioning may lead to an overall lower query cost if $w_2$ dominates the cost. 




\subsection{The Optimal Partitioning Problem}
We formulate an optimal partition problem to find a set of clusters that minimize the query cost over a given set of queries.

\begin{prob}[\textbf{Optimal Partitioning}]
    Given a dataset D = $\{o_1, o_1, \ldots,o_n\}$ and a query workload W = $\{q_1, q_2, \ldots, q_m\}$, we aim to find an optimal partition, i.e., a set of \textit{k} clusters $G = \{c_1, c_2, \ldots, c_k\}$ where (1) each object belongs to exactly one cluster, i.e., $\bigcup_{c_i \in G} c_i = D$, and $\forall_{c_i, c_j \in G}\, c_i \cap c_j = \varnothing$, and (2) the total cost, $\sum_{q_i \in W}C(q_i)$, is minimized, where $C(q_i)$ (Eq. \ref{equ:cost}) is the cost of $q_i$.
    \label{prob:partition}
\end{prob}


\subsubsection{\textbf{Problem Analysis.}} 
%
We proceed to show that the optimal partitioning problem is NP-hard by reducing from the MaxSkip partitioning problem, which has been shown to be NP-hard \cite{DBLP:conf/sigmod/SunFKX14, DBLP:conf/sigmod/YangCWGLMLKA20}. 

\begin{theorem}
    Problem~\ref{prob:partition} is NP-hard.
\end{theorem}

\begin{proof}
    We first briefly introduce the MaxSkip partitioning problem, which arises from big data analytics systems. Let \textit{Q} be a collection of queries. Consider a set of partitions $P = \{p_1, p_2,..., p_k\}$ where each partition is a collection of tuples, and the size of each partition is larger than a minimum size bound $b$. 
    A big data analytics system can prune a partition $p_i$ if none of the tuples in this partition satisfies a query $q\in Q$ when processing $q$. 
    A cost function (Eq.~\ref{equ:np1}) can thus be defined on each partition, which denotes the number of tuples that can be skipped for processing all queries in $Q$, if such a partition is formed. Here, $|p_{i}|$ denotes the number of tuples in partition $p_i$, and $Q_i$ denotes the set of queries that can be processed without accessing $p_i$.
    \begin{equation}
    \label{equ:np1}
        Cost(p_i) = |Q_i||p_i|, where\ Q_i\subseteq Q
    \end{equation}
   The MaxSkip partitioning problem aims to find the optimal partitions $P_{opt}$ maximizing the total number of tuples that can be skipped when executing \textit{Q}, i.e., $P_{opt} = \arg\max_{P}\sum\nolimits_{p_{i} \in P}Cost(p_{i})$.
    
    
    We map one instance of the MaxSkip partitioning problem to an instance of our optimal partitioning problem as below: for each query $q_w$ in $Q$, we create a keyword $d_w$ and form a SKR query $q=\{q.L, q.d_w\}$ where $q.L$ is the MBR of the entire space. For each tuple $t_m$, we create a geo-textual object $o_m$ such that its location is in $L$, and its keywords correspond to the queries it can satisfy in $Q$. 
    
    Given this mapping, in the MaxSkip partitioning problem, for a partition $p_i\in P$, if a set of queries $Q_i\subseteq Q$ can be skipped when processing $p_i$, we can get a cluster $c_i$ in our problem and a set of SKR queries $R_i\in W$ that are irrelevant to $c_i$ (since no geo-textual objects in $c_i$ contains a keyword in $R_i$). Hence, if we could find an optimal partition that maximize the total number of tuples when running queries in $W$, it is equivalent that we can find an optimal partitioning method that minimizes the cost in our problem. Since the mapping is of linear time, we complete the proof. 
\end{proof}

\subsection{A Heuristic Partition Algorithm}
As pointed out by Christoforaki et al.~\cite{DBLP:conf/cikm/ChristoforakiHDMS11}, a query region is usually much smaller than the data space such that many data objects are not queried by the workload $W$. Hence, to fully utilize the query workload for partitioning, a data based partitioning method is not suitable to solve our problem.
Instead, we employ a space-disjoint partitioning approach and propose a heuristic partition algorithm.

%
Our index aims to learn splitting the spatial data space along different dimensions and coordinate values. 
Our partition algorithm starts by initializing one single partition that covers the full data space (which corresponds to a cluster that contains the full dataset). At this point, each query contributes the same $w_1 + |D| \cdot w_2$ cost to the overall cost of the query workload $W$. Then, we find a split dimension $d_s$ and a split value $v_s$ that yield the largest reduction in the query cost. We use the resulting $d_s$ and $v_s$ to split the data space into two sub-spaces and update the total query cost. For each sub-space, we repeat the splitting process recursively until the total query cost cannot be reduced
or some pre-defined conditions, e.g., a minimum number of queries intersecting with the sub-space, are met. When the algorithm terminates, we use the MBR of the data objects in each resultant sub-space as a bottom cluster in \idxname. 


\subsubsection{\textbf{Learning the Split Dimension and Value}}
\label{sec4.3.1}
A na\"ive method to find the value to make a split uses a brute-force search. Let $V_d$ ($d \in \{x, y\}$) be a sorted list of distinct object coordinate values along dimension $d$ in the current (sub-)space to be partitioned. Except for the first and the last values, every value in $V_d$ can be used to split the space into two sub-spaces. Examining all $|V_d|$ values takes  $\mathcal{O}\big(f\cdot(|V_{x}|+|V_{y}|)\big)$ where $f$ denotes the time cost to split on a value and run queries based on such a splitting
This approach becomes impractical for large datasets with a large value of $|V_x| + |V_{y}|$.

Motivated by the recent success of machine learning in solving complex problems~\cite{DBLP:conf/aaai/PratesALLV19, DBLP:journals/eor/BengioLP21},
we propose a learning-based method to predict the query costs given a split dimension and a split value, such that the optimal split can be approximated by minimizing the predicted query cost with high efficiency. 
At the core of the query cost prediction problem of a split is to 
(1)~predict the number of resultant sub-spaces overlapping with the query, and 
(2)~predict the number of objects that contain any of the query keywords and reside in the resultant sub-spaces. 

To address the first prediction problem, we use the indicator function \cite{kleene1952introduction} to denote whether a sub-space overlaps with the query region. For example, let $[q_{x_{b}}, q_{x_{u}}]$ be the \textit{x} range of query \textit{q} and $p_x$ be a split value along dimension \textit{x}. The indicator functions $\vmathbb{1}(p_x \ge q_{x_{b}})$ and $\vmathbb{1}(p_x < q_{x_{u}})$ are used to decide whether $q$ intersects with the resultant left and right sub-spaces, respectively. If a sub-space has an indicator function value of 1, we need to further predict the number of query result objects within the sub-space. Otherwise, we can ignore the sub-space when computing the query cost. The indicator function is not differentiable, and machine learning methods such as gradient descent cannot be applied to solve a split value optimization problem formulated by such functions. 
As such, we use the sigmoid function~\cite{DBLP:conf/iwann/HanM95}, $\sigma(\beta x)$ with $\beta = 3$, to approximate the indicator function as does in prior work~\cite{DBLP:journals/tnn/ChenTY14, cao2020sigmoidal}, e.g., $\vmathbb{1}(p_x \ge q_{x_b}) = \vmathbb{1}(p_{x} - q_{x_b} \ge 0) \approx \sigma(3(p_x - q_{x_b}))$.

To address the second prediction problem, we follow the idea in recent studies~\cite{DBLP:conf/sigmod/KraskaBCDP18, DBLP:conf/sigmod/NathanDAK20, DBLP:conf/sigmod/0001ZC21} that learn the Cumulative Distribution Function (CDF) to estimate the density of objects in a data space. Our goal is to learn the joint CDF $F_{X, Y}(x, y)$ of two variables \textit{X} and \textit{Y}, corresponding to the spatial coordinates in two dimensions. The learned CDF can quickly estimate the number of objects in a rectangular region, i.e., a sub-space.
\textcolor{edit}{To accelerate the CDF learning, we assume that \textit{X} and \textit{Y} are independent,  following a previous study \cite{DBLP:conf/sigmod/NathanDAK20}. Thus, 
we can decompose the joint CDF into the product of two marginal CDFs, $F_X(x)$ and $F_Y(y)$, as shown in Eq. \ref{equ:jointcdf}.}
\begin{equation}
    F_{X, Y}(x, y) = P(X \le x, Y \le y) = F_X(x)F_Y(y)
    \label{equ:jointcdf}
\end{equation}

For ease of presentation, we use F(x) and F(y) to denote the marginal CDFs of \textit{X} and \textit{Y} in the rest of the paper, respectively.

\begin{lemma}
    Given a two-dimensional object $(x, y)$ and a rectangular region $[(x_b, y_b), (x_u, y_u)]$ where $(x_b, y_b)$ and $(x_u, y_u)$ denote the bottom-left and the upper-right points of the rectangular region, respectively, the probability of an object residing in the area is:
    \begin{equation*}
        P(x_b \le x \le x_u, y_b \le y \le y_u)=\big(F(x_u)-F(x_b)\big)\big(F(y_u)-F(y_b)\big)
        \label{equ:lemma1}
    \end{equation*}
\end{lemma}

\begin{proof}
    According to the definition of CDF, we have $P(x_b \le x \le x_u, y_b \le y \le y_u)=F(x_u, y_u) - F(x_b, y_u) - F(x_u, y_b) + F(x_b, y_b)$. Due to the independence assumption, we can decompose each joint CDF based on Eq. \ref{equ:jointcdf}, and obtain the equation in Lemma \ref{equ:lemma1}.
\end{proof}

The CDF in Eq. \ref{equ:jointcdf} only estimates the spatial density of  objects without considering the keyword distribution. To solve this issue, we learn the marginal CDFs, i.e., $F_{k}(x)$ and $F_{k}(y)$, for each keyword $k$. The choice of CDF models will be detailed in Section~\ref{sec6}.

With the CDF models and the sigmoid functions, we formulate the cost for processing a query $q$ with region $[(x_b, y_b), (x_u, y_u)]$ after splitting on dimensions $x$ or $y$ in Eq. \ref{equ:loss}. 
\begin{equation}
\label{equ:loss}
    \begin{aligned}
    L_{q}(x)=\sigma \big(3(x-x_{b})\big) \left|O_{1}\right| + \sigma \big(3(x_{u}-x)\big) \left|O_{2}\right|
    \\
    L_{q}(y)=\sigma \big(3(y-y_{b})\big) \left|O_{1}\right| + \sigma \big(3(y_{u}-y)\big) \left|O_{2}\right|
    \end{aligned}
\end{equation}
where $\left|O_{1}\right|$ and $\left|O_{2}\right|$ denote the number of objects containing the query keywords in the two resulting sub-spaces, respectively, which are estimated through the learned keyword-based marginal CDF models. 
The sigmoid functions (e.g., $\sigma \big(3(x-x_{b})\big)$ and $\sigma \big(3(x_{u}-x)\big)$) predicts whether the query intersects the two resultant sub-spaces, respectively. We apply stochastic gradient descent (SGD) to minimize $L_{q}(x)$ and $L_{q}(y)$ using the query workload as the training data.

\subsubsection{\textbf{Bottom Cluster Generation}}
When splitting a data space, there are both profit and loss in the query costs. The profit is gained by the reduced number of objects to be checked  while the loss reflects an increased number of sub-spaces to be checked. In Example \hyperref[exa:4.1]{4.1}, the profit and loss are equal to $2w_2$ and $2w_1$, respectively. The difference between the profit and the loss determines whether a split is needed, and where the split should be made.

Algorithm \ref{algo:train} summarizes our bottom cluster generation algorithm. The algorithm takes the query workload $W$ and the data space $S$ enclosing all geo-textual objects as the input, and it aims to return a set of clusters that minimize the cost of executing all the queries in $W$. 
\begin{algorithm}[tb]
\caption{Bottom Clusters Generation}
\label{algo:train}
\footnotesize
\KwIn{\textit{W}, the query workload; \textit{S}, the data space}
\KwOut{\textit{G}, the set of clusters}
\SetKwFunction{init}{InitializeCost}
\SetKwFunction{find}{FindOptimalPartition}
\SetKwProg{Fp}{Function}{:}{}
\textit{Q} $\leftarrow$ NewPriorityQueue()\;
\textit{Q}.Enqueue(\textit{S})\;
\textit{G} $\leftarrow \varnothing$\; 
\While{Q is not empty}
{
    $s \leftarrow$ \textit{Q}.Dequeue()\;
    $C_{s} \leftarrow$ InitializeObjectCheckingCost($s$)\;
    $opt_{x} \leftarrow$ FindOptimalPartition($s$, x)\;
    $opt_{y} \leftarrow$ FindOptimalPartition($s$, y)\;
    $best \leftarrow$ $opt_{x}$ if $opt_{x}$.cost $\le$ $opt_{y}$.cost else $opt_{y}$\;
    
    \If{$C_{s} - w_2 \cdot best.cost$ > $w_1\cdot|W|$ }
    {
        $s_{1}, s_{2} \leftarrow$ GenerateSubSpace($best$.dim, $best$.val)\;
        \textit{Q}.Enqueue($s_{1}$)\; 
        \textit{Q}.Enqueue($s_{2}$)\;
    }
    \Else
    {
        \textit{c} $\leftarrow$ GenerateMBR($s$)\;
        \textit{G}.add(\textit{c})\;
    }
}
\textbf{return} $G$\;
\BlankLine
\Fp{\find{$s, d$}}
{
    $opt.dim$ $\leftarrow$ \textit{d} \tcc*{a map structure to record  optimal split result}
    $cost, val \leftarrow$ SGDLearn(\textit{s}.queries) \tcc*{ SGDLearn() returns the optimal cost and split value}
    $opt.cost$ $\leftarrow$ $cost$\;
    $opt.val$ $\leftarrow$ $val$\;
    \textbf{return} $opt$\;
}
\end{algorithm}
The algorithm maintains a priority queue $Q$ of sub-spaces to the examined, which are prioritized  by their numbers of intersecting queries. At the start, $Q$ contains only the input data space $S$ (lines 1 and 2). Then, we iterate  through the sub-spaces in $Q$. Let the current sub-space to be split be $s$. We set the initial object checking the cost of $s$ to be $|O_s| \cdot |W_s| \cdot w_2$ where $|O_s|$ and $|W_s|$ denote the number of objects in $s$ and the number of queries intersecting with $s$, respectively (lines 5 and 6). Then, we find the optimal split along both $x$- and $y$-dimensions, respectively (lines 7 and 8), and we use the one with a smaller object checking cost as our candidate split (line 9). If the reduction in the object checking cost from $C_s$ outweighs the increase in cluster checking cost, i.e., $w_1 \cdot |W|$ (every split adds a cluster to be checked against $|W|$ queries), we execute the split and enqueue  the resultant sub-spaces (lines 10 to 13). Otherwise, $s$ is finalized, and we generate the MBR for the data objects in $s$ and use it as a bottom cluster (lines 14 to 16). The process terminates when $Q$ becomes empty (line 4).

When finding the optimal splitting value along a dimension (lines 18 to 24), we apply SGD~\cite{bottou-98x} to minimize Eq.~\ref{equ:loss} (line 21). Here,  we use a map structure $opt$ to record the new object checking cost, the dimension, and the value of a learned optimal split.

The time complexity of each iteration in Algorithm \ref{algo:train} is $\mathcal{O}\big(h \cdot (E_{x} + E_{y})\big)$ where $h$ and $E_{d}, d \in \{x, y\}$ denote the time complexity of SGD per iteration and the number of epoches respectively. 
Recall that the time complexity of the brute-force algorithm is $\mathcal{O}\big(f\cdot(|V_{x}|+|V_{y}|)\big)$. We note that $f$ is larger than $h$ because our  heuristic algorithm does not need to run a split to calculate a query cost (while the brute-force algorithm does). $E_{x}$ and $E_{y}$ depend on the algorithm configurations, such as the learning rate and the number of model parameters. They are usually much smaller than $V_{x}$ and $V_{y}$, respectively.
Therefore, the time complexity of our heuristic algorithm is lower than that of the brute-force algorithm.

\section{Bottom-up Packing}
\label{sec5}
The bottom clusters generated from Section \ref{sec4} can be used as a flat and coarse-grained index. To further improve the pruning power of our index, we build a hierarchical structure over the clusters. 

\subsection{Design Considerations}
As shown in Section \ref{sec3}, when executing a query with a hierarchical index, we traverse  all qualified nodes until reaching the leaf nodes. An internal node and its descendants can be pruned if it does not intersect with the query or include any query keyword. We build our hierarchical index level by level, i.e., recursively packing the clusters to maximize the reduction in the pruning cost at each level. Here we omit the object checking costs as they are only triggered on the bottom clusters.


\subsubsection{\textbf{Optimization Goal}}
The query time spent on node pruning can directly reflect the pruning capability of a hierarchical index. Measuring the query time, however, needs to run all queries in the query workload on an existing index, which is not suitable to be used as an optimization metric of our bottom-up packing problem. We observe that the pruning time cost is proportional to the number of accessed nodes for the workload such that it can be used to evaluate the pruning capability. To adopt this criterion, we associate each bottom cluster $c_{i}$ with a query label set denoted by  $c_{i}.l$. If a cluster $c_{i}$ intersects with a training query $q_{j}$ and its textual document includes any keyword of $q_{j}$, we add $q_{j}$ to the query label set of this cluster, that is, $c_{i}.l = \{q_{j}\}$. 
During packing, the labels of a node in an upper level (an ``upper node'' for short hereafter) can be easily generated by merging all labels of its sub-tree. 

\subsubsection{\textbf{Bottom-up Packing Problem}}
Next, we define the bottom-up packing problem to minimize the number of accessed nodes.

\begin{prob}[\textbf{Bottom-up Packing}]
\label{prob3}
    Given a query workload W and the set of bottom clusters $G$, the bottom-up packing  process aims to generate a hierarchical index $I$ that minimizes the number of accessed nodes to process the queries in $W$.
\end{prob}

Given the leaf nodes, i.e., bottom clusters, we can build a hierarchical index using techniques from traditional indexes such as the CDIR-Tree. However, those techniques only consider the underlying data distribution, which might lead to worse performance as shown in the later experiment.

To address these issues, and motivated by the strong performance of query-aware structures learned by reinforcement learning (RL), we propose an RL-based algorithm to learn a packing. 
%
We construct our index level by level with a bottom-up packing process, and we model the  packing problem at each level as a sequential decision-making process, i.e., a Markov decision process, which makes it solvable by RL. To pack each level, the nodes from a lower level to be packed (``bottom nodes'' hereafter) are processed  sequentially, and we find an upper node to host each bottom node until there are no more bottom nodes. After the packing process of one level stops, the non-empty upper nodes become the new bottom nodes to be packed for the next level. 

\subsection{Packing with Reinforcement Learning}
We propose an RL-based packing algorithm following the idea of the  Deep-Q-Network (DQN)~\cite{DBLP:journals/nature/MnihKSRVBGRFOPB15} to learn the optimal policy (i.e., a packing strategy) for solving the packing problem (Problem \ref{prob3}). To form a tree structure, we require that the number of upper nodes does not exceed that of the bottom nodes. 

There are two main challenges in our packing problem.

\begin{enumerate}
    \item To use a neural network to estimate an expected reward (e.g., the reduction in the number of node accesses), the states (e.g., the relation of two levels resulting from a packing decision) need to be represented by a fixed-length vector. However, there are many different possibilities of bottom nodes, and it is challenging to generate such a vector to encode the current packing of bottom nodes effectively.
    
    \item Every time a node is added to the structure, it may lead to a reduced reward (i.e., more node accesses). However, it is necessary to add nodes to the structure continuously such that the structure can be built up. How to adapt the cost model for this case is another challenge. 
\end{enumerate}

To address these challenges, we formulate an MDP process for our packing problem as follows:
\noindent \textbf{States.} A state needs to capture the status of a (partially packed) level in an index structure. 
%
\begin{figure}[tb]
    \centering
    \includegraphics[width=.55\linewidth]{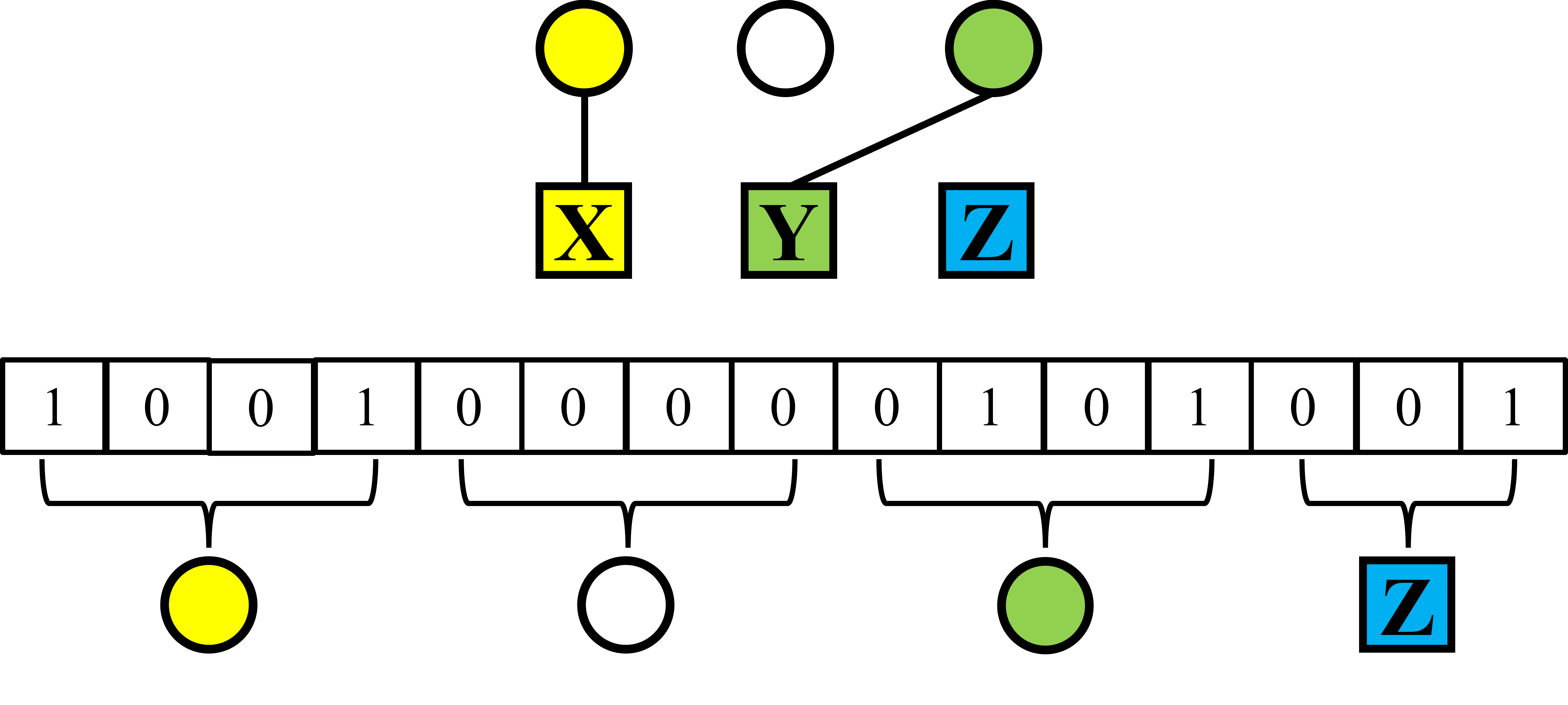}
    \vspace{-0.3cm}
    \caption{An example of the state representation}
    \label{fig:rlstate}
\end{figure}
As mentioned above, the number of bottom nodes bounds that of the upper nodes. Hence, we initialize $N$ empty upper nodes given $N$ bottom nodes. Consider $m$ queries are used in the learning process. 
Each of the $N$ upper nodes to be constructed takes an $(m+1)$-dimensional vector representation. The first $m$ dimensions denote whether the node is labeled by each of the $m$ queries, and the last dimension is a count on the number of bottom nodes to be connected to this node. The $N$ upper nodes together form an $(m+1)\cdot N$-dimensional vector. We further append $m$ dimensions to the vector to represent the query label of the next bottom node to be connected to (i.e., packed into) one of the upper nodes. Overall, these form an $\big((m+1)\cdot N +m\big)$-dimensional vector representing a state. 
Figure \ref{fig:rlstate} shows an example, assuming $m=3$ queries and $N=3$ bottom nodes ($X$, $Y$, and $Z$). The circles denote upper nodes, and the colors denote different query labels. 



\noindent \textbf{Actions.} An action adds a bottom node to an upper node. To make the action space and the state representation consistent, we define the action space $A=\{1, 2, \ldots, N\}$, where action $a = i$ denotes packing the next bottom node into the $i$-th upper node.

\noindent \textbf{Transition.} 
Given a state and an action, the agent transits to a new state by packing a bottom node into the chosen upper node and moving on to the next bottom node. The agent reaches a terminal state when there are no more bottom nodes to be packed.

\noindent \textbf{Reward.} A larger reward represents a packing with better quality. Since we aim to reduce the number of  node accesses when executing the query workload, the reward signal should reflect the expected number of node accesses before and after taking an action. 

We propose to use the average number of node  accesses per query to formulate the rewards since the total number of node accesses grows   monotonically as more bottom nodes are added to the consideration, which will lead to constant negative rewards. 
\begin{equation}
\label{equ:reward}
    r = N_{a} - N_{a}^{\prime}
\end{equation}

The reward function is formulated as Eq.~\ref{equ:reward}, where $N_a$ ($N_a^{\prime}$) denotes the average number of node accesses before (after) action $a$ is taken. The agent chooses the action that maximizes the reward during exploitation. Additionally, we observe a positive correlation between the sum of rewards and the reduction in the average number of node accesses after packing all bottom nodes. Let $N_{a}^{*}$ be the average number of node accesses of packing the last bottom node. As $N_{a}$ in each iteration is identical to $N_{a}^{\prime}$ of the last iteration, the sum of rewards after packing all $N$ bottom nodes is equal to $1 - N_{a}^{*}$, and it is positively correlated to $N + 1 -N_{a}^{*}$. Note that the number of node accesses is equal to $N + 1$ before creating the upper nodes. Thus, if the sum of the rewards at a level is not larger than $-N$, the bottom-up packing process will be terminated.

\noindent \textbf{Example 5.2:\label{exa5.2}} Figure \ref{fig:mdpexample} presents an example of the MDP for the bottom-up packing problem. Same as in Figure~\ref{fig:rlstate}, the colors represent different query labels. Here, we only show state transitions with nonzero probabilities, and we have omitted the rewards to avoid clutter. The ellipse nodes and edges represent states and actions,  respectively. Since there are 3 bottom nodes (rectangle), we initialize 3 upper nodes (circle), and the bottom nodes are to be packed sequentially. When no incoming bottom node is to be inserted at one level, i.e. the leaf node in Figure \ref{fig:mdpexample}, we reach the terminal states at this level and move to the upper level.
\begin{figure*}[tb]
    \centering
    \includegraphics[width=.95\linewidth]{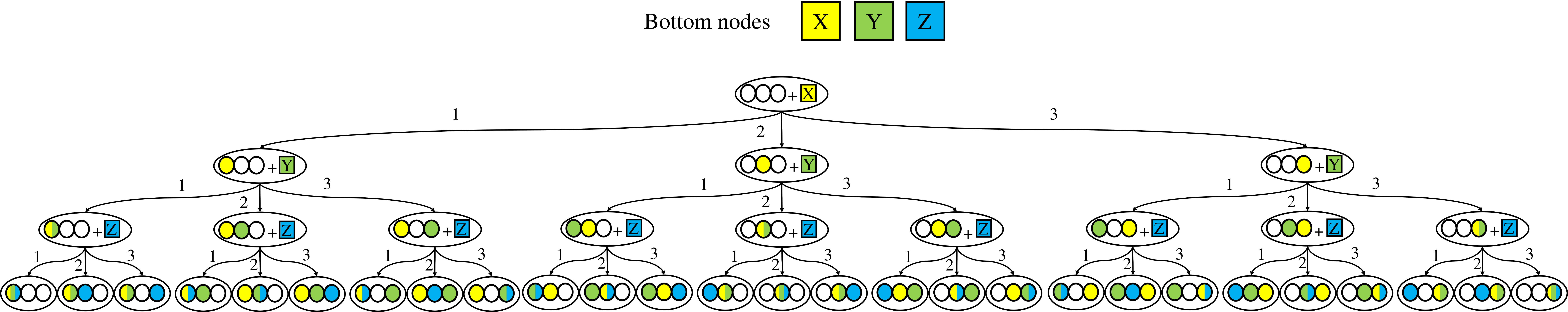}
    \abovecaptionskip 0.2cm
    \caption{An example of MDP formulation for Problem \ref{prob3}}
    \label{fig:mdpexample}
\end{figure*}

\subsection{Training}
Recall that Q-learning is a commonly used RL algorithm as introduced in Section \ref{sec2.3}. We train a deep Q-network (DQN) \cite{DBLP:journals/nature/MnihKSRVBGRFOPB15} to project the high-dimensional state and action spaces to low-dimension spaces using neural networks and efficiently predict the value of the Q-function $Q(s, a)$. In our model, we adopt the deep Q-learning with a technique known as experience replay where we store the agent's experience $e_{t} = (s_{t}, a_{t}, r_{t}, s_{t + 1})$ at each time-step $t$. We implement two networks, a policy network $Q$ and a target network $\hat{Q}$ separately, which has been shown to be more stable than using only one network as done in the standard Q-learning~\cite{DBLP:journals/nature/MnihKSRVBGRFOPB15}.

Given a batch of transitions $(s, a, r, s^{\prime})$, the policy network parameters $\theta$ are updated with a gradient descent step by minimizing the mean square error (MSE) loss as shown in Eq. \ref{equ:rloss}, where $\gamma\in(0,1)$ denotes a discount factor determining the importance of future rewards, and $\theta^{-}$ are the parameters of the target network. 
\begin{equation}
\label{equ:rloss}
    L(\theta) = \sum_{s,a,r,s^{\prime}} \left(r + \gamma\max_{a^{\prime}}\hat{Q}(s^{\prime}, a^{\prime};\theta^{-})-Q(s, a;\theta)\right)^{2}
\end{equation}

Note that the target network parameters $\theta^{-}$ are only synchronized with the policy network parameters $\theta$ every \textit{T} steps and are held fixed between weight updates. However, directly copying the weights has been shown to be unstable due to noise and outliers. Inspired by prior works \cite{DBLP:journals/nn/KobayashiI21,DBLP:journals/corr/LillicrapHPHETS15}, we apply the soft update (Eq. \ref{equ:softup}) to the target network. The weights of the target network are updated by interpolating 
between the weights of the target network and those of the policy network through a fixed ratio $\tau = 0.001$~\cite{DBLP:journals/corr/LillicrapHPHETS15}.
\begin{equation}
\label{equ:softup}
    \theta^{-} = \tau \theta + (1-\tau) \theta^{-}, \tau \ll 1
\end{equation}

We present the learning process in Algorithm \ref{algo:rltrain}. We first initialize the policy network and the target network with the same random parameters (line 1). In each epoch, we reset the replay memory $M$ and the set of upper nodes $G_{u}$ (line 3). Then, the learning process sequential packs the bottom nodes to the upper nodes (lines 4 to 13). For every incoming bottom node $c_i$, we generate the state by combining $G_{u}$ and $c_i$ (line 5). We compute the average number of node accesses based on the query labels of the current upper nodes (line 6). To balance between RL  exploration and exploitation, we use the $\epsilon$-greedy algorithm \cite{DBLP:journals/tnn/SuttonB98} to choose a random action with probability $\epsilon$ (i.e., exploration) or the action that maximizes the action-value function of the policy network (i.e., exploitation) (line 7). After $c_i$ is packed, we update the state representation and compute the average value again (lines 8 and 9). Then, we compute the reward and store this transition in the replay memory (lines 10 and 11). To train the DQN, we draw a batch of transitions to train the policy network (line 12) and periodically copy the policy network parameters to the target network (line 13). Finally, we use the learned action-value function $Q(s, a; \theta)$ to pack the nodes.
\begin{algorithm}[tb]
\caption{DQN Learning for Node Packing}
\label{algo:rltrain}
\footnotesize
\KwIn{$G$, the bottom nodes with query labels; $M$, replay memory; $E$, the number of epochs}
\KwOut{$Q(s, a; \theta)$, action-value function}
Initialize $Q(s, a; \theta), \hat{Q}(s^{\prime}, a^{\prime}, \theta^{-})$\;
\For{epoch $\in [1, E]$}
{
    $G_{u} \leftarrow$ NewList(); $M \leftarrow$ NewList()\;
    \For{$c_i \in$ G}
    {
        Update $s$ using $G_{u}$ and $c_i$\;
        Compute the average number of node accesses $N_{a}$ according to $G_{u}.l$\;
        Choose $a$ by the $\epsilon$-greedy method\;
        Pack $c_i$ into $G_{u}$[$a$] and generate the new state $s^{\prime}$\;
        Compute $N_{a}^{\prime}$ according to the new $G_{u}.l$\;
        Compute reward $r$ based on Eq. \ref{equ:reward}\;
        Store transition $(s,a,r,s^{\prime})$ into $M$\;
        Draw a batch of samples from $M$ and perform a gradient step based on Eq. \ref{equ:rloss}\;
        Update $\hat{Q}(;\theta^{-})$ with $Q(;\theta)$ softly based on Eq. \ref{equ:softup} after every \textit{C} steps\;
    }
}
\Return{$Q(s, a; \theta)$\;}
\end{algorithm}

\section{Design Optimizations}
\label{sec6}
\textbf{Choice of CDF models.} For our heuristic partition algorithm, the number of objects with each keyword is approximated by a model that learns the corresponding CDF. Prior works  \cite{DBLP:conf/sigmod/KraskaBCDP18, DBLP:conf/sigmod/Li0ZY020} have used the neural network (NN) to learn the CDF. 
However, learning an NN for each query keyword may
lead to a large number of NNs to be learned and hence high preparation costs.

We observe that the query time of \idxname\ is more sensitive to high-frequency keywords. To decrease preparation costs, we divide keywords into three classes based on their frequency: low ($\leq 0.001 \text{\textperthousand}$), medium ($0.001-0.1 \text{\textperthousand}$), and high ($\geq 0.1 \text{\textperthousand}$). \textcolor{edit}{Previous studies \cite{yang2019gb, wang2014selectivity} have shown that more resources should be allocated to records with more frequent elements to get better prediction accuracy.} When calculating the query cost, low-frequency keywords are ignored as they have little impact on the query time. We adopt Gaussian functions to approximate the data distribution of each medium-frequency keyword and learn an NN to approximate the CDF of each rest keyword. Our empirical results show that such a strategy balances the preparation costs and the query time.

\noindent \textbf{Correlation between keywords.} In Section \ref{sec4}, we consider each keyword independently when approximating $|O_{s}|$ ($s \in \{1, 2\}$) in Eq.~\ref{equ:loss}. This independence  assumption impacts the performance of the heuristic partitioning algorithm when a query has  more than one keyword, e.g., if an object contains $k$ query keywords of a query, this object will be counted $k$ times when predicting the number of objects in a sub-space, leading to inaccurate query cost prediction.

To solve this issue, we exploit \textit{frequent itemset mining} to discover all frequent keyword sets and extract associations among the given set of keywords \cite{DBLP:conf/sigmod/AgrawalIS93, DBLP:books/mit/fayyadPSU96/AgrawalMSTV96, DBLP:reference/ml/Toivonen10c, DBLP:conf/sigmod/HanPY00}. We apply a classic algorithm, FP-Tree \cite{DBLP:conf/sigmod/HanPY00}, to find frequent keyword sets from the underlying data. Then, we learn a CDF model of objects containing all keywords in one frequent keyword set and use the learned model to predict the number of objects with the set of query keywords more accurately.

\noindent \textbf{Action mask in RL.} 
When the packing of a level starts, the upper nodes are all empty. To choose the upper node to insert for the first bottom node, we observe that actions $a=i$ ($i>1$) are all equivalent to $a=1$. We call these actions \emph{duplicated actions}. Duplicated actions exist when more than one upper nodes are empty. As observed in Figure \ref{fig:mdpexample}, the actions of adding the bottom node to any of these empty upper nodes are equivalent. Such duplicated actions make the exploration inefficient, leading to slow convergence~\cite{DBLP:journals/corr/abs-2103-04541}. 

Thus, motivated by a prior study~\cite{DBLP:conf/sigmod/ZhangCZ022}, another use of the environment is to generate an action mask based on the current state to hide the duplicated actions from the agent. In the example above, before inserting the first bottom node, the action mask generated by the environment makes the agent only chooses action $a = 1$. 

\noindent \textcolor{edit}{\textbf{Training time acceleration.} \idxname has two steps: finding the bottom clusters and packing the bottom clusters through RL. To reduce the training time of \idxname, we design acceleration techniques for both steps. The first technique is to use sampled training queries, following a previous work~\cite{DBLP:conf/sigmod/NathanDAK20}. We use stratified  sampling~\cite{botev2017variance} to obtain query samples that can better represent the distribution of the original workload. 
The second technique groups the bottom clusters using a clustering algorithm to reduce the number of bottom clusters to be packed in the bottom-up packing step. 
We utilize the spectral clustering \cite{ng2001spectral} with the coordinates of the bottom left and top right points of each bottom cluster as features. }

\section{Experiments}
\label{sec7}

    
    
\subsection{Implementation and Setup}
\noindent\textbf{Implementation.} 
The learning process of CDF NN models, Algorithm \ref{algo:train}, and Algorithm \ref{algo:rltrain} are implemented with PyTorch \cite{DBLP:conf/nips/PaszkeGMLBCKLGA19}. The performance evaluation of all index structures is implemented in C++ and compiled using GCC 9.3 with -O3 flag. In the process of generating the bottom clusters, we empirically set 0.1 and 1 to the weights of stage 1 and stage 2, i.e., $w1$ and $w2$, respectively. The CDF network consists of 4 layers, and each hidden layer has 16 units. We use ReLU as the activation function of the hidden layer. The output of the CDF is activated by a sigmoid function. When packing the bottom clusters, we follow the original implementation of DQN \cite{DBLP:journals/nature/MnihKSRVBGRFOPB15}. The neural network consists of 3 layers, and each hidden layer has 64 units. We set the capacity of experience replay to 256, and the discount factor is set to 0.99. For $\epsilon$-greedy algorithm, the initial value of $\epsilon$ is set to 1, and the value decreases with more learning steps, which balances exploration and exploitation well.

\noindent\textbf{Environment.} 
We run single-threaded experiments \textcolor{edit}{in the main memory} on an Ubuntu machine with Intel(R) Xeon(R) Silver 4210R CPU @ 2.40GHz, 128GB RAM, and a 500 GB SSD disk. Besides, we train our CDF models on an Ubuntu machine with Intel(R) Xeon(R) Gold 6240 CPU @ 2.60GHz, 256GB RAM, and RTX 2080 Ti GPU.

\noindent \textbf{Baselines.} We compare \idxname\ with four SOTA conventional indexes, i.e., \textbf{CDIR-Tree} \cite{DBLP:journals/pvldb/CongJW09}, \textbf{SFC-Quad} \cite{DBLP:conf/cikm/ChristoforakiHDMS11}, \textbf{ST2I} \cite{DBLP:conf/cikm/Hoang-VuVF16}, and \textbf{ST2D} \cite{DBLP:conf/ssd/TampakisSDPKV21}. We implement these indexes using the default parameter values reported in their original papers. Note that ST2D is only evaluated on \textbf{FS} by setting the similarity threshold to 0 since it is only suitable for the case that containing a few distinct keywords (a few hundreds) because of the textual clustering.

We also integrate a learned spatial index with a textual index loosely, following traditional spatial keyword indexes. This results in a \emph{learned spatial-first index} (\textbf{SFI}) and a \emph{textual-first index} (\textbf{TFI}). SFI attaches an inverted file for keywords indexing to each leaf node of a learned spatial index, while TFI uses an inverted file as its top-level index and creates a learned spatial index for the objects containing the same keyword. It has been shown that textual-first indexes outperform their spatial-first counterparts~\cite{DBLP:conf/cikm/ZhouXWGM05, DBLP:conf/ssd/VaidJJS05}. Therefore, we only report results for TFI in our experiments. LISA~\cite{DBLP:conf/sigmod/Li0ZY020} is used as the learned spatial index since it returns the exact results. We further extend a learned multi-dimensional index, i.e., Flood~\cite{DBLP:conf/sigmod/NathanDAK20}. We build an inverted file for each grid cell in Flood and also improve its cost function for building the grid index by incorporating the textual information, utilizing our CDF models on the geo-textual data, following the method presented in Section~\ref{sec4}. We denote this index by \textbf{Flood-T}. It splits the data along \textit{only one dimension} in the 2D geographical space, which limits its capability to capture the complex data distribution. 
\textcolor{edit}{We also compare with \textbf{LSTI} \cite{ding2022learned}, the latest index to support spatial keyword queries. This method maps the data into one dimension using a Z-order curve based on the spatial coordinates and builds a RadixSpline index \cite{kipf2020radixspline} using the mapped values. Then, an inverted file is created for each spline point by scanning the dataset again. 
}

\subsection{Datasets and Workloads}
We use three real-world datasets in Table \ref{exp:data}.
The \textbf{FS} dataset \cite{DBLP:conf/www/YangQYC19} consists of global-scale check-in records of Foursquare (\url{https://foursquare.com/}) from Apr. 2012 to Jan. 2014. A check-in data has a spatial location and its category. The \textbf{SP} dataset includes recreational and sports areas extracted from OpenStreetMap (\url{https://www.openstreetmap.org}).
We use the center of each area and the original description as the spatial location and keywords, respectively. The \textbf{BPD} dataset contains global POIs published by the SLIPO project~\cite{DBLP:conf/ssd/PatroumpasSMGA19} (\url{http://slipo.eu/}). The \textbf{OSM} dataset contains 100M POIs extracted from OpenStreetMap, which is published in UCR STAR \cite{GVE+19}. Each POI has a point location, and its keywords include all related information such as street and category.
\begin{table}[tb]
\small
\caption{Dataset Statistics}
\vspace{-0.3cm}
\begin{tabular}{|c|c|c|c|c|}
\hline
\textbf{Property}        & \textbf{FS} & \textbf{SP} & \textbf{BPD} & \textbf{OSM} \\ \hline

Number of data objects & 3M & 4M & 25M & 100M\\ \hline
Number of distinct keywords & 462 & 1M & 24M & 447M\\ \hline
Total number of keywords & 6M & 11M & 116M & 478M\\ \hline
\end{tabular}
\label{exp:data}
\end{table}
\begin{table}[tb]
\small
\caption{Parameters and their settings}
\vspace{-0.3cm}
\begin{tabular}{|c|c|}
\hline
\multicolumn{1}{|c|}{\textbf{Parameter}} & \multicolumn{1}{c|}{\textbf{Setting}} \\ \hline
Query distribution & UNI LAP GAU \textbf{\underline{MIX}} \\ \hline
Query region size (\%) & 0.005 0.01 \textbf{\underline{0.05}} 0.1 0.5 1\\ \hline
Number of query keywords & 1 3 \textbf{\underline{5}} 7 9 \\ \hline
\end{tabular}
\label{exp:query}
\end{table}
As there is no public real-world query workload for the geo-textual datasets, we generate the queries by following previous works \cite{DBLP:journals/pvldb/WangQWWZ21, DBLP:journals/pvldb/ChenCJW13, DBLP:journals/pvldb/WangZZLW14, DBLP:conf/sigmod/HuLXAPRY22, DBLP:journals/corr/abs-2103-04541}.
Specifically, to generate a query, we first sample an object in the dataset, and then generate a bounding rectangular area with the location of this object being its center.
Inspired by previous works \cite{DBLP:journals/pvldb/WangQWWZ21, DBLP:journals/corr/abs-2103-04541}, we use four methods to generate the centers: 
(i) \textbf{UNI}, where centers are uniformly sampled from the dataset.
(ii) \textbf{LAP}, where centers are sampled from the Laplace distribution \cite{Kotz_2001}. We set the location and scale parameters, i.e., $\mu$ and $b$, to $|D| / 2$ and $|D| / 10$ respectively, where $D$ is the object set.
(iii) \textbf{GAU}, where centers are sampled from a Gaussian distribution ($\mu=|D| / 2$, $\sigma=100$).
(iv) \textbf{MIX}, composing of the centers generated from the (i) and (ii) in equal proportions.
Finally, we associate keywords for the queries following prior works \cite{DBLP:journals/pvldb/ChenCJW13, DBLP:journals/pvldb/WangZZLW14}.
If the number of query keywords is less than the number of keywords of the center, we choose the query keywords from the sampled object. Otherwise, we randomly choose the remaining keywords from the global keyword set.

To evaluate the performance of indexes in different scenarios, we generate query sets with different numbers of keywords and query sizes. Table \ref{exp:query} summarizes parameters, where default values are in bold and underlined.
We generate 2000 queries under each setting, in which 1000 queries are utilized to test the performance of all the indexes, and others are used to train learned indexes.

\subsection{Query Time Evaluation}
\textcolor{edit}{To evaluate the query time, we execute testing queries 100 times and report the average cost of the queries in each query set.}

\subsubsection{\textbf{Effect of query distribution.}}
In this experiment, we fix other settings except for the query distribution and show the results on all datasets in Figure \ref{exp:dis}.
Clearly, conventional indexes (SFC-Quad, ST2I, and CDIR-Tree) perform worse on the skewed workload since they do not use the query characteristics when constructing the index. 
For the learned indexes, TFI performs even worse than the conventional indexes since it only loosely combines a learned spatial index with a textual index.
Flood-T shows a slight fluctuation in its performance since it learns from the underlying data and the query workload simultaneously, but in the geo-textual scenario, it only splits along one dimension, making it incompatible with the skewed workload.
Our \idxname\ improves the partitioning and adopts RL to build a tree, so it is less sensitive to this alteration.
\begin{figure}[htb]
    \begin{minipage}{.55\linewidth}
        \centering
        \includegraphics[width=\linewidth]{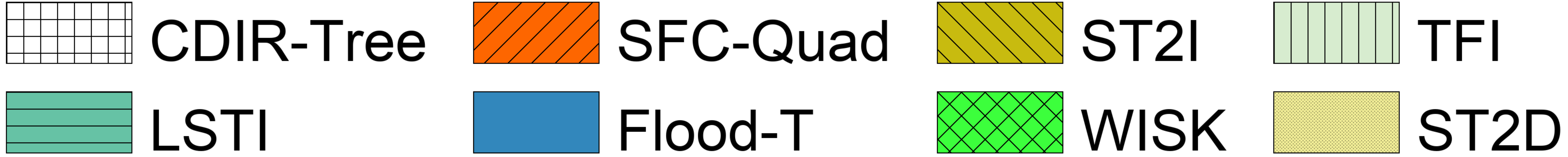}
    \end{minipage}
    \begin{minipage}{\linewidth}
        \centering
        \subcaptionbox{FS\label{dis-fs}}{
            \vspace{-0.2cm}
            \includegraphics[width=.23\columnwidth]{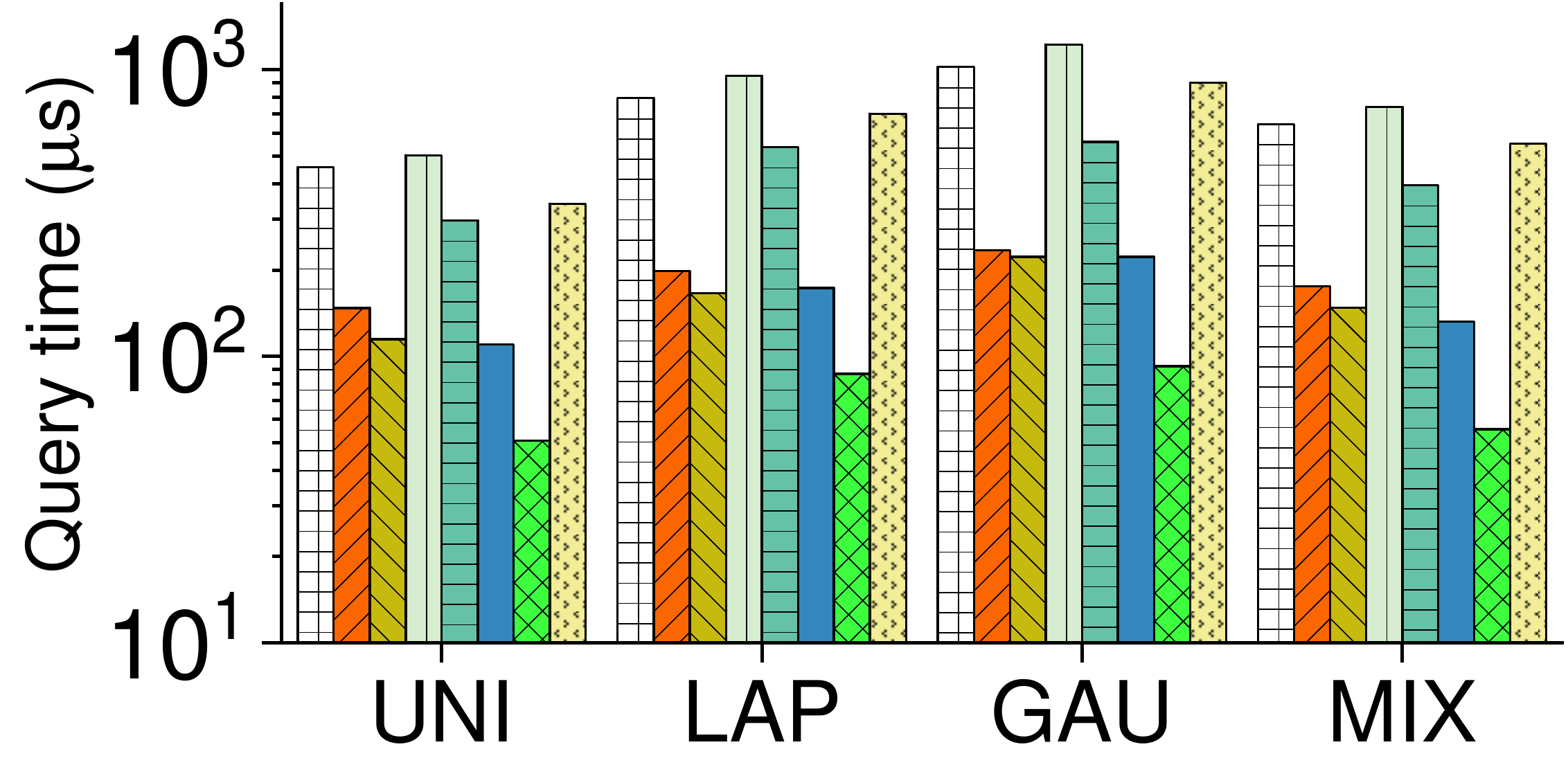}
        }
        \subcaptionbox{SP\label{dis-sp}}{
            \vspace{-0.2cm}
            \includegraphics[width=.23\columnwidth]{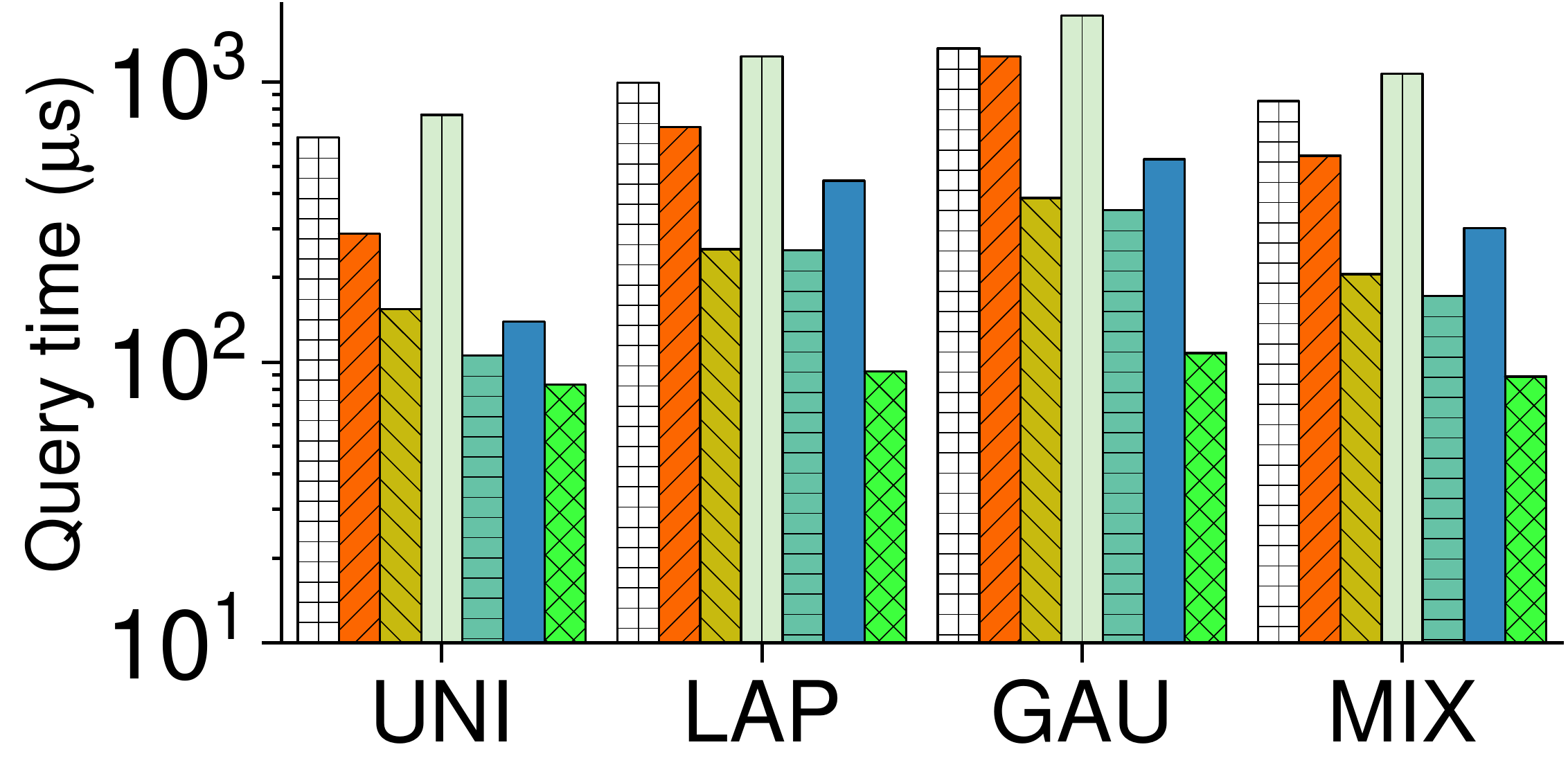}
        }
        \subcaptionbox{BPD\label{dis-bpd}}{
            \vspace{-0.2cm}
            \includegraphics[width=.23\columnwidth]{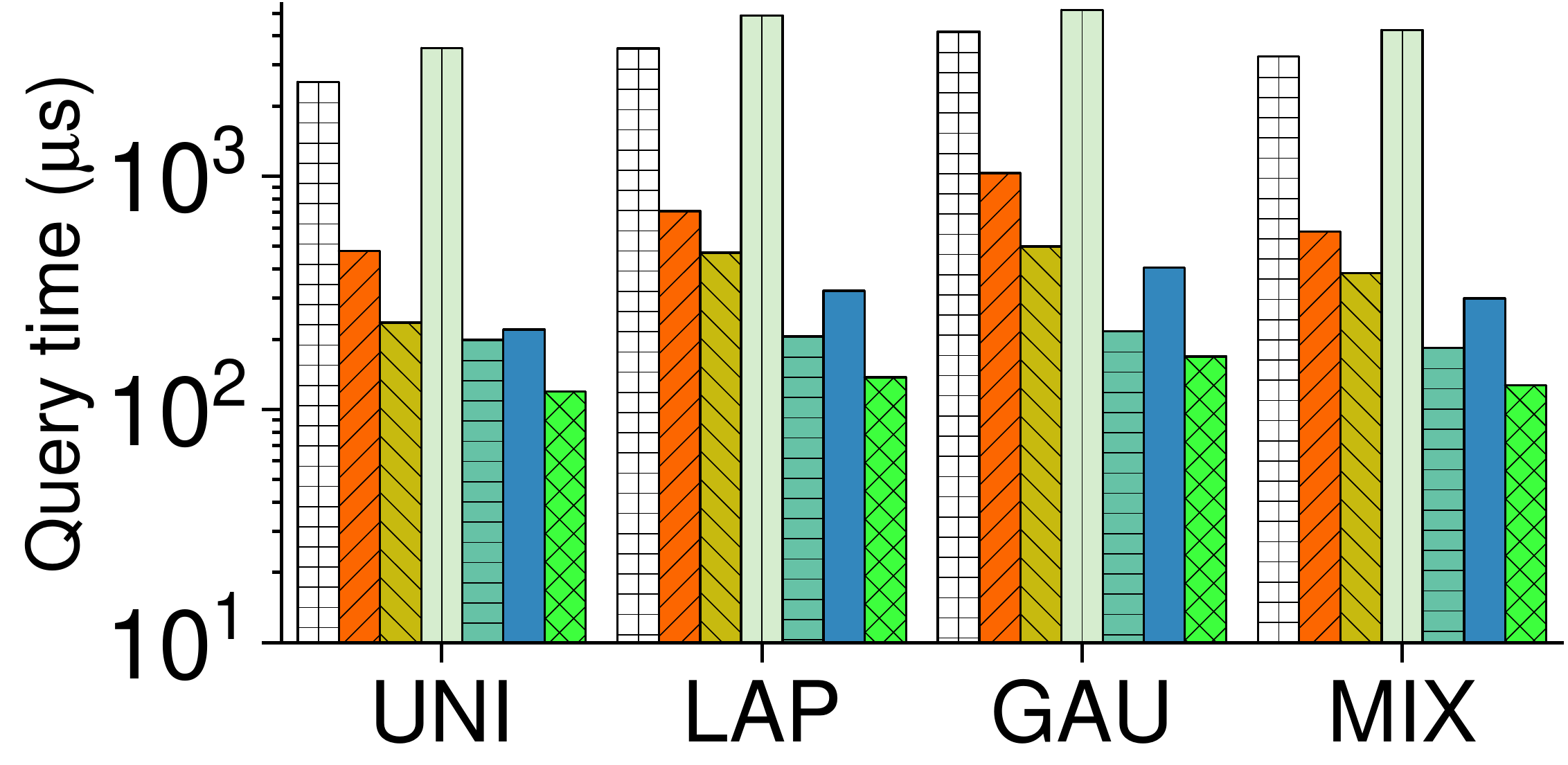}
        }
        \subcaptionbox{OSM\label{dis-osm}}{
            \vspace{-0.2cm}
            \includegraphics[width=.23\columnwidth]{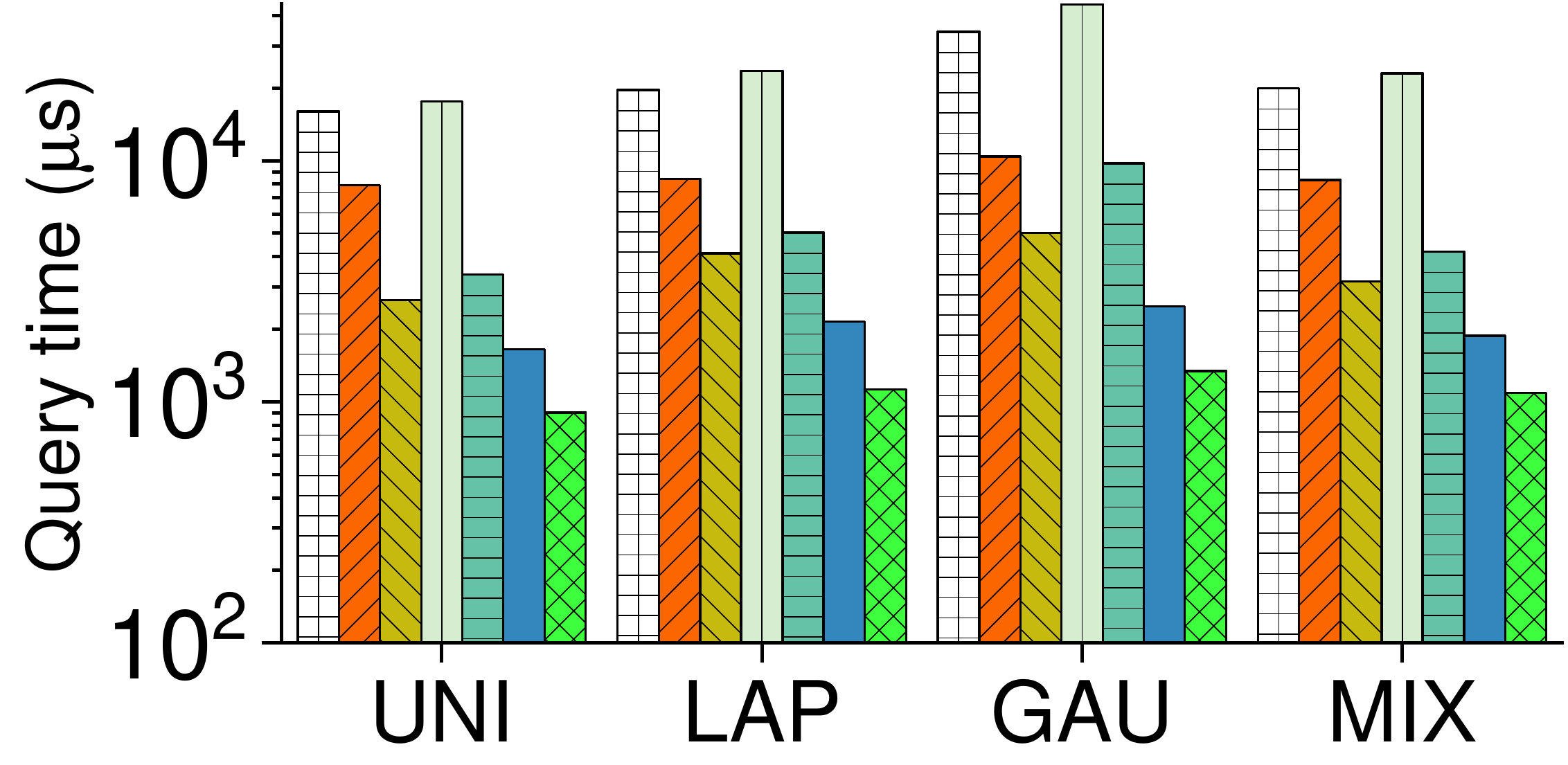}
        }
    \end{minipage}
    \vspace{-0.3cm}
    \caption{Varying the distribution of the query workload}
    \label{exp:dis}
\end{figure}
\subsubsection{\textbf{Effect of query region size.}}
We show the performance of all indexes, by varying the query region size varies from $0.005\%$ to $1\%$ of the whole region in Figure \ref{exp:size}.
Again, \idxname\ performs the best on four datasets.
Besides, Flood-T performs slightly worse than ST2I on \textbf{SP}, even though it optimizes its layout by learning from the data and query workload. It is because it only splits the whole region along one dimension.
Thus, we improve this process to generate the leaf nodes of \idxname and also pack the bottom clusters into a hierarchical structure.
The two techniques simultaneously result in the superiority of \idxname\ over the other indexes.
\begin{figure}[tb]
    \begin{minipage}{.55\linewidth}
        \centering
        \includegraphics[width=\linewidth]{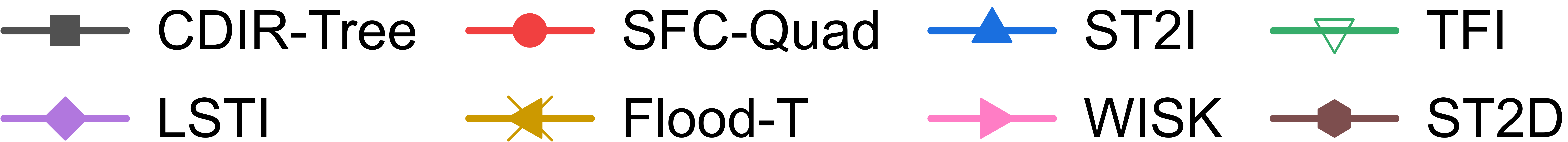}
    \end{minipage}
    \begin{minipage}{\linewidth}
        \centering
        \subcaptionbox{FS\label{size-fs}}{
            \vspace{-0.2cm}
            \includegraphics[width=.23\columnwidth]{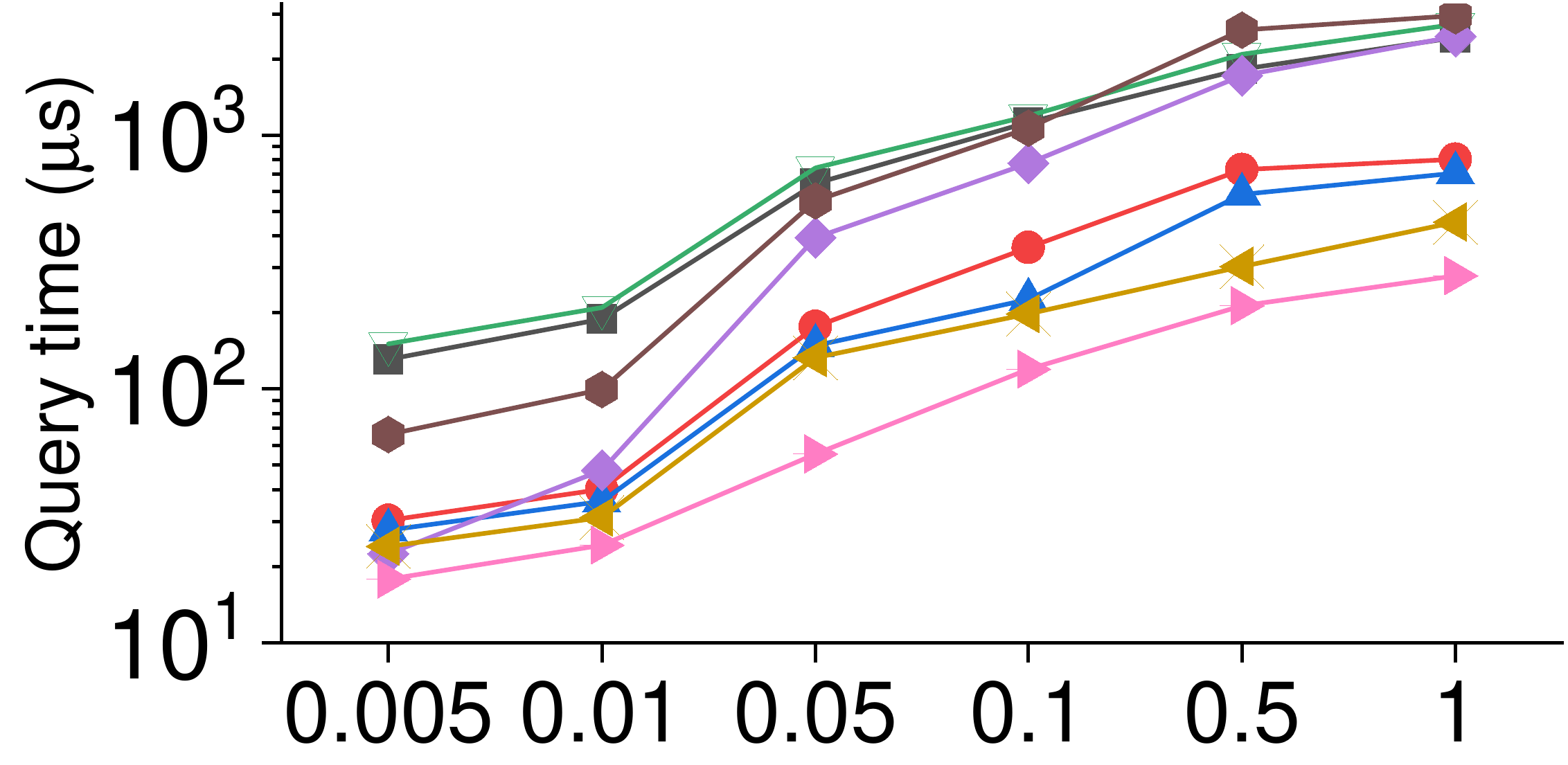}
        }
        \subcaptionbox{SP\label{size-sp}}{
            \vspace{-0.2cm}
            \includegraphics[width=.23\columnwidth]{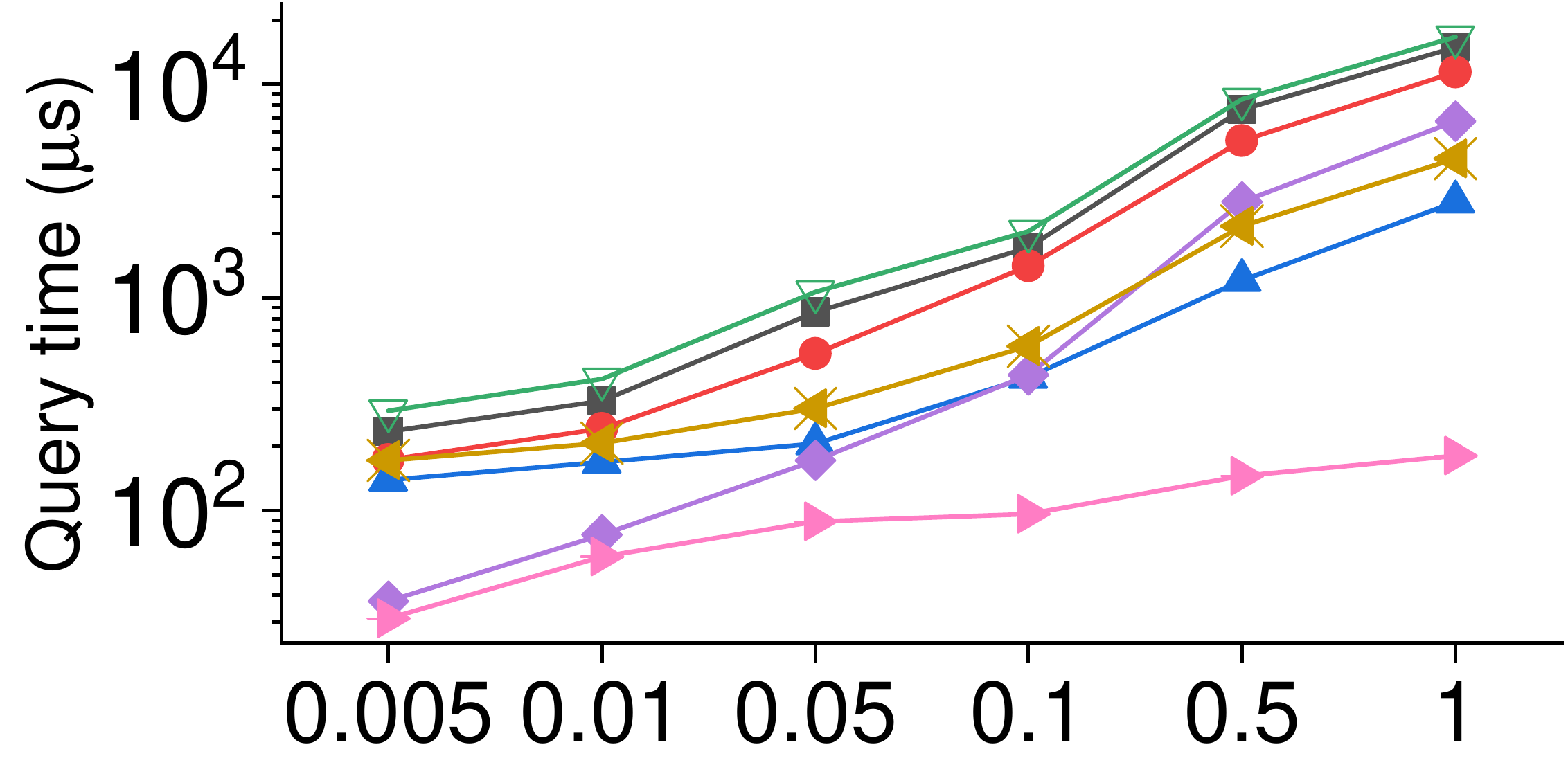}
        }
        \subcaptionbox{BPD\label{size-bpd}}{
            \vspace{-0.2cm}
            \includegraphics[width=.23\columnwidth]{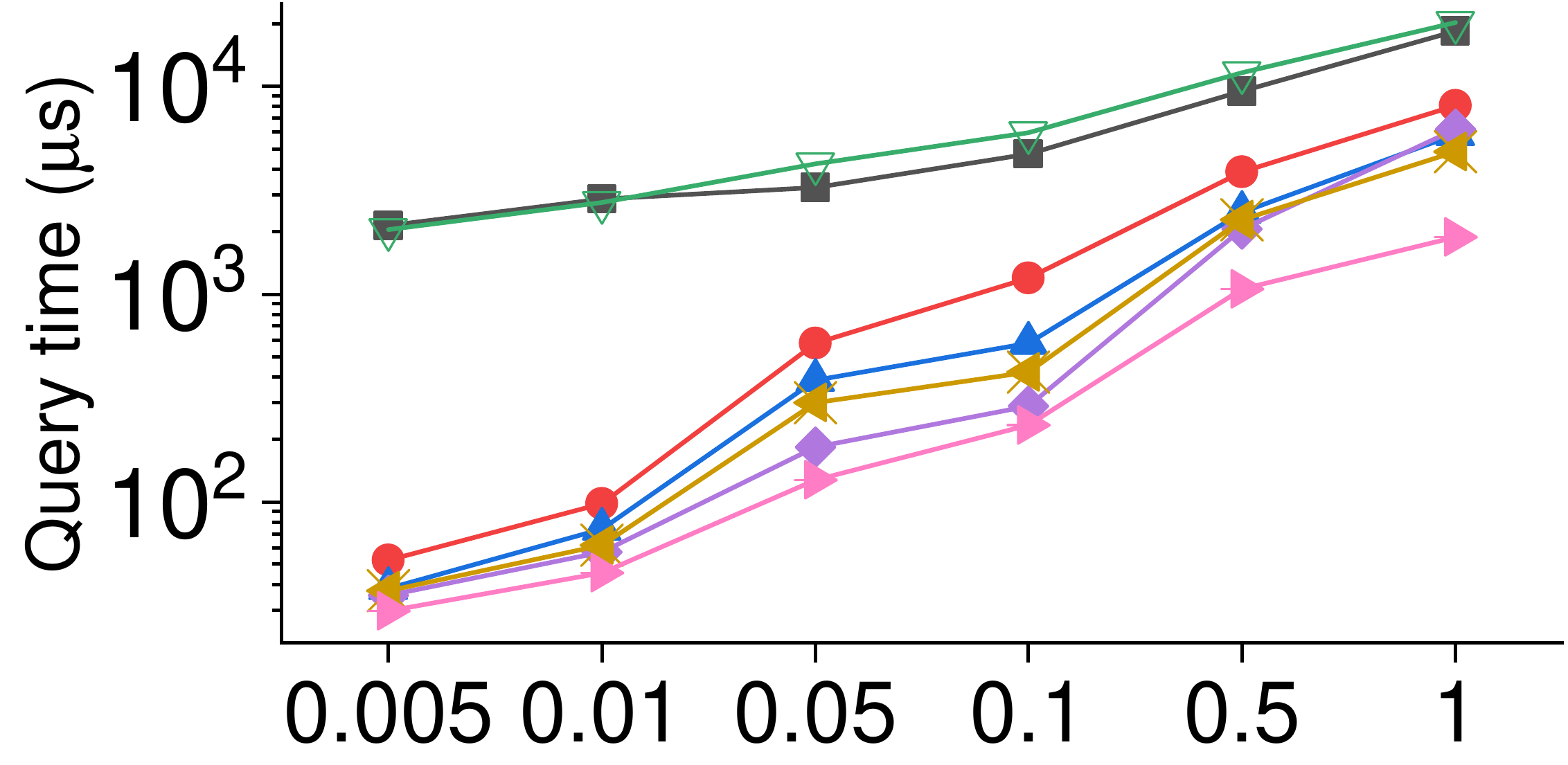}
        }
        \subcaptionbox{OSM\label{size-osm}}{
            \vspace{-0.2cm}
            \includegraphics[width=.23\columnwidth]{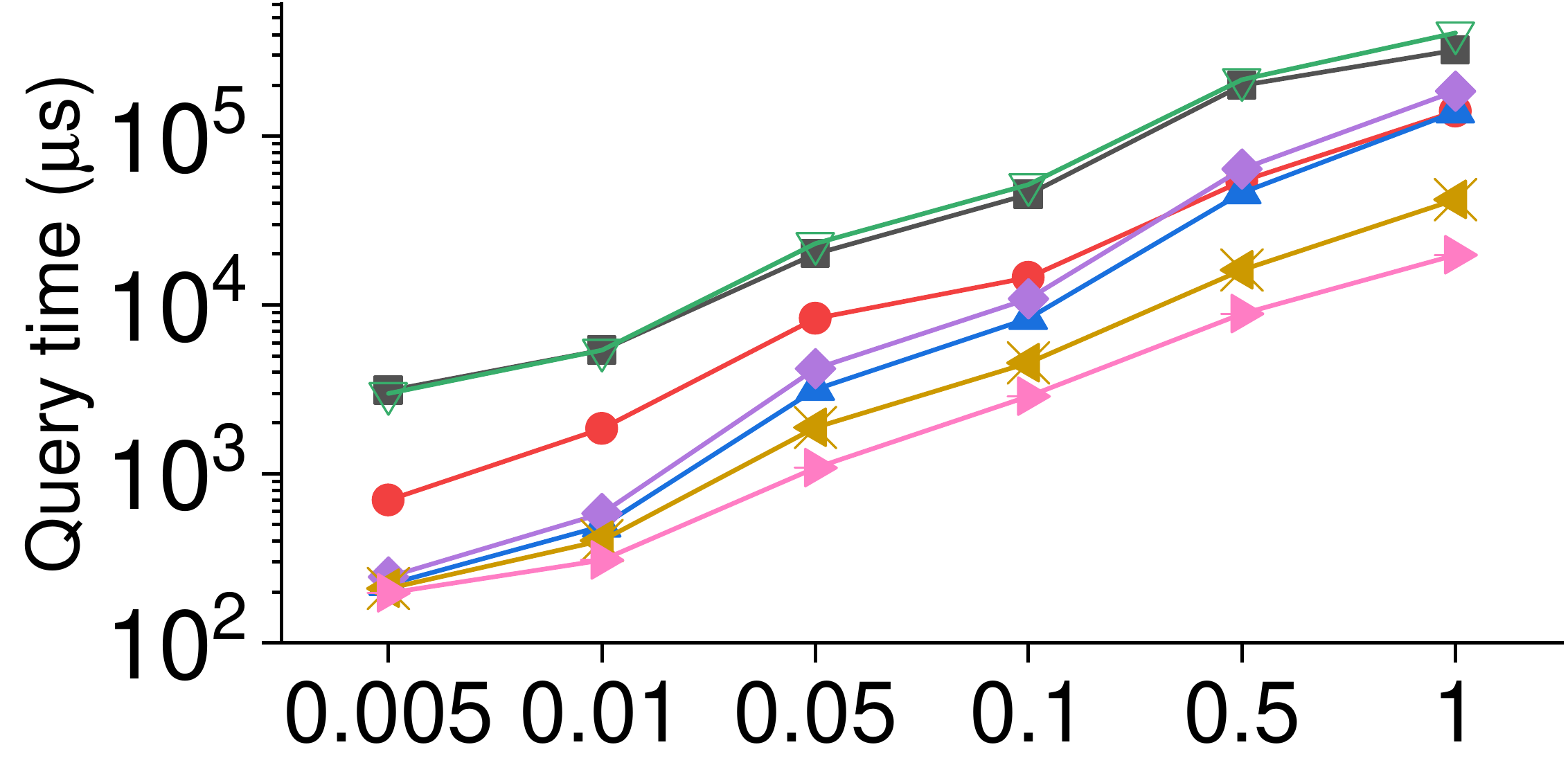}
        }
    \end{minipage}
    \vspace{-0.3cm}
    \caption{Varying the query region size}
    \label{exp:size}
\end{figure}
\subsubsection{\textbf{Effect of number of query keywords.}}
We evaluate the query sets with different numbers of keywords. Figure \ref{exp:keys} shows that the query time of all indexes grows with the query keyword set size.
The reason is that with the increase in the number of query keywords, more candidates need to be verified after the filtering step.
Besides, \idxname\ consistently outperforms other baseline indexes, and its cost grows much slower than those of others, e.g., the increased time of \idxname\ on \textbf{BPD} is around 100 $\mu$s while those of Flood-T and ST2I are both over 250 $\mu$s. Hence, compared to other indexes, \idxname\ is less sensitive to the number of query keywords.
\begin{figure}[tb]
    \begin{minipage}{.55\linewidth}
        \centering
        \includegraphics[width=\linewidth]{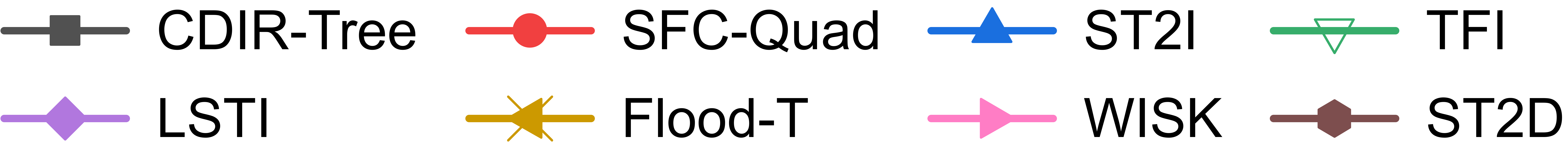}
    \end{minipage}
    \begin{minipage}{\linewidth}
        \centering
        \subcaptionbox{FS\label{keys-fs}}{
            \vspace{-0.2cm}
            \includegraphics[width=.23\linewidth]{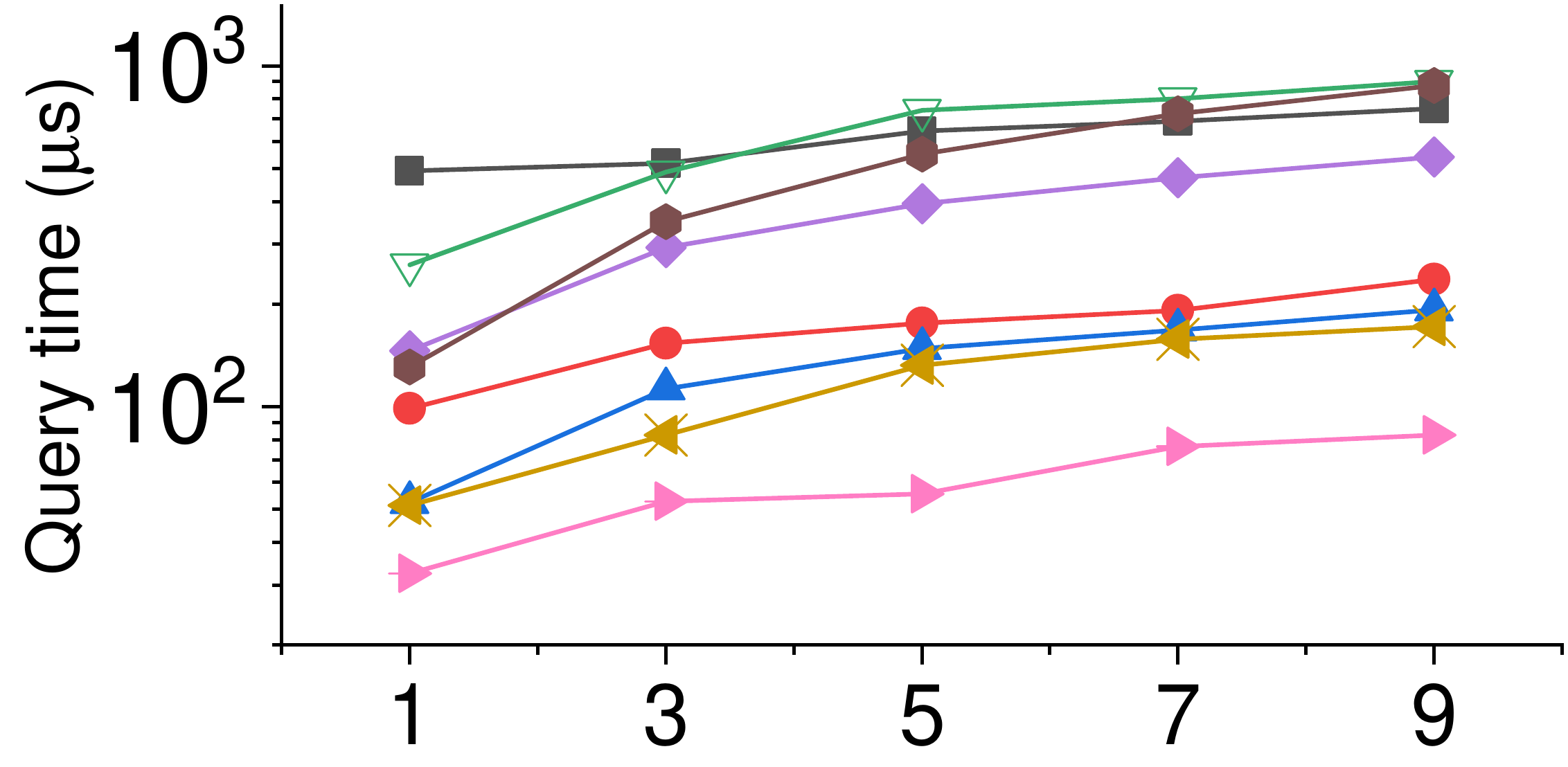}
        }
        \subcaptionbox{SP\label{keys-sp}}{
            \vspace{-0.2cm}
            \includegraphics[width=.23\columnwidth]{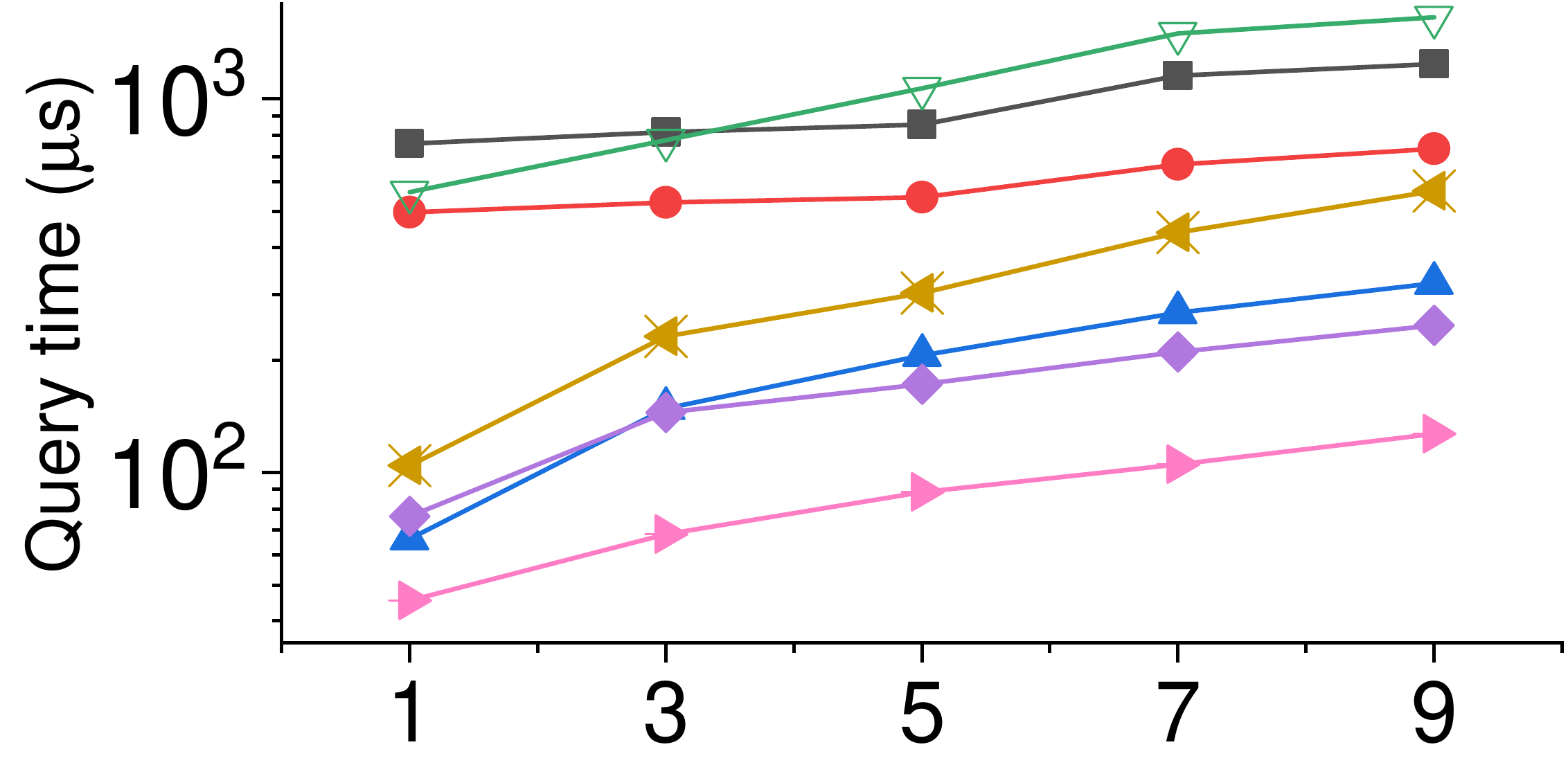}
        }
        \subcaptionbox{BPD\label{keys-bpd}}{
            \vspace{-0.2cm}
            \includegraphics[width=.23\columnwidth]{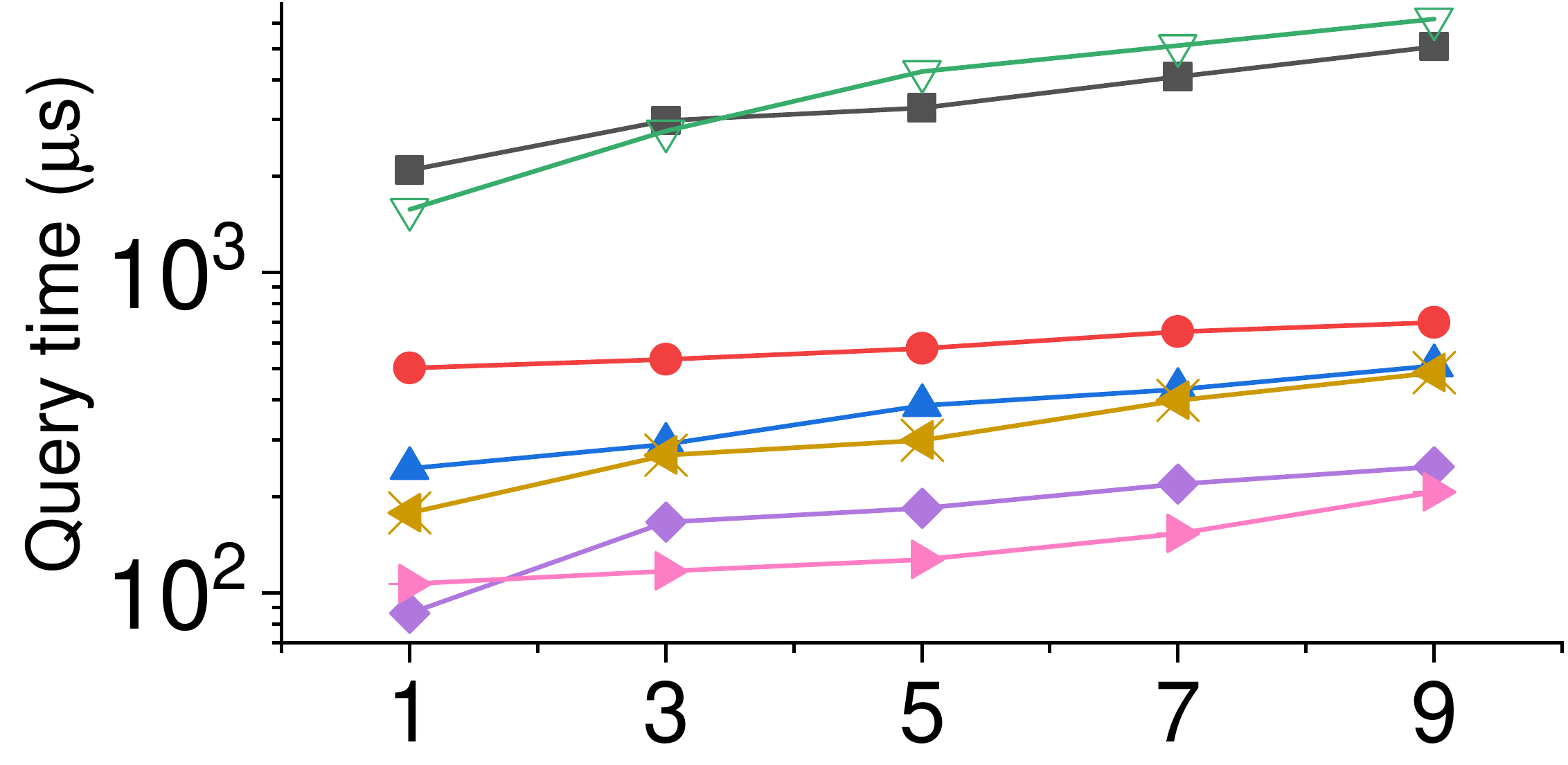}
        }
        \subcaptionbox{OSM\label{keys-osm}}{
            \vspace{-0.2cm}
            \includegraphics[width=.23\columnwidth]{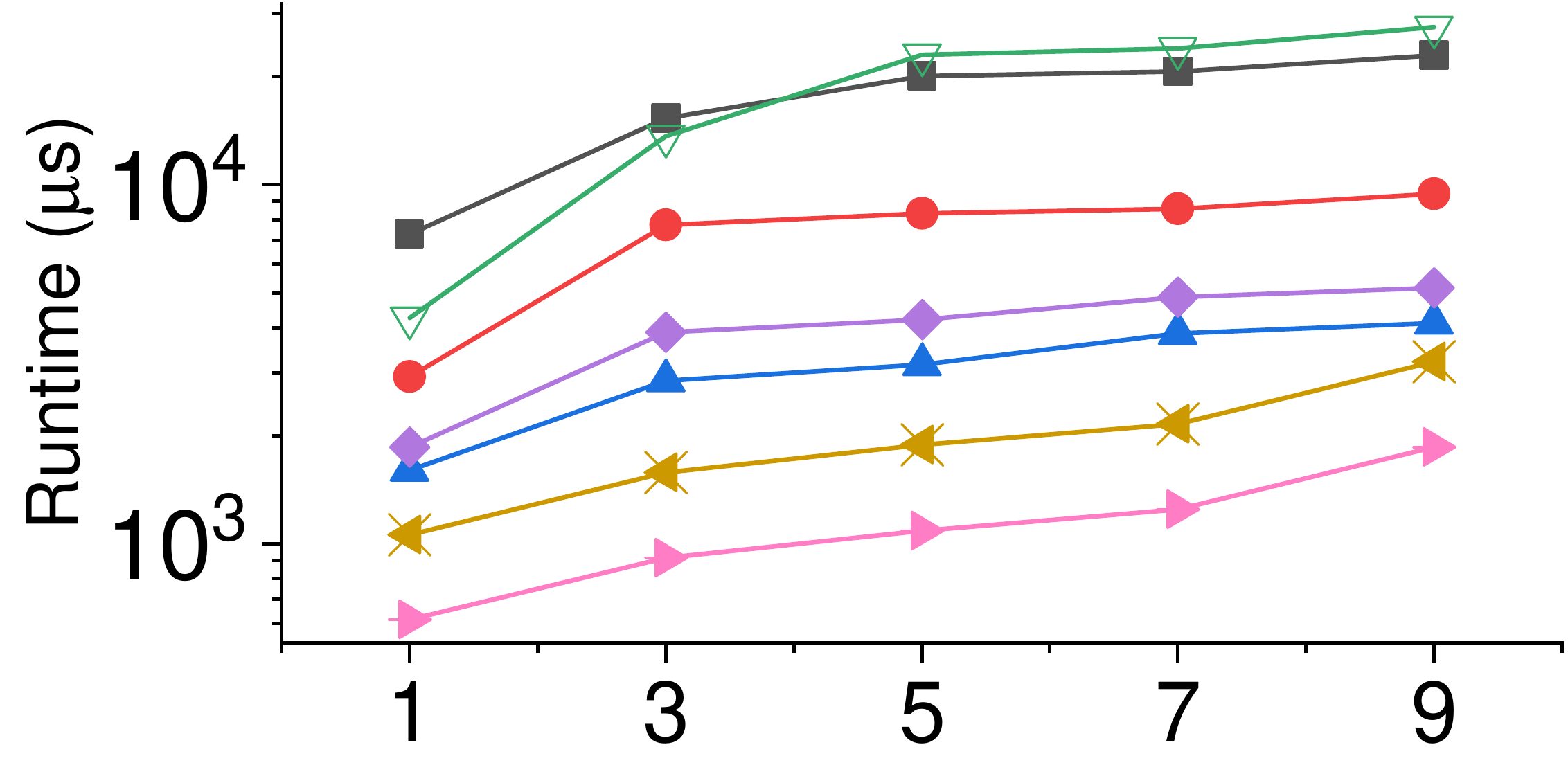}
        }
    \end{minipage}
    \vspace{-0.3cm}
    \caption{Varying no. of query keywords}
    \label{exp:keys}
\end{figure}
\subsubsection{\textbf{Scalability.}}
We generate five sub-datasets of \textbf{OSM} containing from 1 to 100 million objects and run experiments on these sub-datasets. We choose ST2I, LSTI, and Flood-T as our baselines. As shown in Figure \ref{exp:scale}, the query processing time increases with the size of the dataset, but \idxname\ performs more stable.
\begin{figure}[htb]
    \begin{minipage}{.49\linewidth}
        \setcaptionwidth{2in}
        \centering
        \includegraphics[width=.7\linewidth]{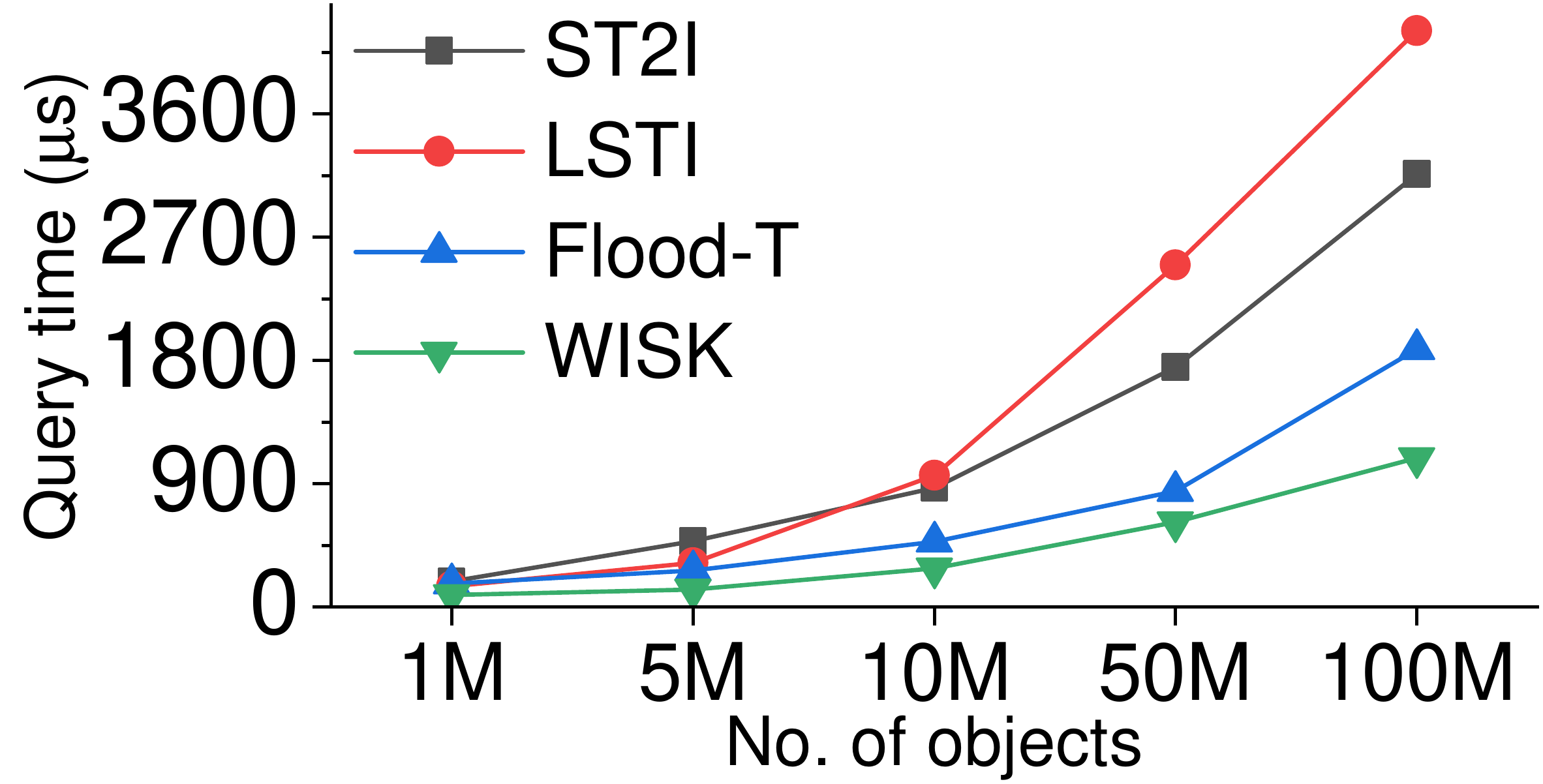}
        \abovecaptionskip 0.1cm
        \caption{Comparison of performance varying dataset size (num records)}
        \label{exp:scale}
    \end{minipage}
    \begin{minipage}{.49\linewidth}
        \setcaptionwidth{2in}
        \centering
        \includegraphics[width=.7\linewidth]{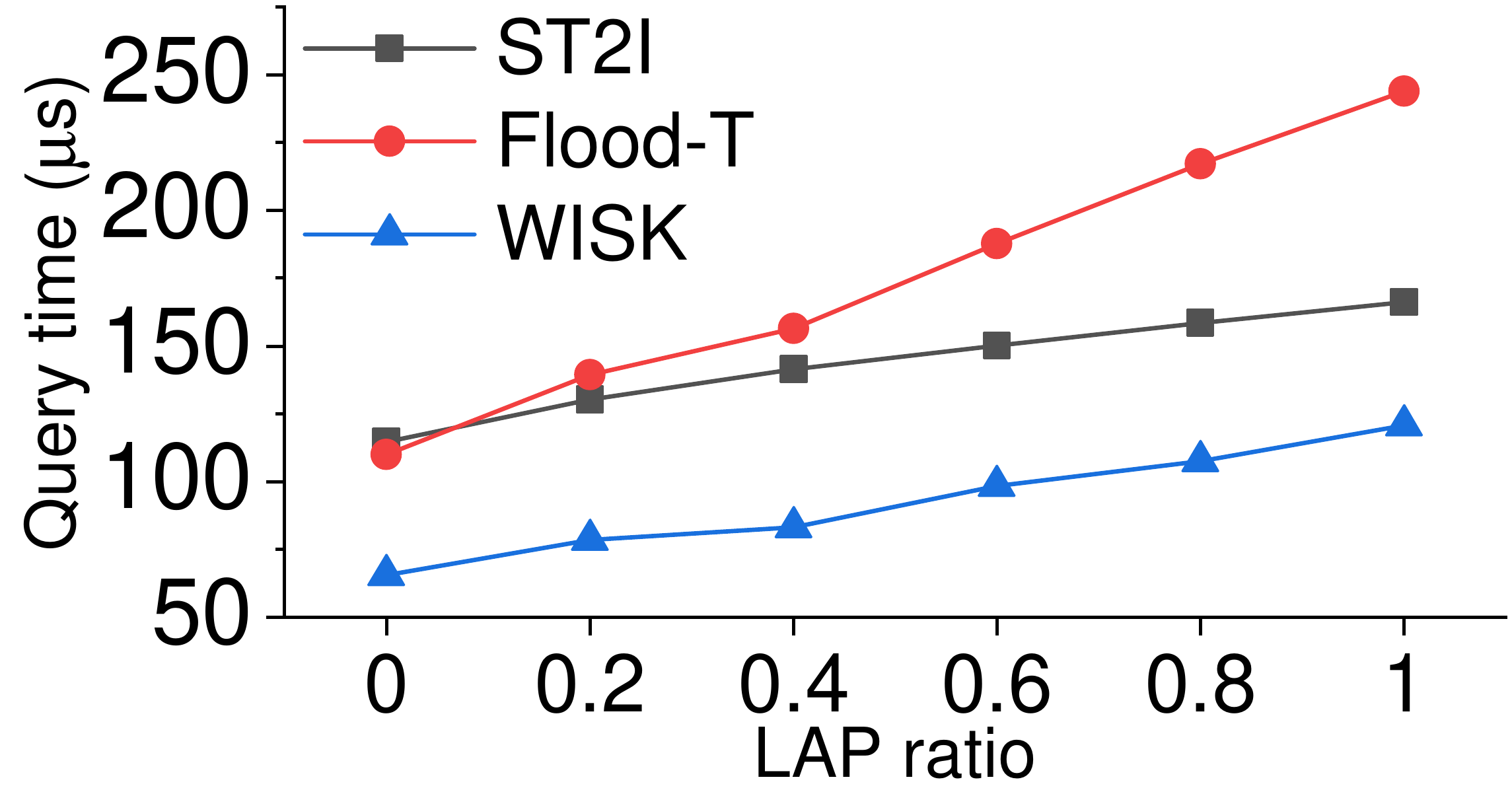}
        \abovecaptionskip 0.1cm
        \caption{Comparison of performance when changing query distribution}
        \label{exp:robust}
    \end{minipage}
\end{figure}
\subsubsection{\textbf{Robustness.}}\label{sec:robust}
We evaluate the performance of ST2I, Flood-T, and \idxname\ when the query distribution changes on \textbf{FS}. We initially train Flood-T and \idxname\ based on the query workload with \textbf{UNI} distribution. Then, we keep the index consistent and adjust the ratio of queries with \textbf{LAP} distribution from 0.2 to 1.0 in the testing query set.
As shown in Figure \ref{exp:robust}, the performance of query-aware indexes becomes worse when query distribution is more different from the training one. However, it can be seen that \idxname\ is more robust than Flood-T due to its improved partitioning algorithm and the bottom-up packing process. Additionally, the query time of ST2I also increases since it ignores the query knowledge when building the index, but the fluctuation is less than the one of Flood-T.

\subsection{Index Size \& Construction}
\subsubsection{\textbf{Index Sizes.}}
Table \ref{exp:space} reports the index sizes. Overall, \idxname\ costs less space than conventional indexes but is comparable to that of the best adapted learned indexes.
In particular, the size of CDIR-Tree is larger than those of the others since each of its nodes has an inverted file. For query efficiency, we have not compressed SFC-Quad, which leads to a larger size. ST2I has the smallest size among the conventional indexes.
Among the learned indexes, the sizes of TFI are the largest, as it uses inverted files.
\idxname\ takes more space than Flood-T on a small dataset since the number of its bottom clusters is similar to the number of the columns of Flood-T, but we build a hierarchical index.
However, on larger datasets, \idxname\ needs less space cost, since Flood-T splits more columns for better performance and builds inverted files for them.
\begin{table}[htb]
\small
\caption{Index structure size}
\vspace{-0.3cm}
\begin{tabular}{|c|c|c|c|c|}
\hline
Index   & \textbf{FS} & \textbf{SP} & \textbf{BPD} & \textbf{OSM} \\ \hline
CDIR-Tree & 2002MB      & 3571MB          & 33.15GB      & 108.45GB     \\ \hline
SFC-Quad   & 1406MB      & 2568MB          & 15.65GB      & 58.71GB      \\ \hline
ST2I    & 761MB       & 1554MB          & 15.18GB      & 56.05GB      \\ \hline
TFI     & 573MB       & 1423MB          & 8.86GB       & 32.05GB      \\ \hline
\textcolor{edit}{LSTI}     & \textcolor{edit}{642MB}       & \textcolor{edit}{1073MB}          & \textcolor{edit}{8.85GB}       & \textcolor{edit}{8.09GB}      \\ \hline
Flood-T & 400MB       & 937MB           & 7.15GB       & 27.94GB      \\ \hline
\idxname    & 483MB       & 980MB           & 7.02GB       & 25.78GB      \\ \hline
\end{tabular}
\label{exp:space}
\end{table}

\subsubsection{\textbf{Index Construction Time.}}
We compare the efficiency of index construction algorithms and report the results in Table \ref{exp:construction}.
It takes the minimum time to build SFC-Quad and ST2I on small datasets.
However, the time cost of ST2I significantly increases when the dataset becomes larger, since ST2I is built based on the set of converted points, and its time cost is positively correlated to the total number of keywords.
CDIR-Tree takes the highest time cost because it inserts the objects sequentially.

\begin{table}[htb]
\small
\caption{Index construction time}
\vspace{-0.3cm}
\begin{tabular}{|c|c|c|c|c|}
\hline
Index                                                            & \textbf{FS} & \textbf{SP} & \textbf{BPD} & \textbf{OSM} \\ \hline
CDIR-Tree                                                          & 391 sec     & 490 sec         & 56.17 min    & 196.87 min   \\ \hline
SFC-Quad                                                         & 20 sec      & 30 sec          & 3.35 min     & 9.18 min     \\ \hline
ST2I                                                             & 19 sec      & 29 sec          & 6.55 min     & 26.23 min    \\ \hline
TFI                                                              & 125 sec     & 283 sec         & 33.75 min    & 143.07 min   \\ \hline
\textcolor{edit}{LSTI}                                           & \textcolor{edit}{23 sec}      & \textcolor{edit}{32 sec}          & \textcolor{edit}{4.16 min}     & \textcolor{edit}{16.32 min}    \\ \hline
\textcolor{edit}{Flood-T}                                        & \textcolor{edit}{188 sec}     & \textcolor{edit}{974 sec}         & \textcolor{edit}{19.66 min}    & \textcolor{edit}{25.97 min}    \\ \hline
\textcolor{edit}{\idxname}                                       & \textcolor{edit}{353 sec}     & \textcolor{edit}{1216 sec}        & \textcolor{edit}{55.28 min}    & \textcolor{edit}{65.37 min}    \\ \hline
\begin{tabular}[c]{@{}c@{}}\textcolor{edit}{\idxname}\\ \textcolor{edit}{(Accelerated)}\end{tabular} & \textcolor{edit}{131 sec}     & \textcolor{edit}{547 sec}         & \textcolor{edit}{12.18 min}    & \textcolor{edit}{17.43 min}    \\ \hline
\end{tabular}
\label{exp:construction}
\end{table}

\textcolor{edit}{For the learned indexes, we report the training time. LSTI takes the least time to build because it only needs to scan the whole dataset twice. The time cost of TFI increases significantly when there are more different  keywords. For Flood-T and \idxname, we report the average time as the time costs of query-aware learned indexes usually increase with more query keywords. 
}

\begin{figure}[htb]
    \begin{minipage}{.55\linewidth}
        \centering
        \includegraphics[width=\linewidth]{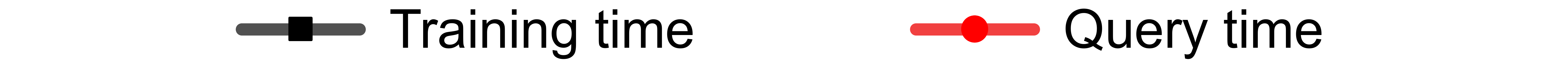}
    \end{minipage}
    \begin{minipage}{\linewidth}
        \centering
        \subcaptionbox{Sampling ratio (\%)\label{sample}}{
            \vspace{-0.2cm}
            \includegraphics[width=.35\columnwidth]{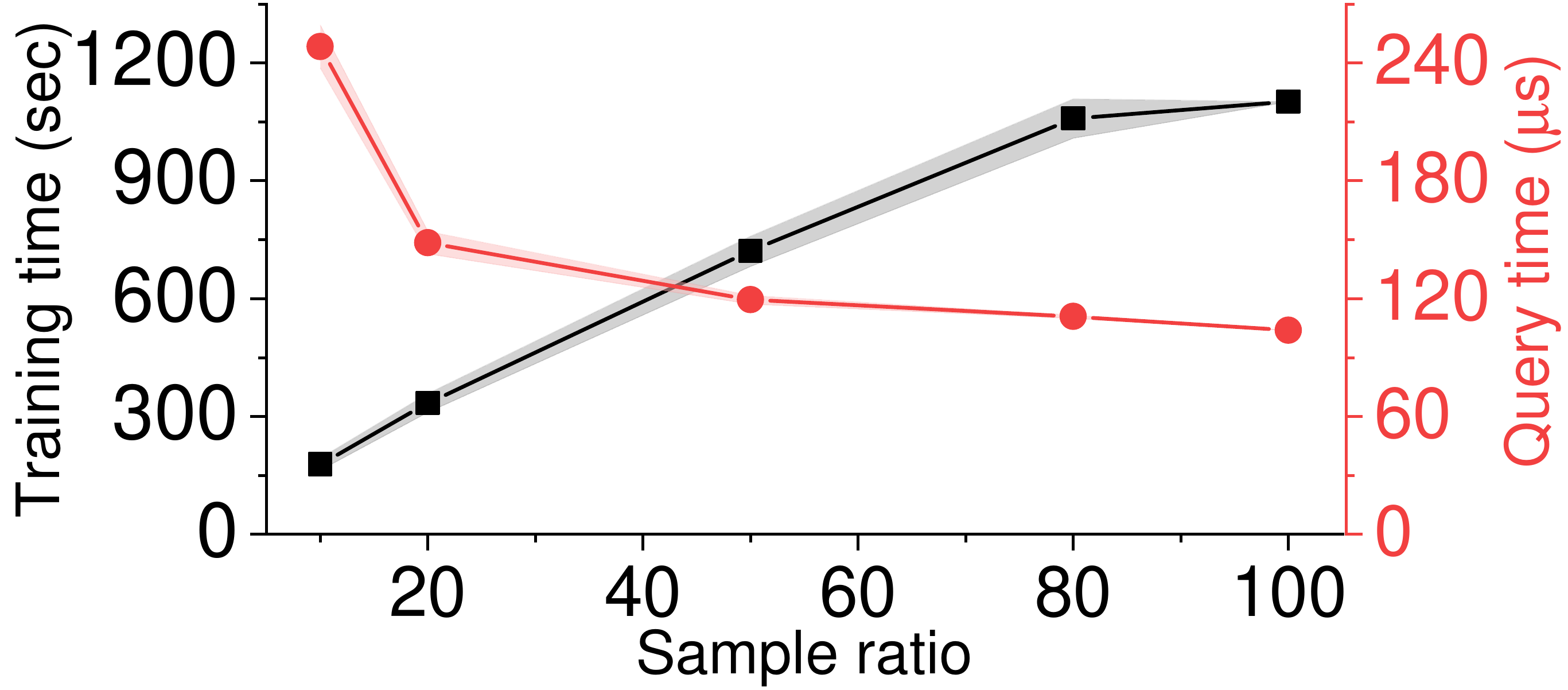}
        }
        \subcaptionbox{Clustering ratio (\%)\label{spectral}}{
            \vspace{-0.2cm}
            \includegraphics[width=.35\columnwidth]{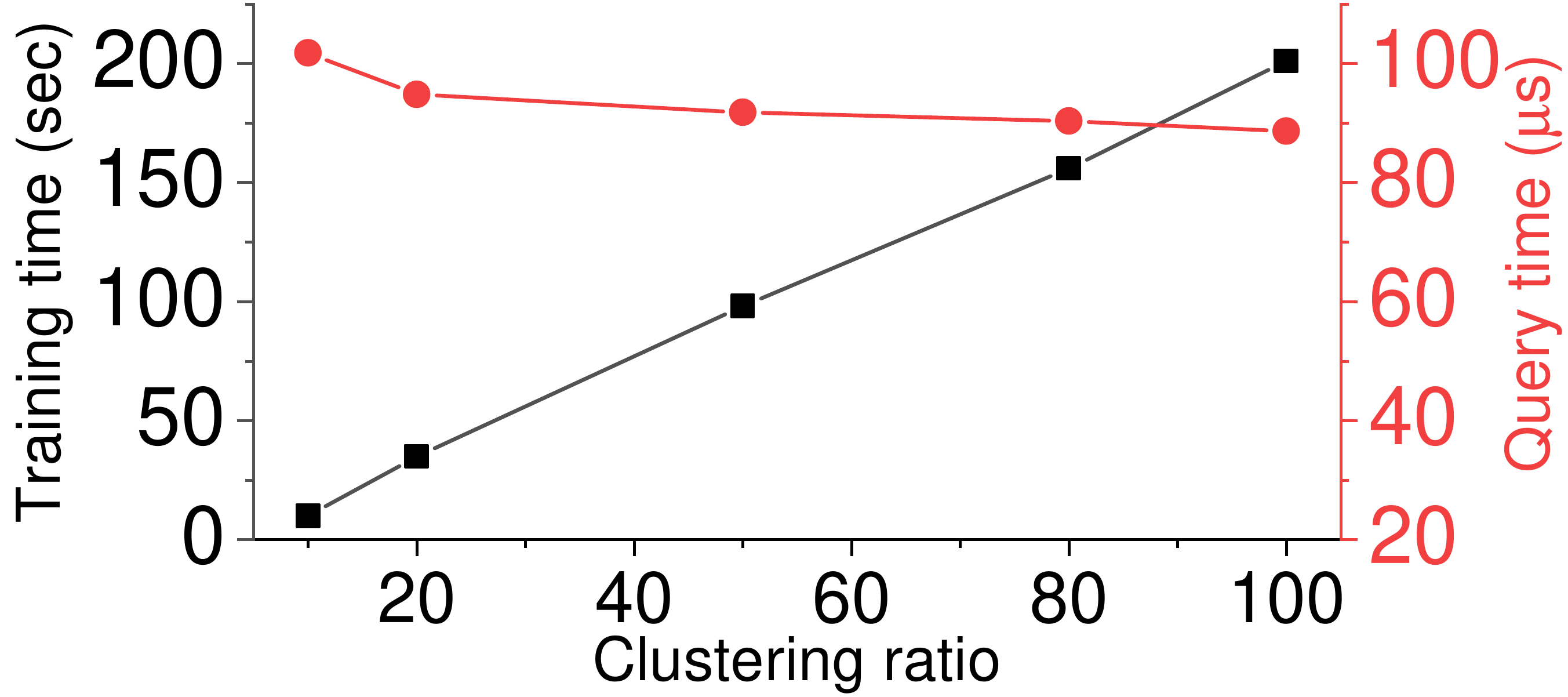}
        }
    \end{minipage}
    \vspace{-0.3cm}
    \caption{Training time and resulting query time on \textbf{SP}}
    \label{fig:speed-up}
\end{figure}

\textcolor{edit}{We designed two training time acceleration techniques as presented in Section~\ref{sec6}. We report the training and query times of \idxname with different sampling ratios in Figure~\ref{sample}. The result of each sampling ratio is an average of 10 runs. While the training time decreases by 72\%, we do not observe a large drop in query performance with a sample of only 30\% of the full query workload. We also observe that the standard deviation (represented by the width of the bands) of the training and query times of \idxname is consistently small for all sampling ratios. This demonstrates that \idxname has a stable performance using stratified sampling.
We vary the clustering ratio, i.e., the number of groups obtained over the number of bottom clusters, to balance the training and querying time. Figure~\ref{spectral} shows that even when the number of bottom clusters decreases by 80\%, the query time of  \idxname still only changes slightly.
We set the sampling ratio to 30\% and the clustering ratio to 20\% and generate the \textbf{Accelerated \idxname}. As shown in Table~4, \idxname has longer training times than the other learned indexes, but the acceleration techniques can reduce index training time up to 4 times while the query time is only affected marginally.}

\subsection{Index Update}
\subsubsection{\textbf{Dynamic Query Workload Changes.}}
To update the index when query distribution changes, we can retrain \idxname periodically following the former study \cite{DBLP:conf/sigmod/NathanDAK20}. 
\textcolor{edit}{We generate six workloads for \textbf{FS}. For each workload, we randomly select the query region size and the number of query keywords, and the query distribution adopts the default settings (\textbf{MIX}) and we randomly select the proportions of \textbf{UNI} and \textbf{LAP}. Each workload runs for 30 minutes and consists of 100 queries.
As Figure \ref{fig:query-retrain} shows, at the start of each 30-minute period, i.e., a new query workload starts, retraining \idxname is triggered, which happens in a separate thread and does not interrupt the query processing. While the index is being rebuilt, \idxname runs the new queries on its old layout, which explains the jumps in the figure. The retraining lasts about 3 minutes, and then \idxname switches to the new layout adapted to the new query workload. Thus, the query time drops back again.} 

\textcolor{edit}{
To capture minor changes in query distribution when retraining, we propose to apply incremental updates to the original index, in parallel to the retraining process. We locate the bottom clusters that are affected by the new queries, re-partition these clusters if the query costs can be reduced using new queries, and then insert the new clusters back into the non-leaf nodes they previously belonged to. 
The incremental updates may also help reduce the query times, which explains the multiple drops (e.g., at 00:30 and 01:30).}

\textcolor{edit}{Figure \ref{fig:query-retrain} also indicates the necessity of learning from the query workload. When a new workload arrives, the performance of \idxname drops due to its outdated layout. Re-learning the layouts based on the new workload mitigates the impact of the changing query distribution. The other two indexes do not utilize the queries during construction, which leads to much worse performance than \idxname (e.g. at 01:00, 02:00, and 02:30) on the skewed query workloads (i.e. high proportion of LAP).}
\begin{figure}[htb]
    \centering
    \includegraphics[width=.7\linewidth]{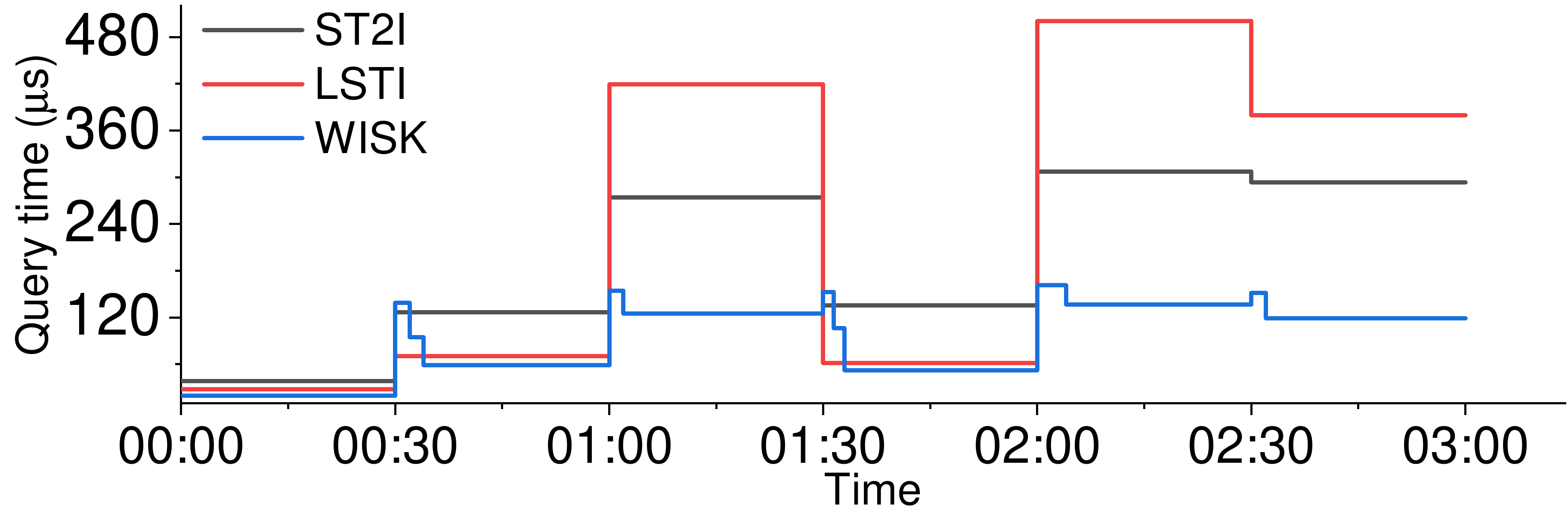}
    \vspace{-0.3cm}
    \caption{Impact of dynamic workload changes}
    \label{fig:query-retrain}
\end{figure}

\subsubsection{\textbf{Data Insertion.}}
\textcolor{edit}{\idxname also handles data insertion well. Given a new object \textit{o}, we can traverse \idxname to find the bottom cluster where \textit{o} falls. Next, we update the inverted file or the bitmap of the affected nodes to obtain an updated index. This simple process, however, cannot guarantee an optimal layout because the bottom clusters might need to be split after the insertion. Thus, we buffer the inserted objects and retrain our index when the buffer is full.}

\textcolor{edit}{We set the buffer size at 100,000 (around 20MB) and run experiments to examine the impact of data insertions. We randomly select 500,0000 objects from \textbf{FS} for the insertions. We insert 100,000 objects every 30 minutes. 
Figure~\ref{fig:data-insert} shows the performance of ST2I, LSTI, and \idxname. We compare with \idxname using the simple insertion process without retraining. It can be seen that the query time of all indexes increases when more objects are inserted. Between the two \idxname variants, we see that the query time of \idxname without retraining increases faster with more insertions, thus verifying the importance of retraining in improving the query time of \idxname under dynamic data settings. 
We also observe that the retraining process takes only 1 to 2 minutes each time, since only the affected bottom clusters need to be split, and the RL-based packing can inherit knowledge from the previous training process, i.e., the unaffected bottom clusters are initially packed into the previous corresponding upper nodes.}
\begin{figure}[htb]
    \centering
    \includegraphics[width=.7\linewidth]{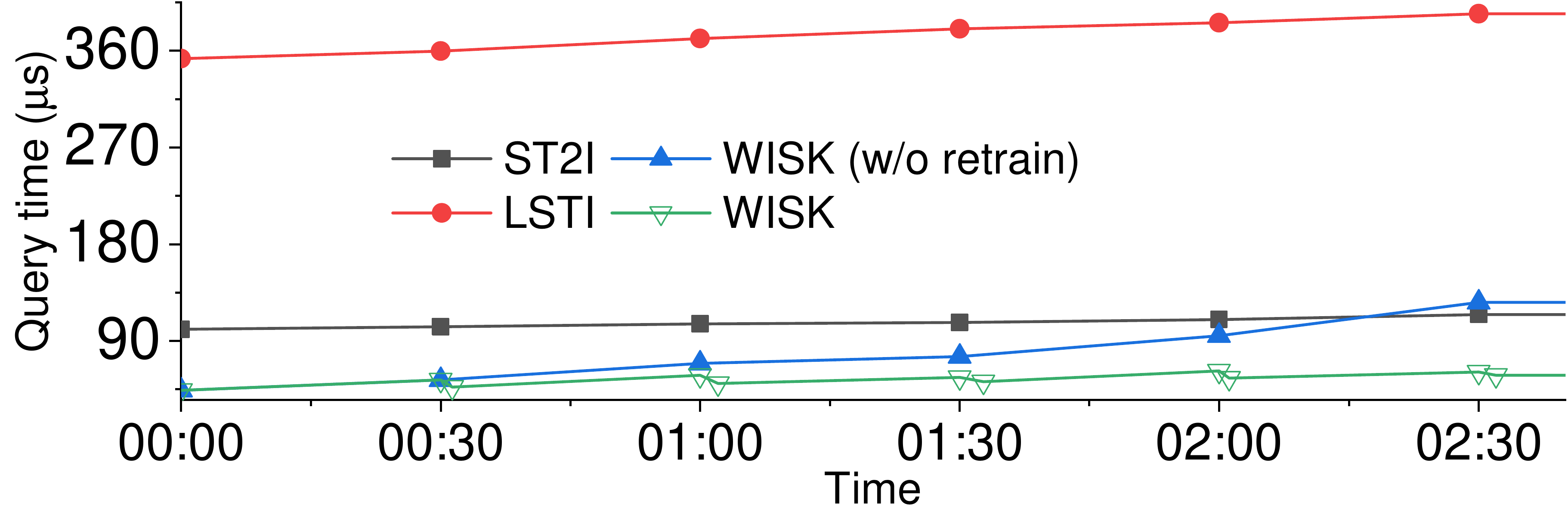}
    \vspace{-0.3cm}
    \caption{Impact of data insertion}
    \label{fig:data-insert}
\end{figure}

\subsection{Ablation Study}
\subsubsection{\textbf{RL-based Packing.}}
We conduct an experiment to compare the cost at the leaf level and that at the non-leaf level. Figure \ref{exp:leaf-inter} shows that the time at the leaf level dominates the query processing time, which occupies around 90\% of the total cost, verifying the way to define cost function is reasonable. In Figure \ref{exp:hier-comp}, we observe that packing our bottom clusters by directly using CDIR-Tree construction method might affect the query time because it may pack some leaf nodes intersecting with various queries.
\begin{figure}[htb]
    \begin{minipage}{.49\linewidth}
        \setcaptionwidth{2in}
        \centering
        \includegraphics[width=.7\linewidth]{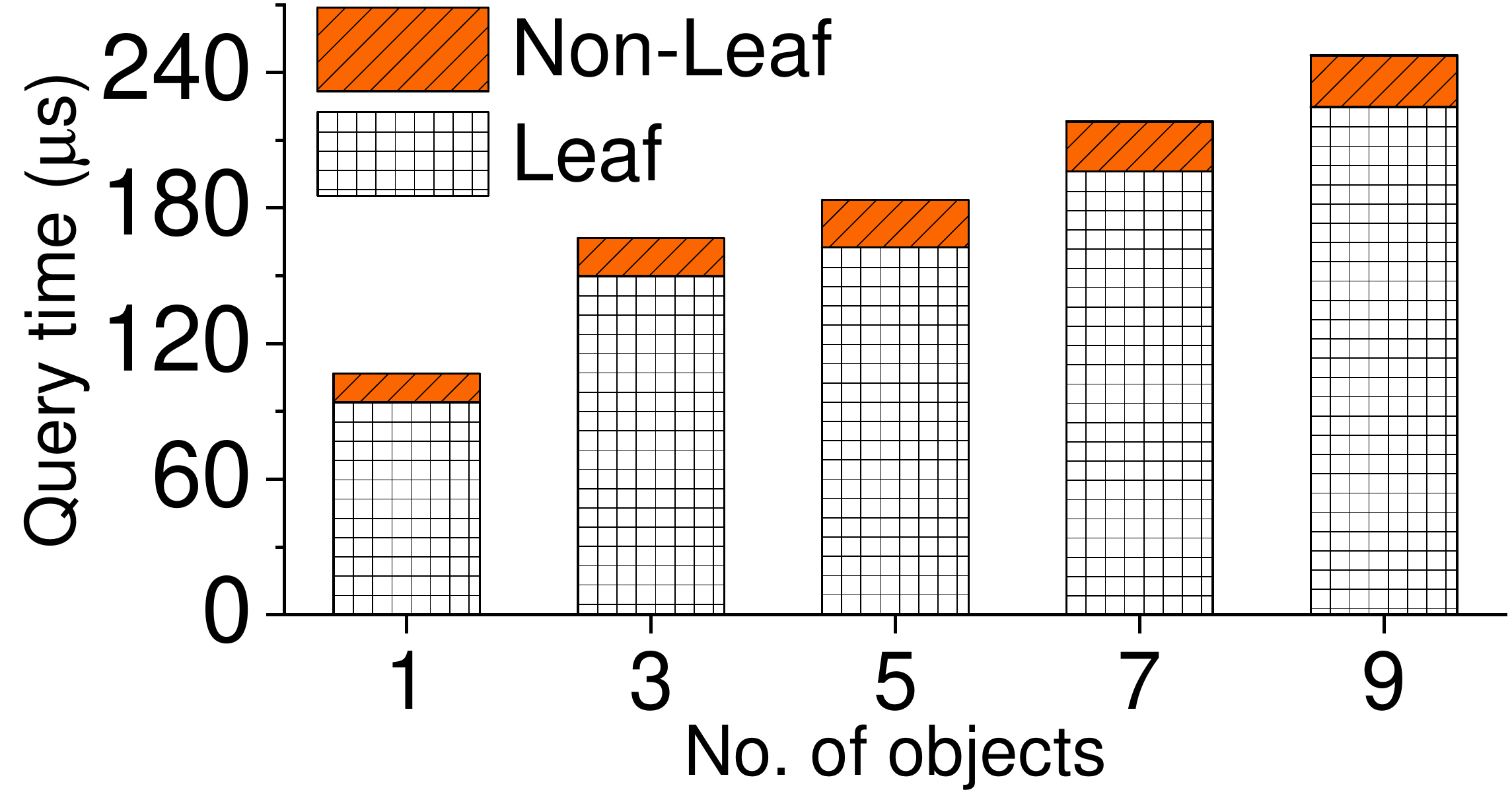}
        \abovecaptionskip 0.1cm
        \caption{Comparison of processing time}
        \label{exp:leaf-inter}
    \end{minipage}
    \begin{minipage}{.49\linewidth}
        \setcaptionwidth{2in}
        \centering
        \includegraphics[width=.7\linewidth]{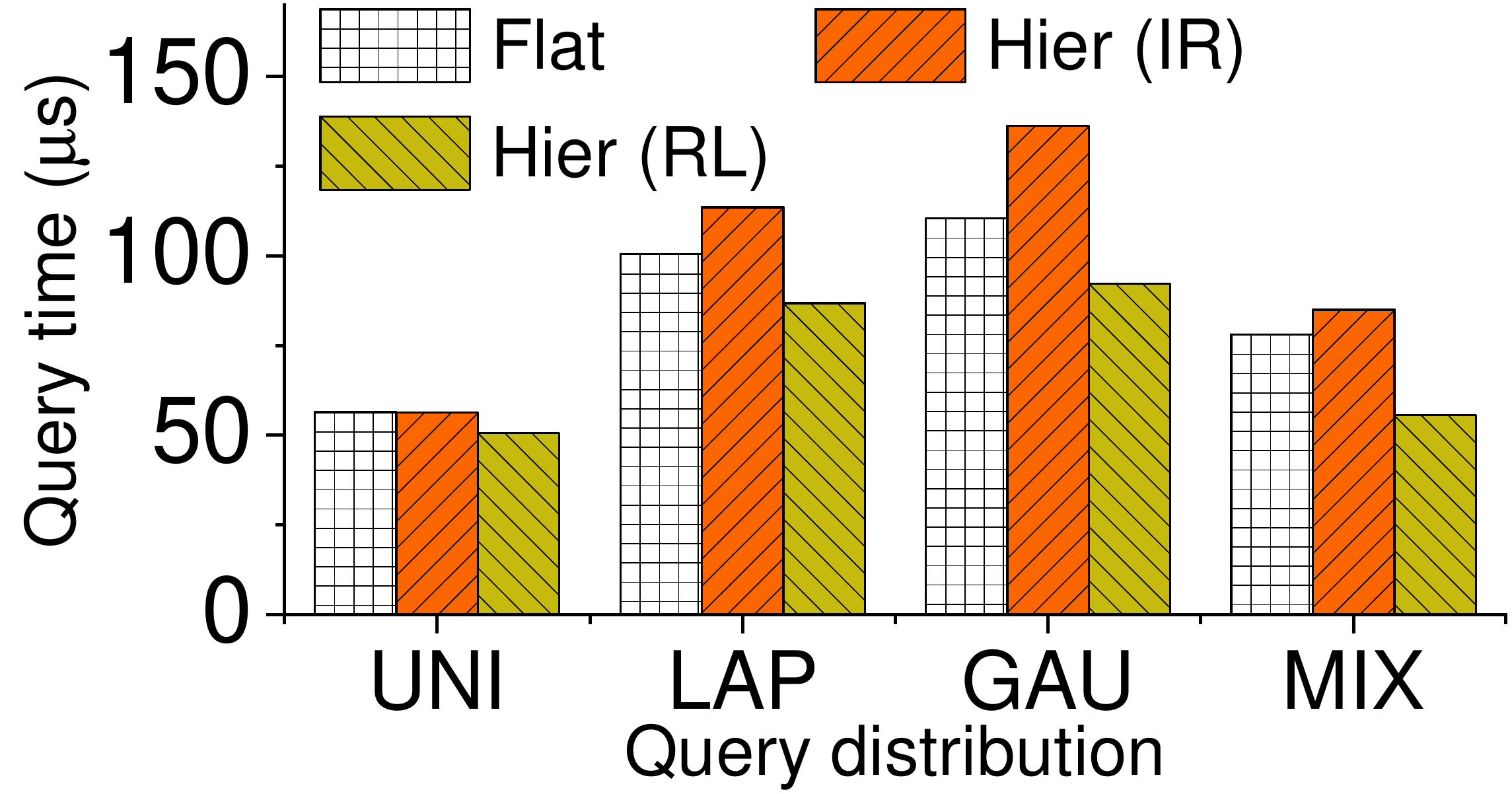}
        \abovecaptionskip 0.1cm
        \caption{Comparison of packing methods}
        \label{exp:hier-comp}
    \end{minipage}
\end{figure}

We evaluate the effectiveness of the bottom-up construction process. 
As shown in Figure \ref{rl-bpd-keys}, the improvement of the different number of keywords is similar. This is because the number of query keywords has little effect on the number of bottom clusters. Thus, the improvement is stable using this RL-based grouping algorithm.
\begin{figure}[htb]
    \begin{minipage}{.55\linewidth}
        \centering
        \includegraphics[width=\linewidth]{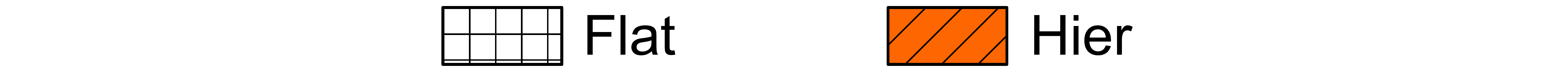}
    \end{minipage}
    \begin{minipage}{\linewidth}
        \centering
        \subcaptionbox{Varying no. of keywords\label{rl-bpd-keys}}{
            \vspace{-0.2cm}
            \includegraphics[width=.35\columnwidth]{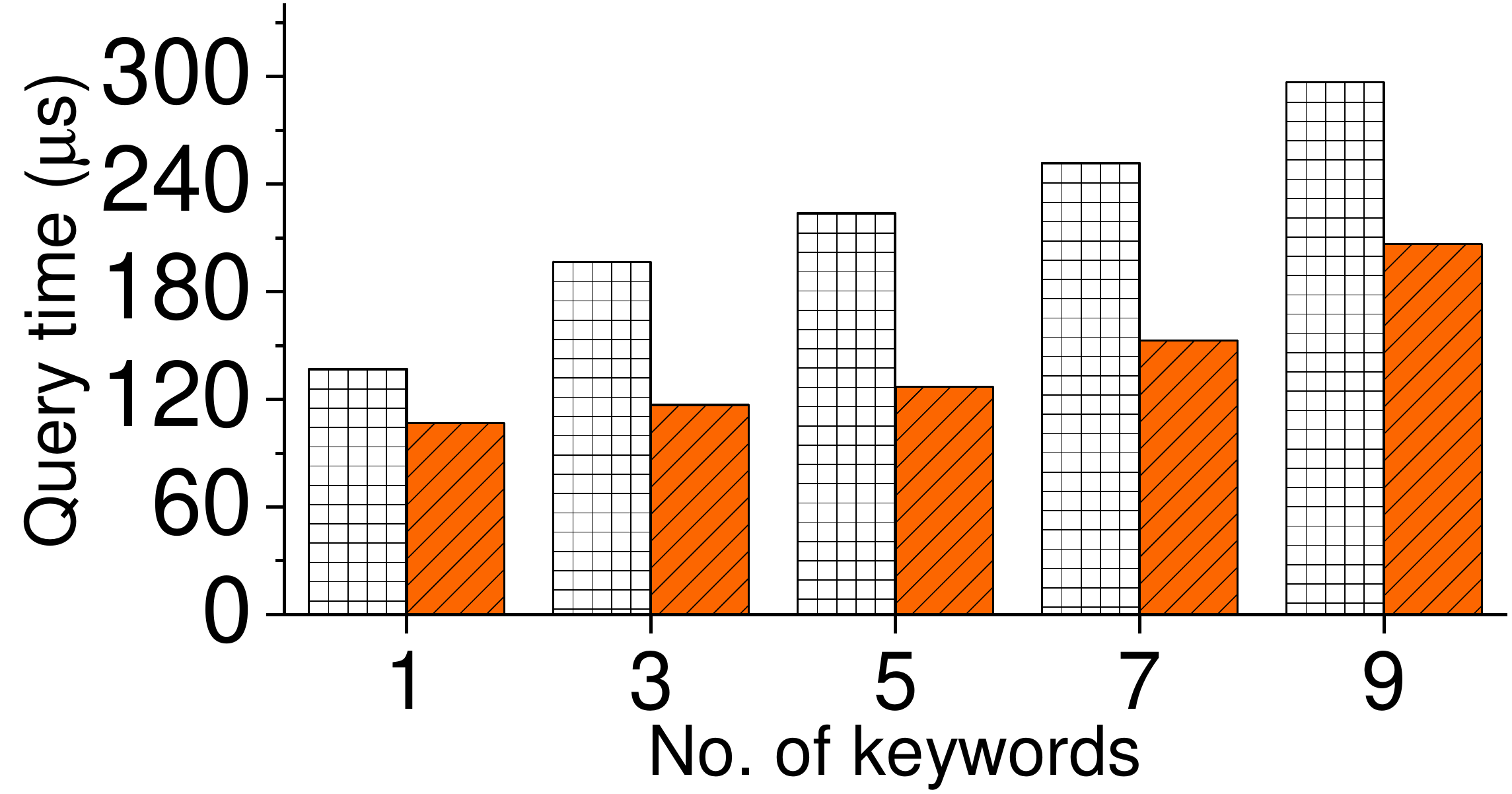}
        }
        \subcaptionbox{Varying query region size\label{rl-bpd-size}}{
            \vspace{-0.2cm}
            \includegraphics[width=.35\columnwidth]{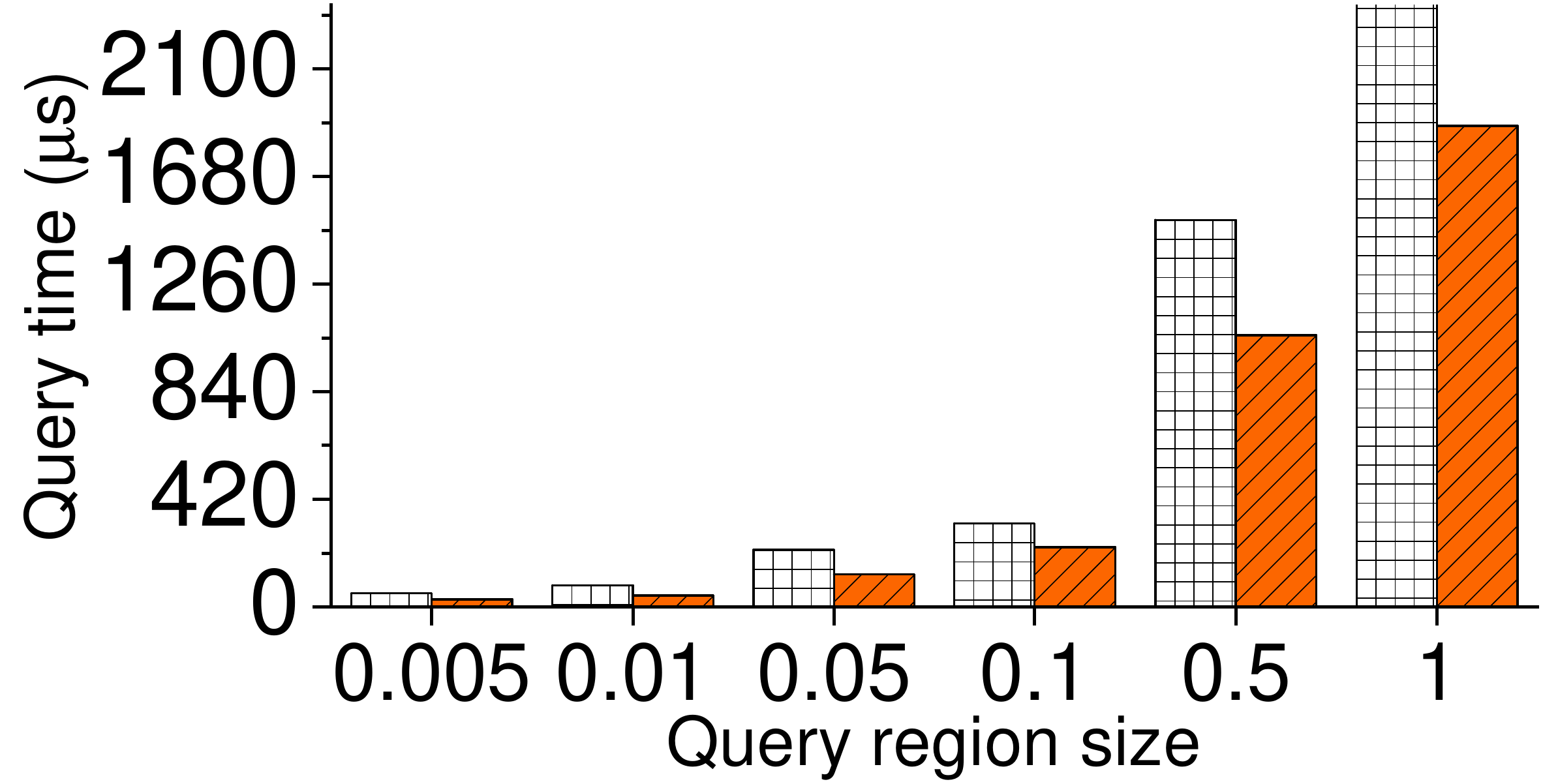}
        }
    \end{minipage}
    \vspace{-0.3cm}
    \caption{Comparison of index layouts on \textbf{BPD}}
    \label{fig:abl_rl}
\end{figure}

It can be seen from Figure \ref{rl-bpd-size} that the improvement of hierarchical indexes becomes more significant for queries of larger region sizes. This is because queries of larger region size correspond to a wider space covered by these queries and more bottom nodes such that the bottom-up packing process can reduce more filtering cost. However, we also observe that the improvement becomes stable as the region size continues to get larger, as the query region covers most of the data space.

\subsubsection{\textbf{CDF Model.}}
\textcolor{edit}{When generating bottom clusters, we use CDF models to estimate the number of objects sharing the same keywords inside a region. To reduce the parameters, we propose to use Gaussian functions and NNs for keywords with different frequencies. In Figure~\ref{gau-nn}, we compare our method with the settings in which only Gaussian models or NN models are used. Although the Gaussian-only method has the least training time, its estimation results are inaccurate, leading to much worse query time. In contrast, the NN-only method achieves the best query time, but it needs much more training time. In comparison, the proposed mixed method achieves similar query performance as the NN-only method without significantly increasing the training time.}

\begin{figure}[htb]
    \centering
    \subcaptionbox{Model selection\label{gau-nn}}{
        \vspace{-0.2cm}
        \includegraphics[width=.35\columnwidth]{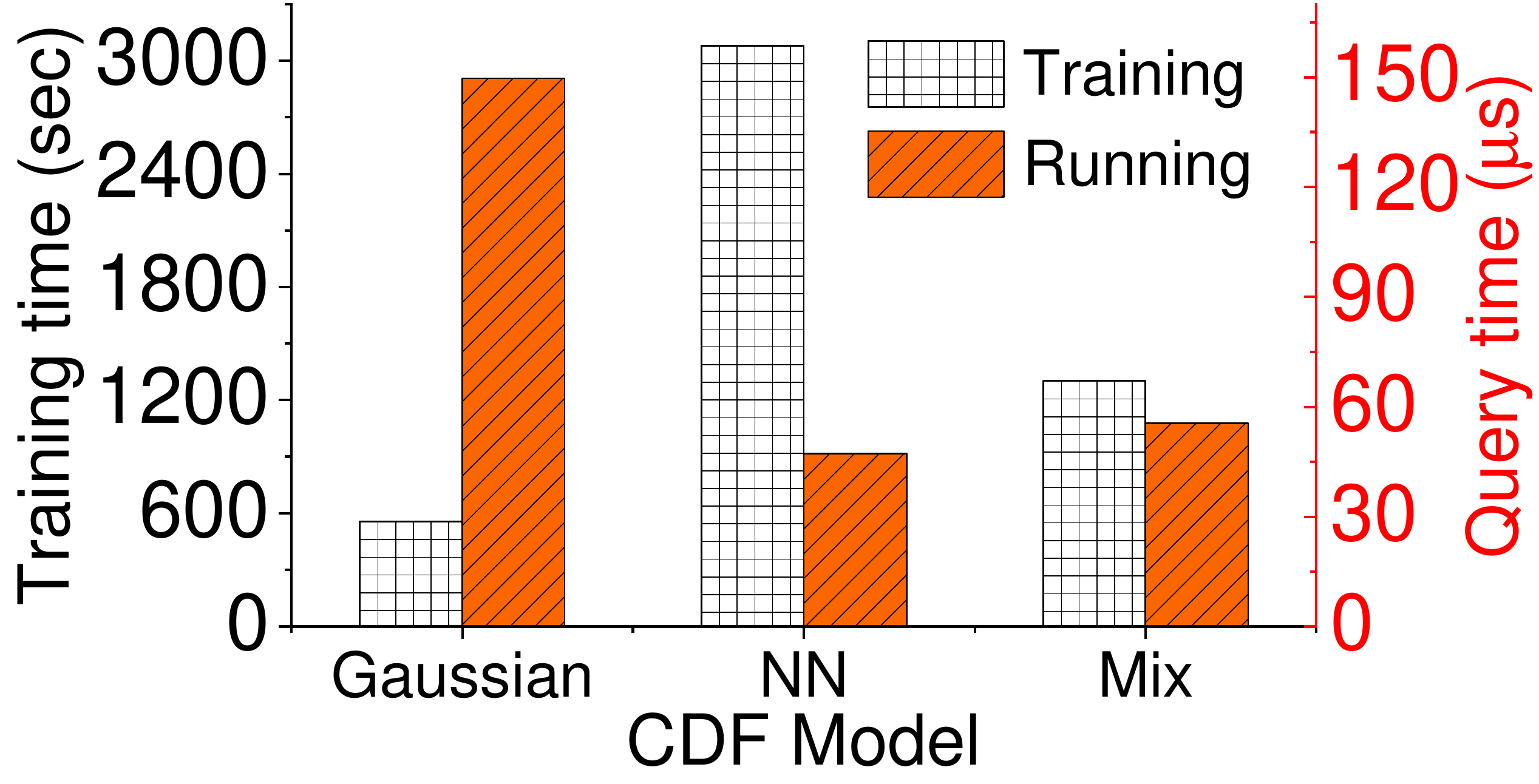}
    }
    \subcaptionbox{1D vs. 2D model \label{cdf-loss}}{
        \vspace{-0.2cm}
        \includegraphics[width=.35\columnwidth]{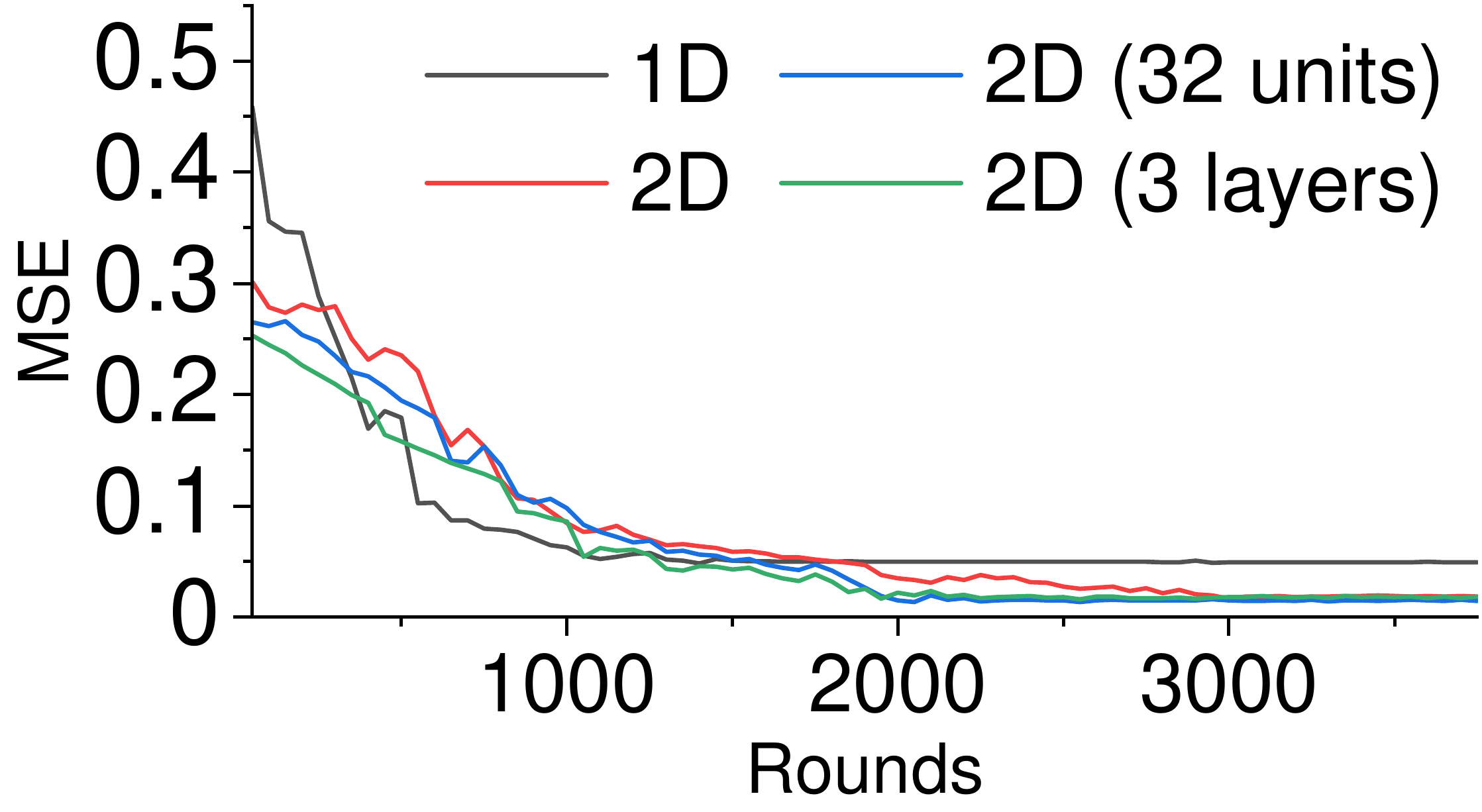}
    }
    \vspace{-0.3cm}
    \caption{Effect of different CDF settings}
    \label{fig:cdf-model}
\end{figure}
\textcolor{edit}{To speed up pre-processing, we assume that the two spatial dimensions are independent following the existing work \cite{DBLP:conf/sigmod/NathanDAK20}. We next study the impact of such an assumption.} 

\textcolor{edit}{
We observe that keywords with higher frequency have a stronger impact on the query time, and thus they need a more accurate CDF estimation. We run experiments with a randomly selected high-frequency keyword on \textbf{FS}.
We train two marginal (1D) models and a joint (2D) model for the selected keyword on the whole dataset. The 1D models and the 2D model all employ a neural network with 2 hidden layers and 16 hidden units. We also compare with two variants of the 2D model, one uses more hidden units (32 units) and the other uses more layers (3 layers). 
We randomly sample 1,000 rectangular query regions within the data space as testing data. For each query region, we compute the proportion of objects that fall within it as the ground truth. The product of the two 1D models and the output of the 2D models are used as the estimation results corresponding to 1D and 2D models, respectively. For every 50 training rounds, we calculate the mean squared error (MSE) between the estimations and the ground truth. Figure \ref{cdf-loss} shows that the 1D model converges much faster while its final loss is comparable to those of the 2D models. These results justify the use of 1D CDFs in our method.}


\subsubsection{\textbf{Frequent Itemset (FI)}}
The FI mining extracts the frequent interrelation among keywords. In this experiment, the minimum support is set to $0.01\text{\textperthousand}$ and the maximum size of target itemsets is equal to the number of query keywords.
Figure \ref{fig:abl_fi} shows the effect of FI on the index construction by query efficiency on \textbf{FS} and \textbf{BPD}. 
We see that the FI mining improves the performance of \idxname\ consistently when there is more than one query keyword.
Without the FI mining component, we learn models of each keyword separately where redundancies occur when there is more than one keyword. We also observe that this adaptation is more beneficial given more query keywords. In general, more query keywords lead to a higher possibility of resulting in redundancy since the probability of an object including more than one query keyword increases.

\begin{figure}[htb]
    \begin{minipage}{.55\linewidth}
        \centering
        \includegraphics[width=\linewidth]{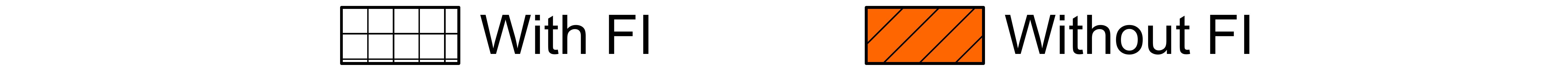}
    \end{minipage}
    \begin{minipage}{\linewidth}
        \centering
        \subcaptionbox{FS\label{fi-fs}}{
            \vspace{-0.2cm}
            \includegraphics[width=.35\columnwidth]{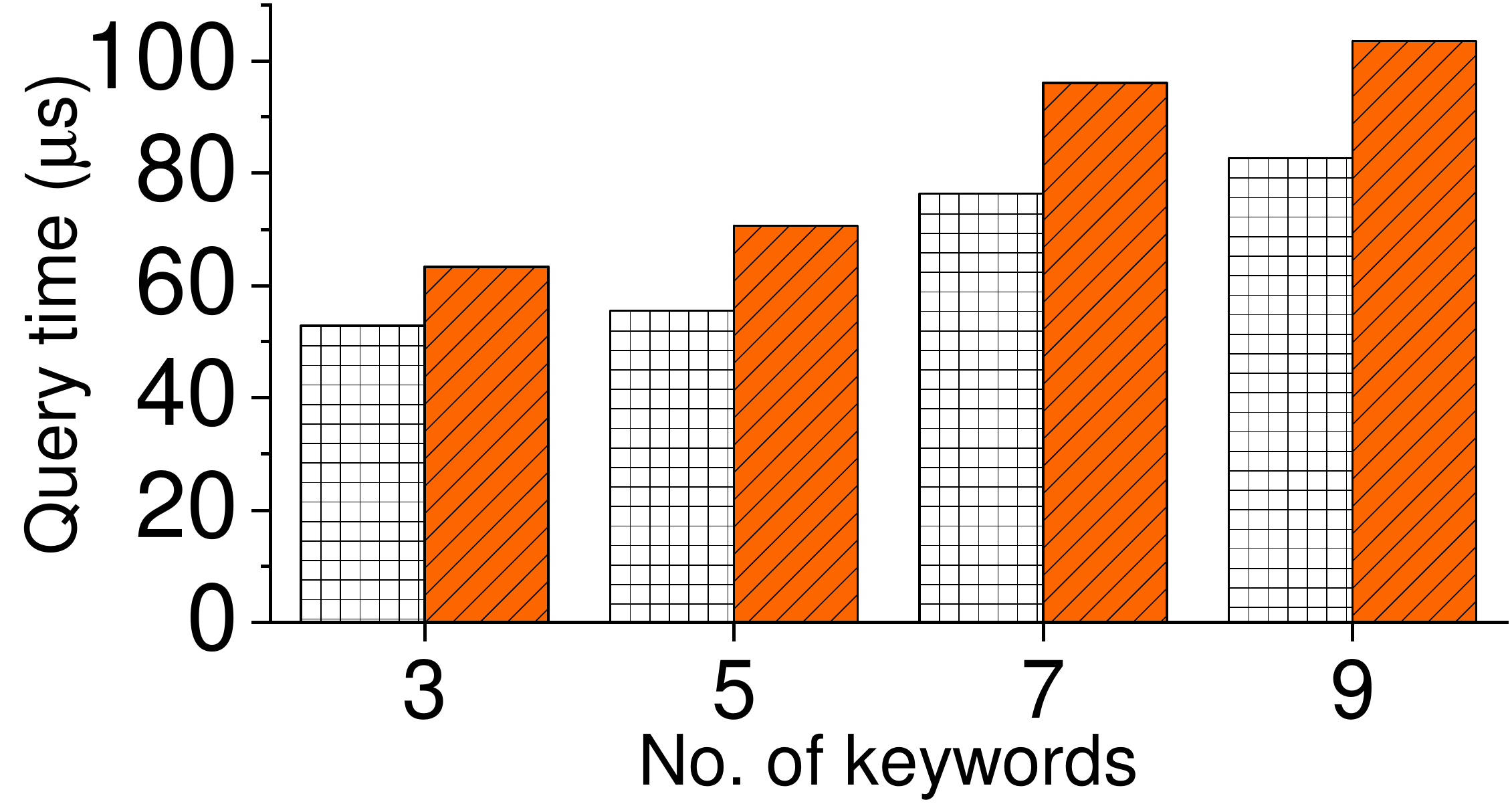}
        }
        \subcaptionbox{BPD\label{fi-bpd}}{
            \vspace{-0.2cm}
            \includegraphics[width=.35\columnwidth]{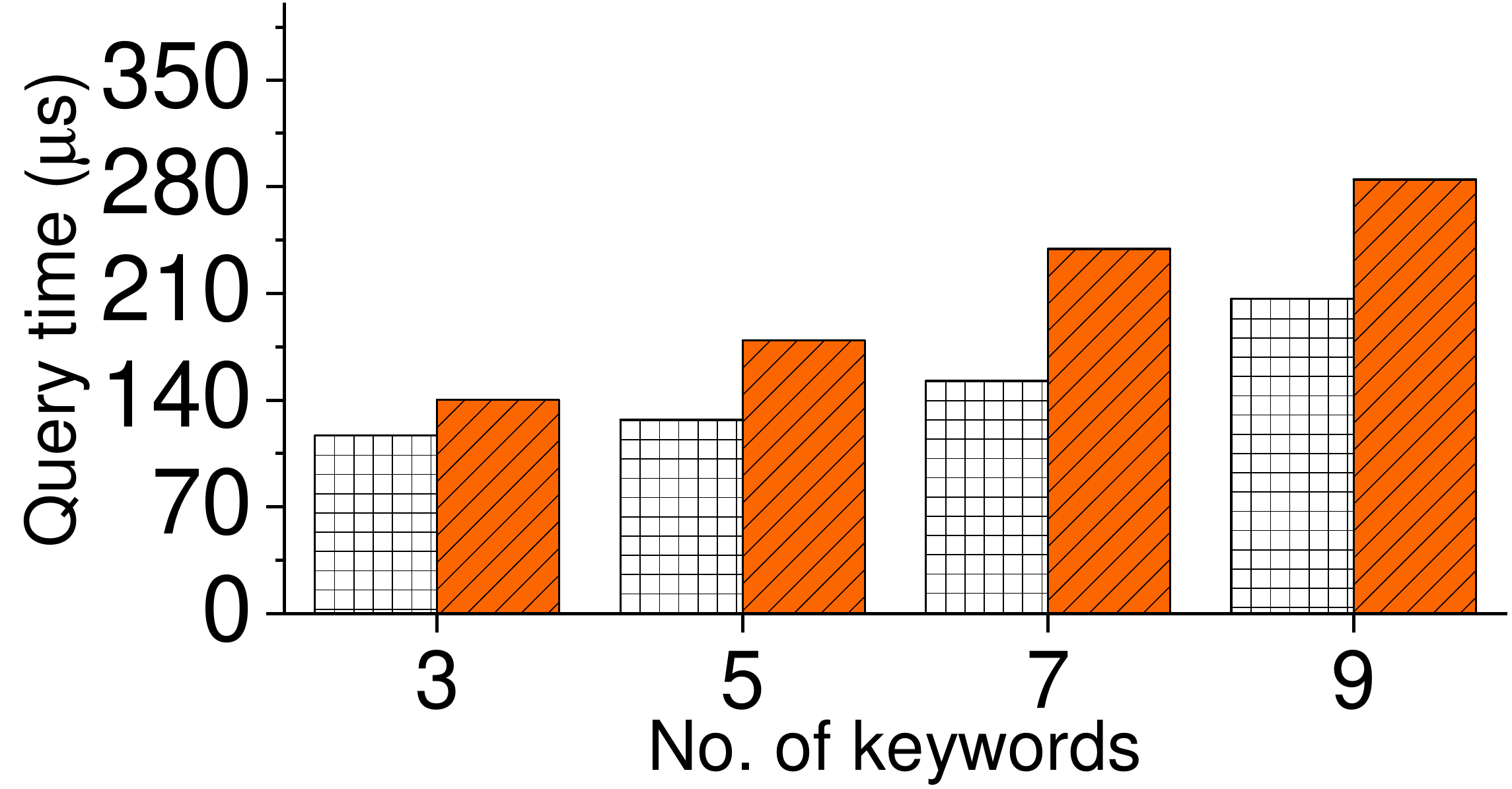}
        }
    \end{minipage}
    \vspace{-0.3cm}
    \caption{Effect of the frequent itemset}
    \label{fig:abl_fi}
\end{figure}
Besides, the performance improvement on \textbf{BPD} is more obvious than that on \textbf{FS}. This is because the number of distinct keywords on \textbf{FS} is much less, making the number of frequent itemsets less than others.
Additionally, the improvement becomes consistent when the number of query keywords is larger than a threshold. The reason is that each object includes finite keywords, and thus we cannot generate frequent itemsets with more keywords.
\subsubsection{\textbf{Action Mask.}}
When using the RL framework, the environment applies the action mask to reduce the action space. Here, we evaluate its effectiveness in two aspects. We use the SmoothL1Loss with the sum reduction as the loss function in our RL framework.
Figure \ref{rl-loss} shows that the RL framework with the action mask can speed up the model convergence and reach a smaller loss. In Figure \ref{rl-reduction}, we sum the total rewards and the number of bottom nodes in each epoch to show the reduction in the average number of accessed nodes. The result shows that the pruning capability is always better when applying the action mask.

\begin{figure}[htb]
    \begin{minipage}{.55\linewidth}
        \centering
        \includegraphics[width=\linewidth]{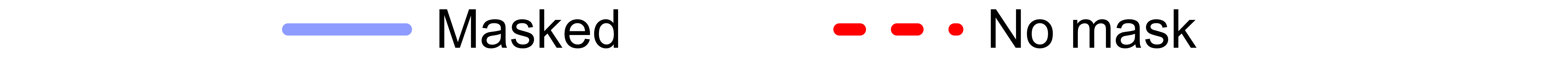}
    \end{minipage}
    \begin{minipage}{\linewidth}
        \centering
        \subcaptionbox{Loss\label{rl-loss}}{
            \vspace{-0.2cm}
            \includegraphics[width=.35\columnwidth]{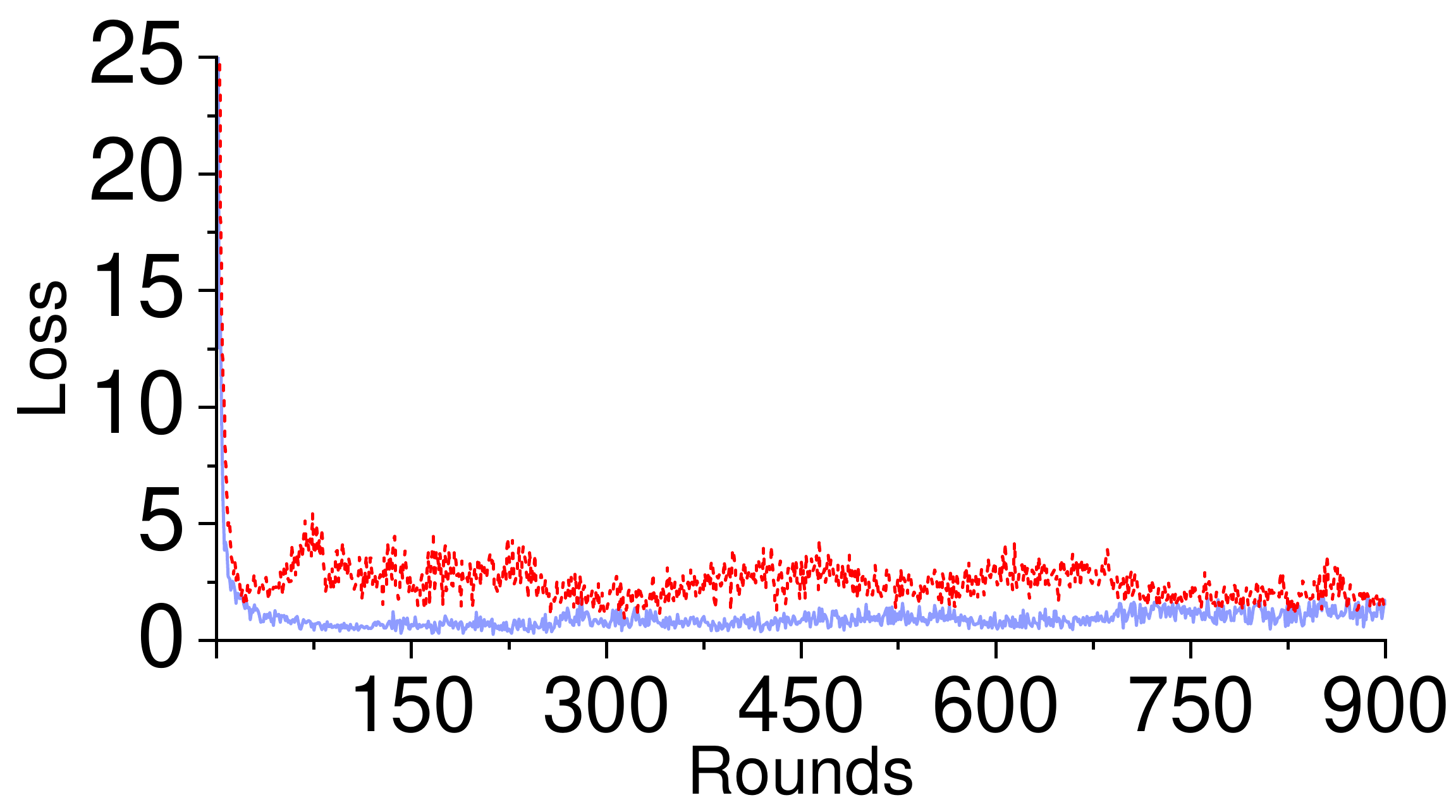}
        }
        \subcaptionbox{Filtering capacity\label{rl-reduction}}{
            \vspace{-0.2cm}
            \includegraphics[width=.35\columnwidth]{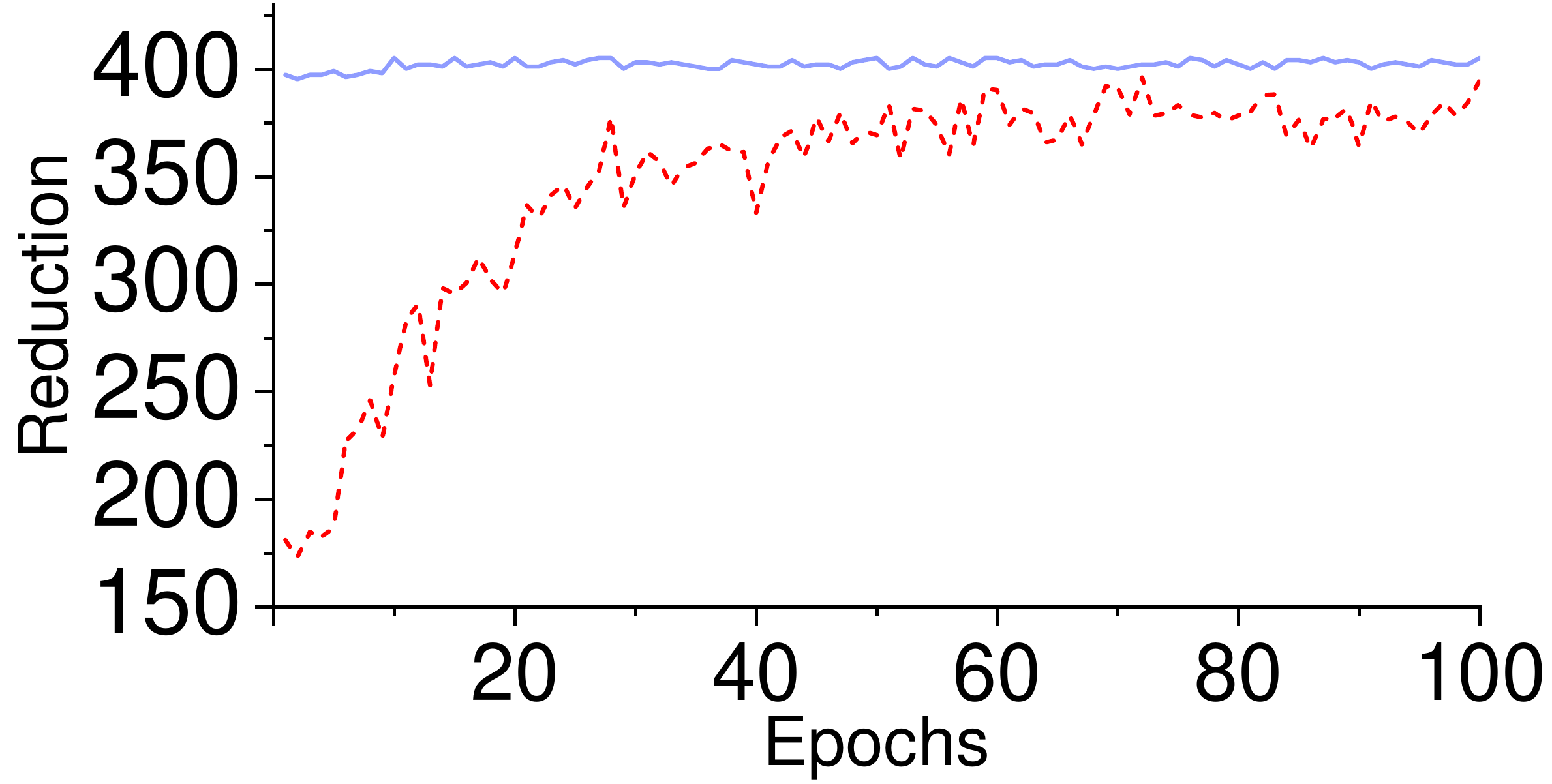}
        }
    \end{minipage}
    \vspace{-0.3cm}
    \caption{Comparison of model convergence and reward}
    \label{fig:abl_mask}
\end{figure}
there are some fluctuations in the results as the RL agent learns from the feedback through trial-and-error interactions with the environment. RL balances the trade-off between exploration and exploitation, which reduces the fluctuation with more training epochs. These results confirm that the action mask helps to decrease the number of training epochs, leading to less training time.

\subsection{Parameter Sensitivity Study}
\textcolor{edit}{We further evaluate the sensitivity of \idxname's training  and query times to our key parameters. We show the results of varying the numbers of hidden units and layers in the neural networks in Figure \ref{hidden-unit}. Using more hidden units significantly increases the training time but only slightly improves the query time. Increasing the number of hidden layers shows a similar effect. Additionally, increasing the size of the structure requires larger memory space. Thus, we set the default hidden units and layers to 16 and 2.}

\begin{figure}[htb]
    \begin{minipage}{.55\linewidth}
        \centering
        \includegraphics[width=.9\columnwidth]{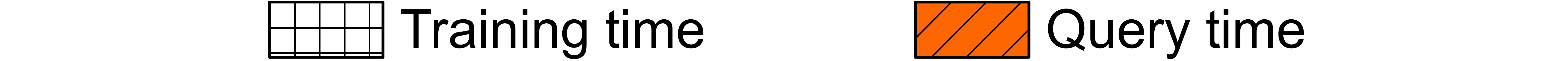}
    \end{minipage}
    \begin{minipage}{\linewidth}
        \centering
        \subcaptionbox{Varying no. of hidden units\label{hidden-unit}}{
            \vspace{-0.2cm}
            \includegraphics[width=.35\columnwidth]{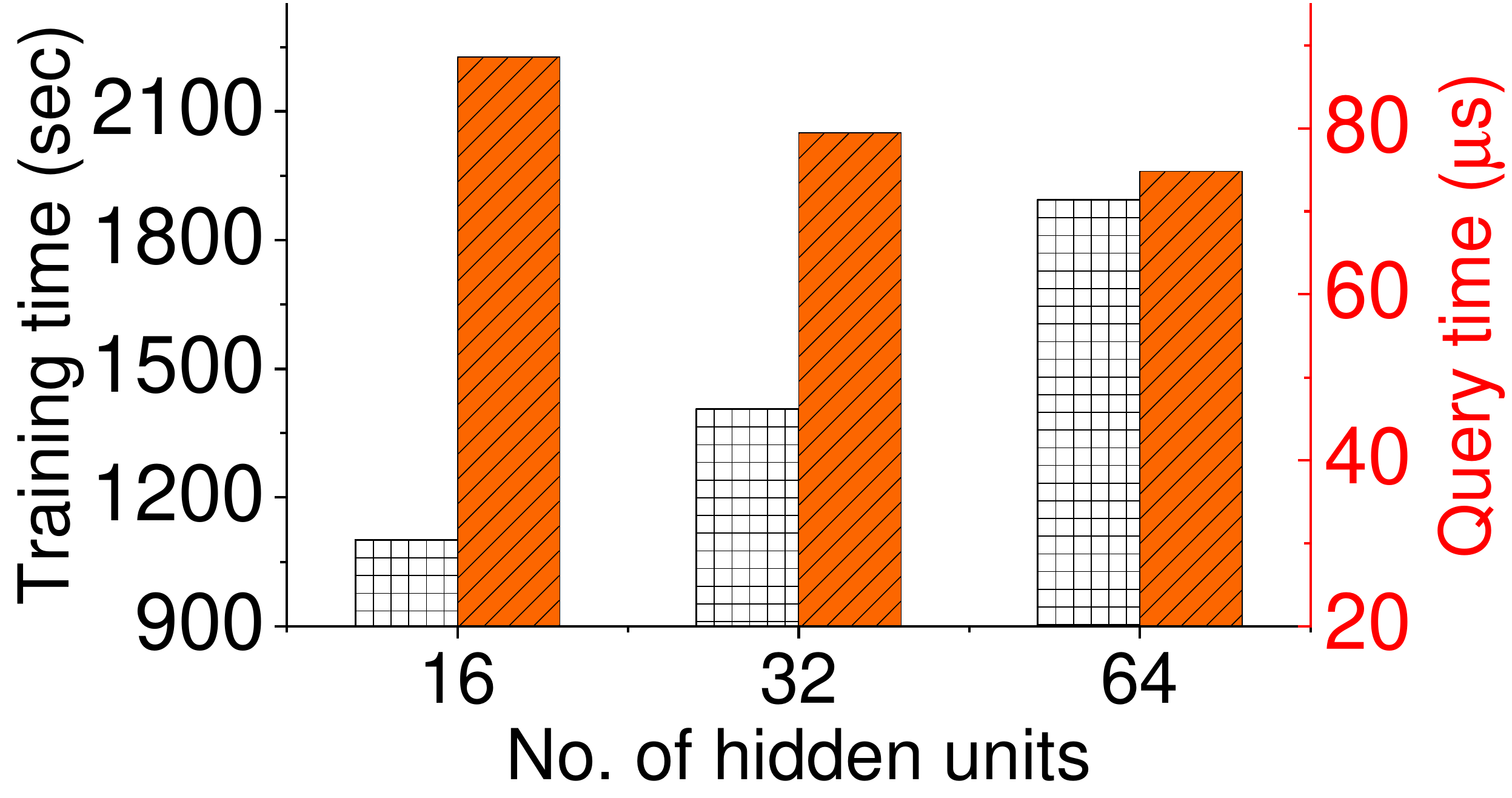}
        }
        \subcaptionbox{Varying no. of hidden layers\label{hidden-layers}}{
            \vspace{-0.2cm}
            \includegraphics[width=.35\columnwidth]{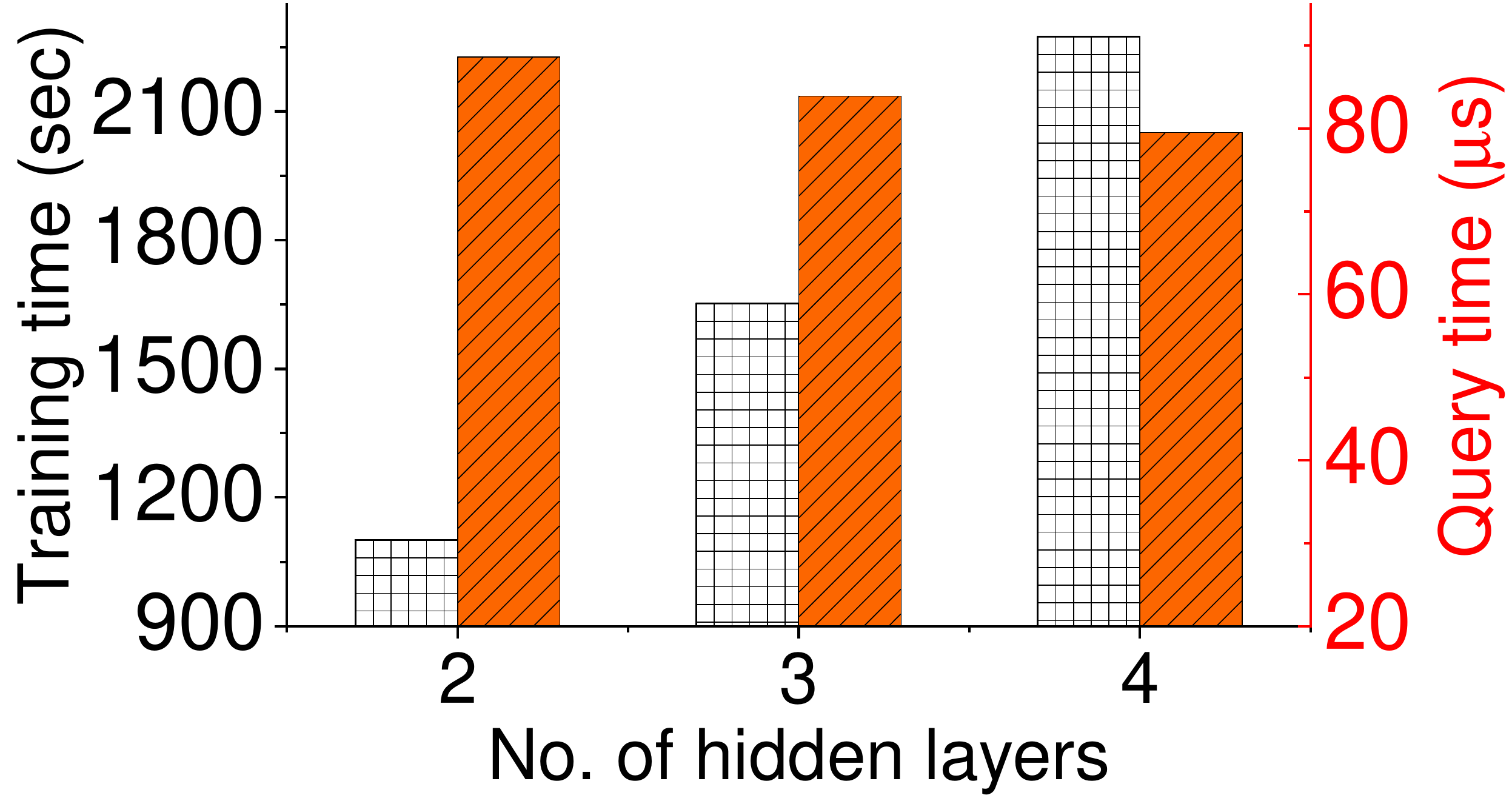}
        }
    \end{minipage}
    \vspace{-0.3cm}
    \caption{Comparison of training time and query time}
    \label{fig:hidden}
\end{figure}
\textcolor{edit}{We also vary the capacity of experience replay and the discount factor. The results show that they have minor impacts on the query time and the RL convergence rate. We omit the details due to space limits. The performance of \idxname\ is less sensitive to these hyper-parameters, and we can follow the existing work \cite{DBLP:journals/nature/MnihKSRVBGRFOPB15} to set them.}

\section{Related Work}
\label{sec8}
\textbf{Traditional geo-textual indexes.} In recent years, due to the popularity of SKR queries and their applicability in practical scenarios, a series of geo-textual indexes \cite{DBLP:conf/cikm/ChristoforakiHDMS11,DBLP:conf/sigmod/ChenSM06, DBLP:conf/ssdbm/HariharanHLM07, DBLP:conf/cikm/ZhouXWGM05, DBLP:conf/ssd/VaidJJS05, DBLP:conf/ssd/TampakisSDPKV21, DBLP:conf/dexa/KhodaeiSL10} have been proposed to support efficient SKR query processing. 

The general idea of geo-textual indexes is to combine spatial and textual indexes to exploit their pruning capabilities based on the spatial and textual attributes of the data. Early works only loosely combine both types of indexes. For example, 
Vaid et al. \cite{DBLP:conf/ssd/VaidJJS05} propose the first grid-based geo-textual indexes, i.e., the spatial primary index (ST) and the text primary index (TS), which are spatial-first and textual-first integrations, respectively. 
Parallel to this work, the R*-tree-inverted file (R*-IF) and the inverted file-R*-tree (IF-R*) \cite{DBLP:conf/cikm/ZhouXWGM05} combine the inverted file with the R*-tree \cite{DBLP:conf/sigmod/BeckmannKSS90}.
The loose integrated indexes has been shown to result in unsatisfactory query time in both the follow-up study \cite{DBLP:journals/pvldb/ChenCJW13} and our experiments.

Later works combine spatial and textual indexes more tightly such that they use both types of attributes of the data for search space pruning in parallel. For example, 
each grid cell in Spatial-Keyword Inverted File (SKIF) \cite{DBLP:conf/dexa/KhodaeiSL10} is presented by an inverted list, and it works in the rectangular object scenario. Hariharan et al. \cite{DBLP:conf/ssdbm/HariharanHLM07} propose the Keyword-R*-tree (KR*-tree). Each node of this index is associated with the set of keywords that appear in the sub-tree rooted at this node. Thus, each tree node can prune the search space with both a spatial region and a set of keywords at the same time.

The indexes above focused on SKR queries. 
There are also geo-textual indexes \cite{DBLP:journals/pvldb/CongJW09, DBLP:conf/cikm/Hoang-VuVF16, DBLP:conf/edbt/ZhangTT13, DBLP:conf/icde/FelipeHR08, DBLP:conf/ssd/RochaGJN11, DBLP:journals/tkde/WuYCJ12}, such as IR-Tree, developed for other types of spatial keyword queries. Some of these can be adapted to answer SKR queries. However, since they are not tailored for SKR queries, their query time is usually worse \cite{DBLP:journals/pvldb/ChenCJW13}.

\noindent \textbf{Learned indexes.} Kraska et al. \cite{DBLP:conf/sigmod/KraskaBCDP18} propose the recursive model index (RMI), which leverages machine learning models to replace a traditional index over one-dimensional search keys. The motivation is that an index can be seen as a function mapping a search key to the storage position of the corresponding record. Several follow-up studies propose learned indexes for one-dimensional  data~\cite{DBLP:conf/sigmod/DingMYWDLZCGKLK20, DBLP:journals/pvldb/WuZCCWX21, DBLP:journals/pvldb/FerraginaV20}. More details can be found in a benchmark study \cite{DBLP:journals/pvldb/MarcusKRSMK0K20}.

To handle multi-dimensional data, the Z-order model \cite{DBLP:conf/mdm/WangFX019} extends RMI by utilizing a Z-order curve to map multi-dimensional search keys into one-dimensional keys. Since this index might lead to large and uneven gaps between the mapped keys of adjacent objects, RSMI \cite{DBLP:journals/pvldb/QiLJK20} proposes a rank space-based technique, and it further proposes a hierarchical learned partitioning strategy for index learning over large spatial datasets. A parallel work, LISA~\cite{DBLP:conf/sigmod/Li0ZY020}, designs a grid-based index that supports data updates.
The RLR-tree~\cite{DBLP:journals/corr/abs-2103-04541} 
uses machine learning techniques to build a better R-tree without the need to change the structure or query
processing algorithms of the R-tree. 

Although these indexes have shown performance gains by exploiting the data distribution, they have ignored the query workload in index construction. Several studies \cite{DBLP:conf/sigmod/NathanDAK20, DBLP:journals/pvldb/DingNAK20, DBLP:conf/sigmod/MaAHMPG18} take into account the query workload and propose to automatically optimize the index structure for a given data and query distribution. 


\noindent \textbf{Reinforcement learning.} RL is often utilized in sequence generation applications, such as game playing \cite{DBLP:journals/nature/SilverHMGSDSAPL16}, machine translation \cite{kang2020dynamic}, and bin packing \cite{DBLP:conf/aaai/ZhaoS0Y021}. Recently, it  has been adapted to solve database optimization problems, such as query optimization~\cite{DBLP:conf/sigmod/ZhangCZ022}, index tuning~\cite{DBLP:conf/sigmod/00010SWNCB22}, and trajectory simplification \cite{DBLP:conf/icde/WangLC21}. However, these problems are quite different from ours by definitions, and so are their state and reward formulations. Thus, our RL formulation requires new designs in the states and reward functions. 

\section{Conclusions and Future Work}
\label{sec9}
We proposed a hierarchical index named \idxname\ for SKR queries, which is jointly optimized for a given dataset and a query workload. \idxname\ is built in two stages.  First, a partitioning algorithm finds the data clusters that minimize the time cost of executing the query workload. Then, an RL-based algorithm packs the data clusters into a hierarchical index in a bottom-up manner for more efficient pruning at query time. 
Learning from the query workload enables \idxname\ to significantly outperform traditional SKR indexes. 
Experimental results on real-world datasets show that \idxname\ yields strong query performance over various workloads, achieving up to 8$\times$ times speedups with comparable storage overhead.

There are several directions for future work. 
First, we intend to better answer Boolean $k$NN queries and to support more types of spatial keyword queries.
Second, even \idxname\ can adapt to workload changes by model retraining, a data-driven learned index, which is compatible with frequent workload shifts and data changes, is an interesting direction to explore.



\bibliographystyle{ACM-Reference-Format}
\bibliography{ref-scholar}

\appendix
\section{\textit{k}NN Query Support}
\textcolor{edit}{\idxname can also support the Boolean \textit{k}NN (\textbf{B\textit{k}}) query without any modification to the index layout. A B\textit{k} query on geo-textual data is formed by a set of keywords $\psi$, a spatial query point \textit{o}, and the result size \textit{k}. It aims to retrieve \textit{k} objects, each of which covers at least one keyword in $\psi$ and is top-\textit{k} closest to \textit{o}. To process B$k$  queries, we follow existing works~\cite{DBLP:journals/tkde/WuYCJ12, DBLP:journals/pvldb/CongJW09} by using a best-first search. Here, we compare \idxname with two SOTA indexes, WBIR-Tree \cite{DBLP:journals/tkde/WuYCJ12} and LSTI, and we use the index layout generated under the default setting. As shown in Figure \ref{knn-key}, the query times of WBIR-Tree and our \idxname grow with the number of query keywords, which is expected. 
LSTI shows an opposite trend, as it needs to scan more spline points with fewer keywords. We note that \idxname shows comparable performance with the best baseline results.  Figure \ref{knn-top} further shows the query times when the result size $k$ is varied.  \idxname and WBIR-Tree show stable performance as $k$ increases, while LSTI degrades rapidly. \idxname achieves the best performance when $k > 15$.}

\begin{figure}[htb]
    \begin{minipage}{.55\linewidth}
        \centering
        \includegraphics[width=\linewidth]{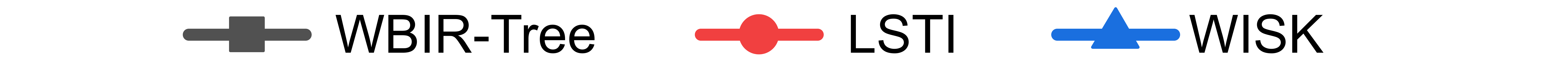}
    \end{minipage}
    \begin{minipage}{\linewidth}
        \centering
        \subcaptionbox{Varying no. of query keywords\label{knn-key}}{
            \vspace{-0.2cm}
            \includegraphics[width=.35\columnwidth]{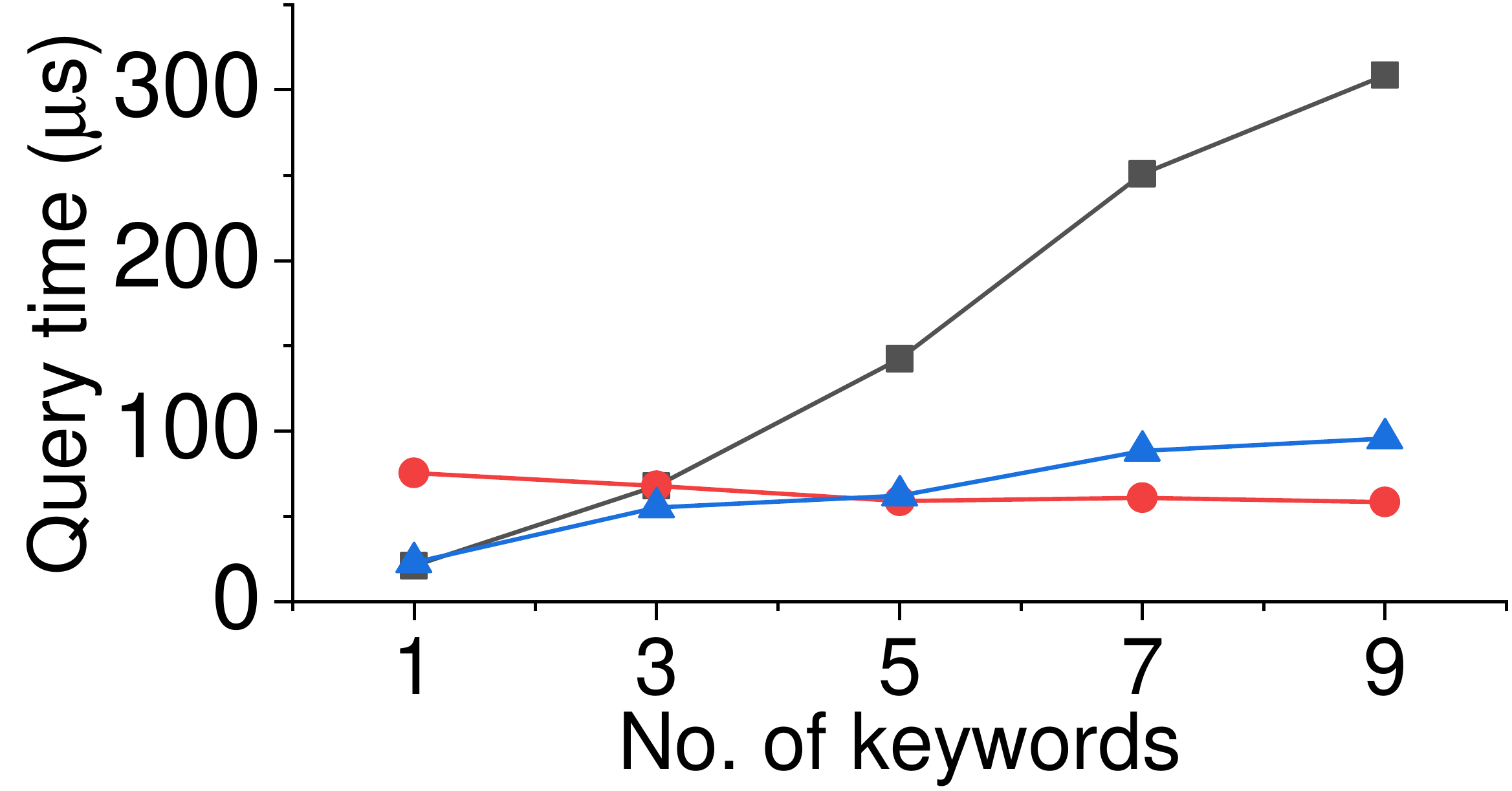}
        }
        \subcaptionbox{Varying $k$\label{knn-top}}{
            \vspace{-0.2cm}
            \includegraphics[width=.35\columnwidth]{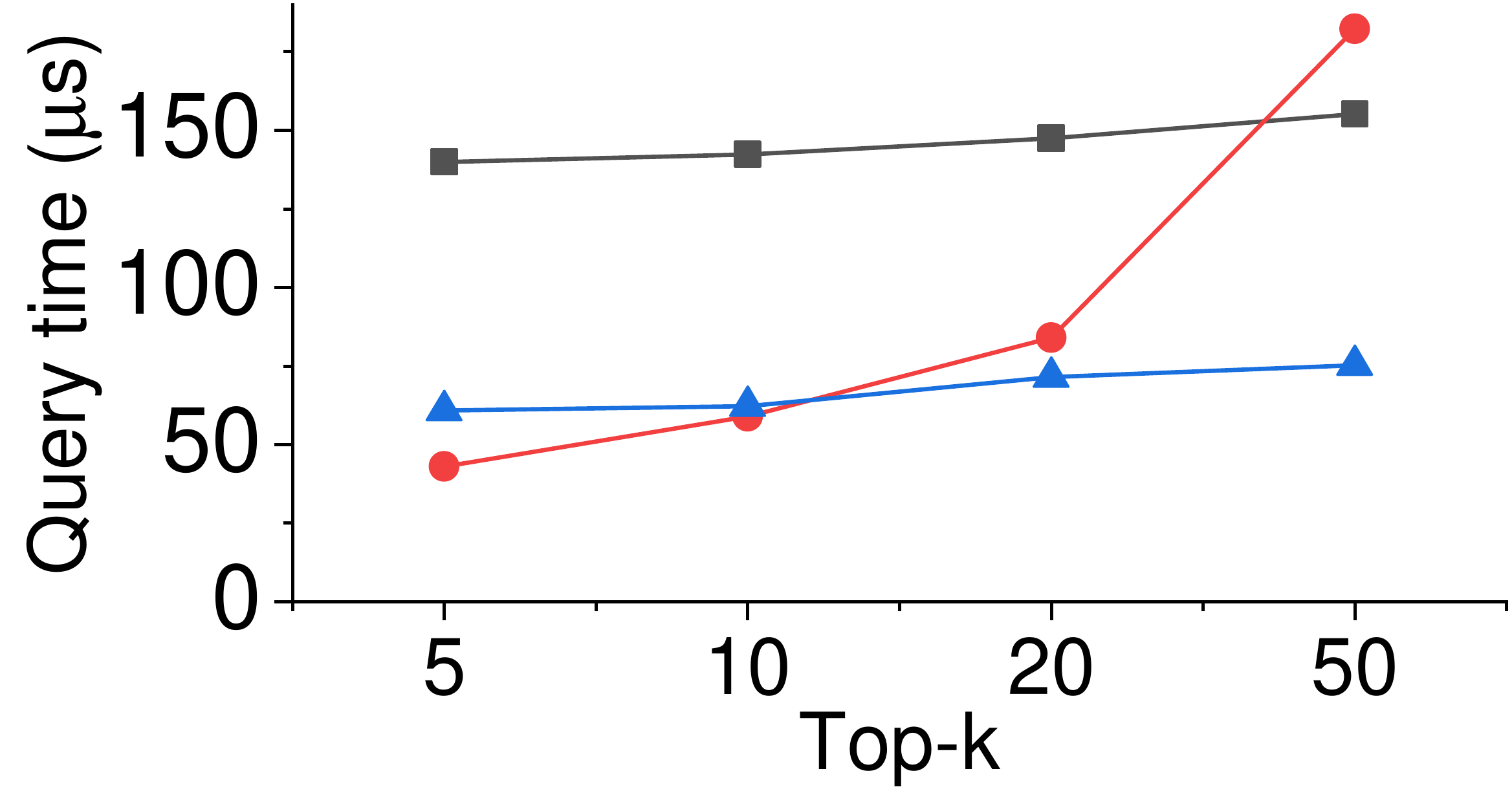}
        }
    \end{minipage}
    \vspace{-0.3cm}
    \caption{\textit{k}NN query time}
    \label{fig:knn-query}
\end{figure}
\textcolor{edit}{Note that these experiments only aim to show that \idxname is applicable to the B\textit{k} query as well. Our work mainly focuses on SKR queries. Optimizing \idxname or designing an optimized learned index for other types of queries remains our future work.}
\end{document}